\documentclass[10pt]{article}

\usepackage{araim-common}
\usepackage{araim-tex}

\makeatletter
\def\blfootnote{\xdef\@thefnmark{}\@footnotetext}
\makeatother

\title{An Extension of Generalized Linear Models\\to Finite Mixture Outcome Distributions}
\date{}
\author{Andrew M. Raim$^*$, Nagaraj~K. Neerchal$^\dagger$ \& Jorge~G. Morel$^\dagger$
\vspace{0.5em} \\
$^*$Center for Statistical Research and Methodology, U.S. Census Bureau \\
$^\dagger$Department of Mathematics and Statistics, University of Maryland, Baltimore County
}

\begin{document}

\maketitle

\begin{abstract}
Finite mixture distributions arise in sampling a heterogeneous population. Data drawn from such a population will exhibit extra 
variability relative to any single subpopulation. Statistical models based on finite mixtures can assist in the analysis of categorical and count outcomes when standard generalized linear models (GLMs) cannot adequately account for variability observed in the data. We propose an extension of GLM where the response is assumed to follow a finite mixture distribution, while the regression of interest is linked to the mixture's mean. This approach may be preferred over a finite mixture of regressions when the population mean is the quantity of interest; here, only a single regression function must be specified and interpreted in the analysis. A technical challenge is that the mean of a finite mixture is a composite parameter which does not appear explicitly in the density. The proposed model is completely likelihood-based and maintains the link to the regression through a certain random effects structure. We consider typical GLM cases where means are either real-valued, constrained to be positive, or constrained to be on the unit interval. The resulting model is applied to two example datasets through a Bayesian analysis: one with success/failure outcomes and one with count outcomes. Supporting the extra variation is seen to improve residual plots and to appropriately widen prediction intervals.
\end{abstract}

\blfootnote{%
This paper is released to inform interested parties of ongoing research and to encourage discussion of work in progress. Any views expressed are those of the authors and not necessarily those of the U.S. Census Bureau.
\begin{flushleft}
For correspondence: \\
A.M.~Raim (\url{andrew.raim@gmail.com}) \\
Center for Statistical Research and Methodology \\
U.S. Census Bureau \\
Washington, D.C. 20233, U.S.A.
\end{flushleft}
}

\section{Introduction}
\label{sec:intro}
The Generalized Linear Model (GLM) is heavily used by researchers and practitioners for regression analysis on categorical, count, and continuous outcomes \citep{McCullaghNelder1989}. Standard GLM theory assumes an exponential family distribution, such as Poisson to model counts and Binomial to model success/failure data. These distributions are limited in the amount of variability they can express. GLM users often encounter the issue of overdispersion, where the data exhibit variability which cannot be expressed by the model. This can manifest itself in a number of ways, depending on the specific nature of the overdispersion and its departure from the model. For example, assuming independence in clustered data can result in standard error estimates which are too small and lead to tests with an inflated type I error rate \citep[Chapter 1]{MorelNeerchalSAS2012}.

The objective of this paper is to extend the GLM so that a finite mixture of $J$ simpler densities can be used as the distribution for the response. There is a well-established literature on finite mixtures of regressions, in which each component distribution of a finite mixture is linked to a separate regression \citep{FruhwirthSchnatter2006}. An analyst may employ a finite mixture of regressions model if heterogeneity is suspected in the relationship between covariate $\vec{x}$ and response $y$ among sampled units, yet not enough is known to model the heterogeneity explicitly. Specifying regressions for $J$ latent subpopulations may complicate model selection in practice. Often, the interest may be in modeling the mean response, and heterogeneity is simply a nuisance rather than a target for inference. This motivates us to formulate the Mixture Link model, which uses a finite mixture to capture extra variation, but constrains the mean of the finite mixture to be linked to a single regression function. The mean of a finite mixture is composed of multiple parameters which may not appear directly in the likelihood. Central to the development of Mixture Link is the set in which the link constraint is honored. In the case of positive-valued means, this constraint set is a polytope, while for probability-valued means it is the intersection of a polyhedron and a unit cube. For real-valued means, the constraint set is the basis of a linear space. A random effects structure is assumed on this set to complete specification of the likelihood. Under Poisson and Normal outcome types, the random effects can be integrated out to yield a tractable form for the density. The case of Binomial outcomes is more computationally challenging. Taking a Bayesian approach to inference, a simple Random-Walk Metropolis-Hastings sampler can be used for the Normal and Poisson Mixture Link models. For Binomial outcomes, we consider a Metropolis-within-Gibbs sampler with data augmentation to avoid repeated evaluation of the marginal density.

A number of methods have been established to handle overdispersion. \citet{MorelNeerchalSAS2012} provide an overview in the settings of count and categorical data. One common approach is to extend a basic distribution by assuming the presence of latent random variables, and then integrating them out. The Beta-Binomial~\citep{OtakePrentice1984}, Zero-Inflated Binomial \citep{Hall2000}, and Random-Clumped Binomial~\citep{MorelNagaraj1993} distributions are all obtained in this way starting from the Binomial distribution. Similarly, the negative Binomial and zero-inflated negative Binomial distributions \citep{Hilbe2011} are obtained starting from the Poisson distribution. In this same way, the t-distribution \citep{LiuRubin1995} may be considered an overdispersion model relative to the normal distribution. Generalized Linear Mixed Models are obtained by adding random effects to the regression function~\citep{McCullochSearleNeuhaus2008}; the marginal likelihood of the outcomes usually cannot be written without an integral for non-normal outcomes. Quasi-likelihood methods extend the likelihood in ways that do not yield a proper likelihood, but allow inference to be made on regression coefficients. A simple quasi-likelihood is obtained from placing a dispersion multiplier to the variance \citep[Section~4.7]{Agresti2002}. The method of \citet{Wedderburn1974} requires specification of only the mean-variance relationship to form a system of equations and carry out inference. Generalized Estimating Equations (GEE) is a quasi-likelihood method for grouped data where the analyst assumes a working correlation structure for observations taken within a subject~\citep{HardinHilbe2012}. Some Bayesian overdispersion methods are discussed in the collection assembled by \citet{DeyEtAl2000}; for example, \citet*{BasuMukhopadhyayBookChapter2000} consider generalizing the link function of a GLM to a mixture distribution and \citet*{DeyRavishanker2000} propose generalized exponential families for the outcome. More recently, \citet{KleinKneibLang2015} proposed a Bayesian approach to generalized additive models under the Zero-Inflated Negative Binomial model to estimate complicated regression functions.

The rest of the paper proceeds as follows. Section~\ref{sec:mixlink-model} formulates the Mixture Link general model. Section~\ref{sec:mixlink-probs} develops Mixture Link under probability-valued means, with special attention given to Binomial outcomes. Sections \ref{sec:mixlink-positive} and \ref{sec:mixlink-real} develop Mixture Link for positive- and real-valued means, respectively, and obtain specific models for Poisson and Normal outcomes. Section~\ref{sec:data-examples} presents example data analyses with Mixture Link Binomial and Mixture Link Poisson. Finally, Section~\ref{sec:conclusions} concludes the paper. The \code{mixlink} package for \code{R} (available from \url{http://cran.r-project.org}) provides much of the Mixture Link functionality discussed in this paper.

\section{Mixture Link Formulation}
\label{sec:mixlink-model}
The usual GLM formulation is based on a density in the exponential dispersion family,
\begin{align}
f(y \mid \theta, \phi) = \exp\left\{
\frac{\theta y - b(\theta)}{a(\phi)} + c(y; \phi)
\right\},
\label{eqn:exp-fam}
\end{align}
where $\theta$ is the canonical parameter which influences the mean and $\phi$ is the dispersion parameter. Here it can be shown that $\E(y) = b'(\theta)$ and $\Var(y) = a(\phi) b''(\theta)$, and expressions for the score vector and information matrix can be obtained \citep[Section 4.4]{Agresti2002}. Estimation can be carried out routinely, using Newton-Raphson or scoring algorithms to compute maximum likelihood estimates, or standard MCMC algorithms for a Bayesian analysis. Our objective is to modify this framework to allow a finite mixture as the outcome distribution, establishing a link between the mixture mean and a regression function of interest. Because finite mixtures can support more variation than distributions of the form \eqref{eqn:exp-fam}, this extension should naturally support variation beyond standard GLMs. We are especially interested in finite mixtures of three common GLM outcome types: Normal, Binomial, and Poisson.

Consider a random variable $Y$ following the finite mixture distribution,
\begin{align}
f(y \mid \btheta) = \sum_{j=1}^J \pi_j g(y \mid \btheta_j).
\label{eqn:finite-mixture}
\end{align}
Here, the mixing proportions $\vec{\pi} = (\pi_1, \ldots, \pi_J)$ belong to the probability simplex $\mathcal{S}^J = \{ \vec{\lambda} \in [0,1]^J : \lambda_j \geq 0, \vec{\lambda}^T \vec{1} = 1 \}$. The densities $g(y \mid \btheta_j)$ belong to a common family parameterized by $\btheta_j = (\mu_j, \vec{\phi}_j)$, consisting of a mean parameter $\mu_j = \int y \, g(y \mid \btheta_j) d\nu(y)$ and where all other parameters are contained in $\vec{\phi}_j$. Writing $\nu$ as the dominating measure for densities $g$ allows expectations over discrete and continuous random variables to be treated with a common integral notation. The overall expected value is $\E(Y) = \sum_{j=1}^J \pi_j \mu_j = \vec{\pi}^T \vec{\mu}$. The $\mu_j$ may naturally be restricted to a subset of $\mathbb{R}$, depending on the outcome type. For example, if $Y$ is a count, $\mu_j \in [0, \infty)$ often represents a rate. Alternatively, if $Y$ is the number of successes among $m$ trials, which result in either success or failure, then $\mu_j \in [0, 1]$ can represent the probability of a success. In general, denote the natural space of $\mu_j$ as $\mathcal{M}$, so that $\vec{\mu} = (\mu_1, \ldots, \mu_J)$ is an element of $\mathcal{M}^J$.

In a regression setting, we observe a random sample $Y_1, \ldots, Y_n$ from the finite mixture
\begin{align}
f(y_i \mid \btheta_i) = \sum_{j=1}^J \pi_j g(y \mid \mu_{ij}, \vec{\phi}_{ij}),
\label{eqn:finite-mixture-obs}
\end{align}
with an associated (fixed) predictor $\vec{x}_i \in \mathbb{R}^d$, for $i \in \{ 1, \ldots, n \}$. As in the traditional GLM, we wish to link $\E(Y_i)$ to a regression function such as $\vec{x}_i^T \bbeta$ through an inverse link function $G$. To simplify expressions in the rest of the paper, denote $\vartheta(\vec{x})$ as the inverse-linked regression $G(\vec{x}^T \bbeta)$. We will write $\vartheta_i = G(\vec{x}_i^T \bbeta)$ for brevity when specifically referring to the $i$th observation, and $\vartheta$ in place of $\vartheta(\vec{x})$ when not emphasizing a specific observation. With this notation, our objective is to link
\begin{align}
\vec{\pi}^T \vec{\mu} = \vartheta_i.
\label{eqn:link-objective}
\end{align}
The left-hand side of \eqref{eqn:link-objective} must vary with the observation for the link to be achievable. In this work, we will assume that subpopulation means $\vec{\mu}_i = (\mu_{i1}, \ldots, \mu_{iJ})$ are specific to the $i$th observation, but that mixing proportions $\vec{\pi}$ are common across observations. In contrast to the traditional GLM setting, $\vec{\pi}^T \vec{\mu}_i$ is a composite parameter which does not appear directly in the density of $Y_i$. Therefore, we cannot simply plug $\vartheta_i$ into the likelihood.

To enforce \eqref{eqn:link-objective}, consider the set
\begin{align}
A(\vartheta, \vec{\pi}) = \{ \vec{\mu} \in \mathcal{M}^J : \vec{\mu}^T \vec{\pi} = \vartheta \}.
\label{eqn:link-constrant-set}
\end{align}
For a given $\bbeta$ and $\vec{\pi}$, restricting ourselves to $\vec{\mu}_i \in A(\vartheta_i, \vec{\pi})$ is equivalent to enforcing the link. We will write $A$ as a shorthand for $A(\vartheta, \vec{\pi})$ and $A_i$ for $A(\vartheta_i, \vec{\pi})$. Our approach will be to take $\vec{\mu}_i$ as a random effect drawn from set $A(\vartheta_i, \vec{\pi})$. In Sections \ref{sec:mixlink-probs}, \ref{sec:mixlink-positive}, and \ref{sec:mixlink-real} we will consider several commonly used choices of the space $\mathcal{M}$---the unit interval, the positive real line, and the real line respectively---to determine an appropriate distribution for $\vec{\mu}_i$. Figure~\ref{fig:vertices} displays an example of the set $A(\vartheta_i, \vec{\pi})$ for each of these three cases. \citet{BoydVandenberghe2004} is a useful reference for basic concepts in the analysis of convex sets which emerge in the remainder of the paper. Note that $\vec{x}_i = 1$ may be taken for all $i = 1, \ldots, n$ to yield a non-regression version of Mixture Link.

Selection of a distribution over $A(\vartheta, \vec{\pi})$ determines the density of $Y_i$,
\begin{align}
f(y_i \mid \bbeta, \vec{\pi}, \vec{\phi}_i)
&= \int \sum_{j=1}^J \pi_j g(y_i \mid \mu_{ij}, \vec{\phi}_{ij}) \cdot f_{A^{(i)}}(\vec{\mu}_i) d\vec{\mu}_i \nonumber \\
&= \sum_{j=1}^J \pi_j \int g(y_i \mid w, \vec{\phi}_{ij}) \cdot f_{A^{(i)}_j}(w) dw.
\label{eqn:mixture-link-density}
\end{align}
Here, $f_{A^{(i)}}$ represents the $J$-dimensional random effects density over $A(\vartheta_i, \vec{\pi})$ and $f_{A^{(i)}_j}$ represents the marginal density of the $j$th coordinate. In the trivial case $J=1$, there is only a single point in $A(\vartheta_i, \vec{\pi})$, and $f(y_i \mid \bbeta, \vec{\pi}, \vec{\phi}_i)$ simplifies to $g(y_i \mid \vartheta_i, \vec{\phi}_{i1})$. In general, evaluating $f(y_i \mid \bbeta, \vec{\pi}, \vec{\phi}_i)$ requires computation of $J$ univariate integrals, which can be achieved numerically using quadrature or other standard techniques. This can become a computational burden if $f(y_i \mid \bbeta, \vec{\pi}, \vec{\phi}_i)$ must be computed many times (e.g.~for a simulation or iterative estimation procedure) or if $f_{A^{(i)}_j}(w)$ is difficult to evaluate. By construction, $\E(Y_i) = \vartheta_i$, but variance and other moments depend on $g$ and the distribution of $\vec{\mu}_i$. As in more basic finite mixture models, the value of density \eqref{eqn:mixture-link-density} is invariant to permutations of the subpopulation labels $\{ 1, \ldots, J \}$.

\begin{figure}
\centering
\begin{subfigure}[b]{.32\textwidth}
  \centering
  \includegraphics[width=0.95\textwidth]{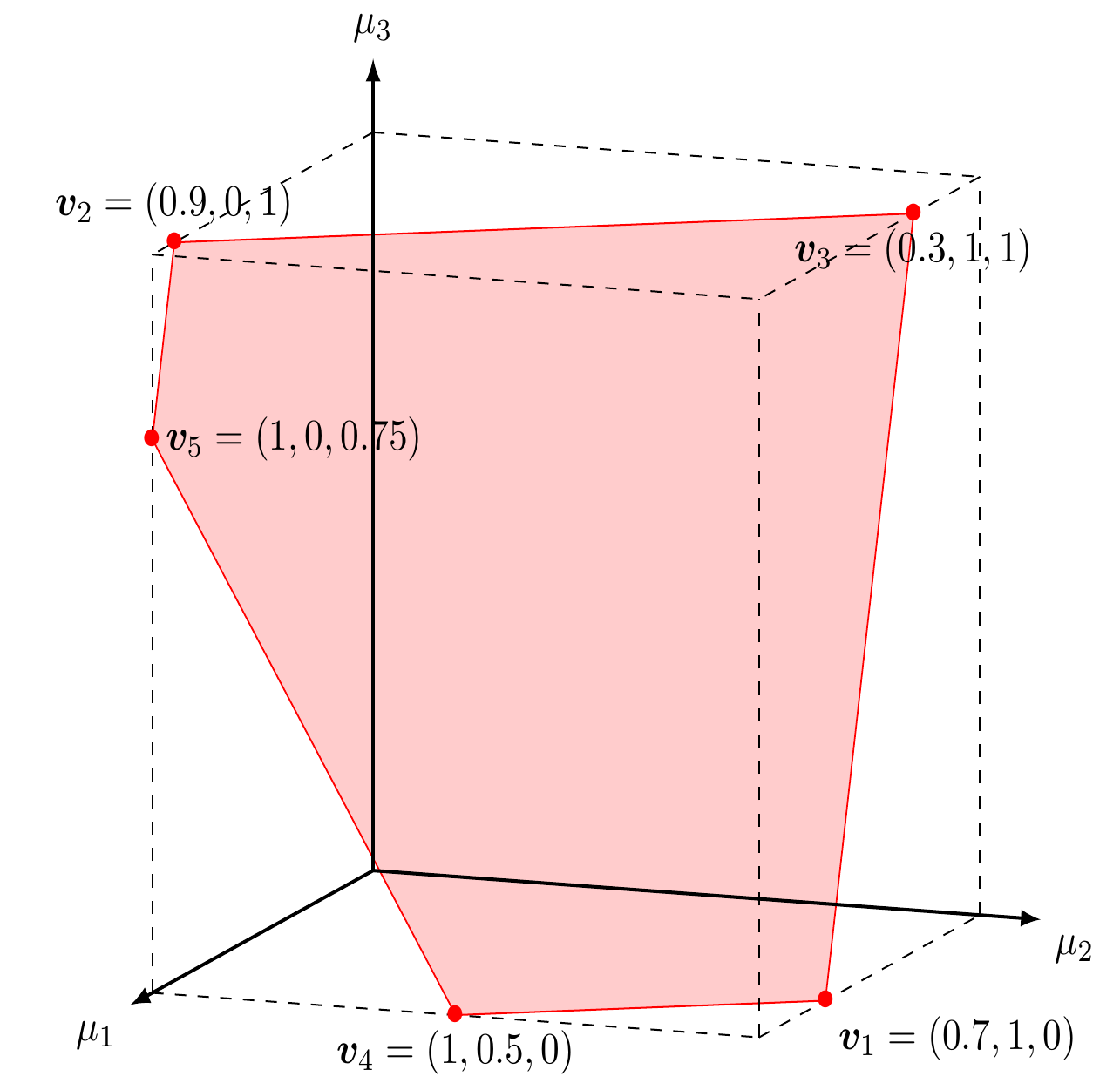}
  \caption{}
  \label{fig:vertices-prob}
\end{subfigure}
\begin{subfigure}[b]{.32\textwidth}
  \centering
  \includegraphics[width=0.95\textwidth]{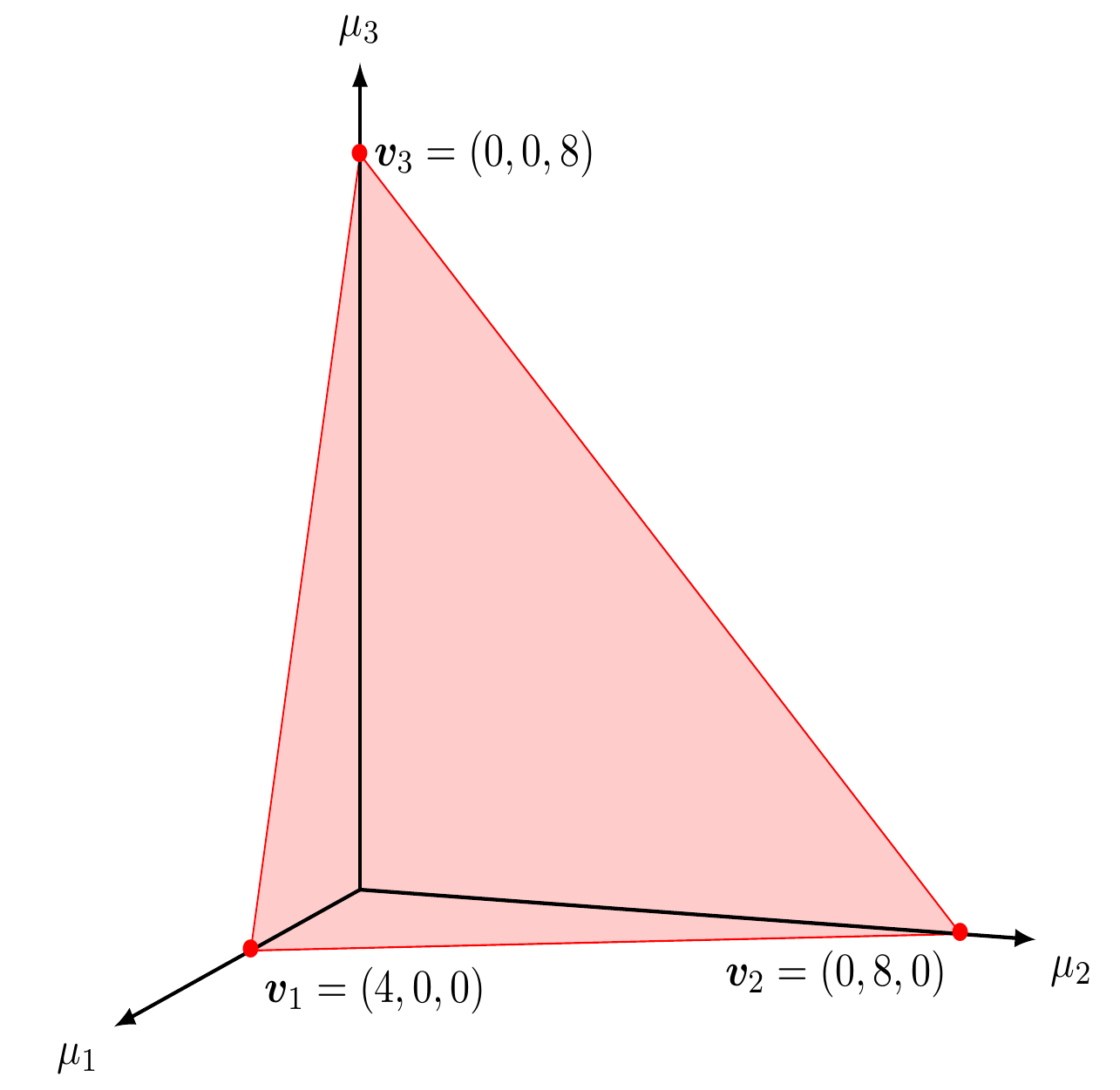}
  \caption{}
  \label{fig:vertices-pos}
\end{subfigure}
\begin{subfigure}[b]{.32\textwidth}
  \centering
  \includegraphics[width=0.95\textwidth]{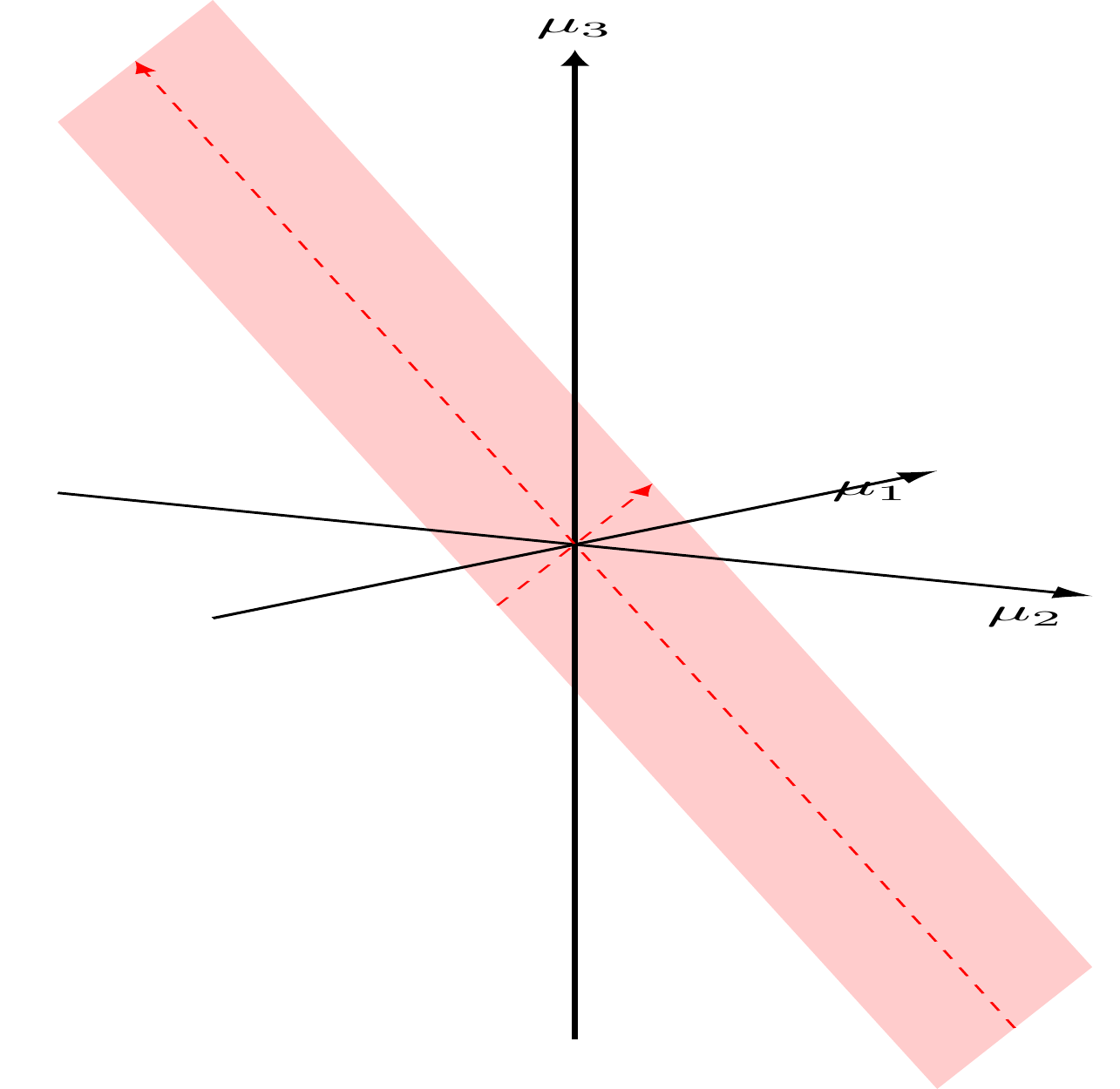}
  \caption{}
  \label{fig:vertices-real}
\end{subfigure}
\caption{Examples of the set $A(\vartheta, \vec{\pi})$ in dimension $J=3$: (\subref{fig:vertices-prob}) probability-valued means with $\vec{\pi} = (0.5, 0.3, 0.2)$ and $\vartheta = 0.65$, (\subref{fig:vertices-pos}) positive means with $\vec{\pi} = (0.5, 0.25, 0.25)$ and $\vartheta = 2$, (\subref{fig:vertices-real}) real-valued means with $\vec{\pi} = (0.5, 0.3, 0.2)$ and $\vartheta = 0$.}
\label{fig:vertices}
\end{figure}

\section{Probability-Valued Means}
\label{sec:mixlink-probs}
Consider the setting $\mathcal{M} = [0, 1]$, which is useful for Bernoulli or Binomial data where means represent probabilities. It is straightforward to verify that $A(\vartheta_i, \vec{\pi}) = \{ \vec{\mu} \in [0, 1]^J : \vec{\mu}^T \vec{\pi} = \vartheta_i \}$ is a bounded convex set in $\mathbb{R}^J$. Therefore, we have the decomposition
\begin{align}
A(\vartheta_i, \vec{\pi})
= \Big\{ \sum_{\ell=1}^{k_i} \lambda_\ell \vec{v}_\ell^{(i)} : \vec{\lambda} \in \mathcal{S}^{k_i} \Big\}
= \Big\{ \vec{V}^{(i)} \vec{\lambda} : \vec{\lambda} \in \mathcal{S}^{k_i} \Big\}.
\label{eqn:mink-weyl}
\end{align}
The $J \times k_i$ matrix $\vec{V}^{(i)}$ is composed of the columns $\vec{v}_1^{(i)}, \ldots, \vec{v}_{k_i}^{(i)}$ which are vertices of $A(\vartheta_i, \vec{\pi})$. Any element $\vec{\mu} \in A(\vartheta_i, \vec{\pi})$ can be written as a convex combination of these vertices. The matrix $\vec{V}^{(i)}$ depends on both $\vec{\pi}$ and $\vartheta_i$; both its elements and the dimension $k_i$ may vary with the observation $i = 1, \ldots, n$. The vector $\vec{\lambda}^{(i)}$ belongs to the probability simplex $\mathcal{S}^k$.

The Minkowski-Weyl decomposition of a polyhedron is
$
P = 
\{ \sum_{\ell=1}^k \lambda_\ell \vec{v}_\ell  : \vec{\lambda} \in \mathcal{S}^k \} +
\{ \sum_{\ell=1}^h \lambda_\ell \vec{\xi}_\ell  : \vec{\lambda} \geq 0 \},
$
relative to extreme points $\vec{v}_1, \ldots, \vec{v}_k$ (i.e.~vertices) and extreme directions $\vec{\xi}_1, \ldots, \vec{\xi}_h$ of $P$. The set $A_i$ in \eqref{eqn:mink-weyl} is a polytope, a bounded polyhedron not having extreme directions, for which we need only consider extreme points. Assuming a distribution on the coefficients of the Minkowski-Weyl decomposition has been advocated by \citet{DanaherEtAl2012}, who sought a class of priors to enforce biologically motivated polyhedral constraints in a Bayesian analysis.

A natural choice for a random effects distribution on $\mathcal{S}^{k_i}$ is $\vec{\lambda}^{(i)} \indep \text{Dirichlet}_{k_i}(\vec{\alpha})$. However, this choice leads to each component of $\vec{\mu}_i = \vec{V}^{(i)} \vec{\lambda}^{(i)}$ following the distribution of a linear combination of a $k$-dimensional Dirichlet. This distribution is computationally impractical; for example, its density has no known closed form for general $k$ \citep{ProvostCheong2000}. Our approach will first be to state the model using a Dirichlet random effect, then to state a more practical form of the model using Beta random effects with matched first and second moments. This ensures, for example, that $\E(\vec{\mu}_i) \in A(\vartheta_i, \vec{\pi})$. The Dirichlet formulation of the model is
\begin{align}
&Y_i \indep \sum_{j=1}^J \pi_j g(y_i \mid \mu_{ij}, \vec{\phi}_{ij}), \label{eqn:interval-hierarchy-dirichlet} \\
&\vec{\mu}_i = \vec{V}^{(i)} \vec{\lambda}^{(i)}, \quad
\text{where $\vec{V}^{(i)}$ contains vertices of $A(\vartheta_i, \vec{\pi})$},
\nonumber \\
&\vec{\lambda}^{(i)} \indep \text{Dirichlet}_{k_i}(\vec{\alpha}^{(i)}).
\nonumber
\end{align}
We restrict $\vec{\alpha}^{(i)}$ to the $k_i$-dimension vector $\kappa \vec{1}$ so that all $\vec{\lambda}^{(i)}$ follow a Symmetric Dirichlet distribution parameterized by a single scalar $\kappa$; this is done for several reasons. First, the dimension $k_i$ can vary with the observation so that an arbitrary $\vec{\alpha}$ would not be compatible with all observations. Second, the ordering of the vertices in $\vec{V}^{(i)}$ is somewhat arbitrary, and it is difficult to maintain a correspondence between individual vertices and the elements of $\vec{\alpha}$. Figure~\ref{fig:dirichletk3} plots the symmetric Dirichlet density for several $\kappa$ when $k=3$. Note that $\kappa = 1$ corresponds to the uniform distribution on the simplex, while $0 < \kappa < 1$ results in more density focused toward the vertices, and $\kappa > 1$ focuses density toward the interior.

\begin{figure}
\centering
\begin{subfigure}[b]{.40\textwidth}
  \centering
  \includegraphics[width=1.0\textwidth]{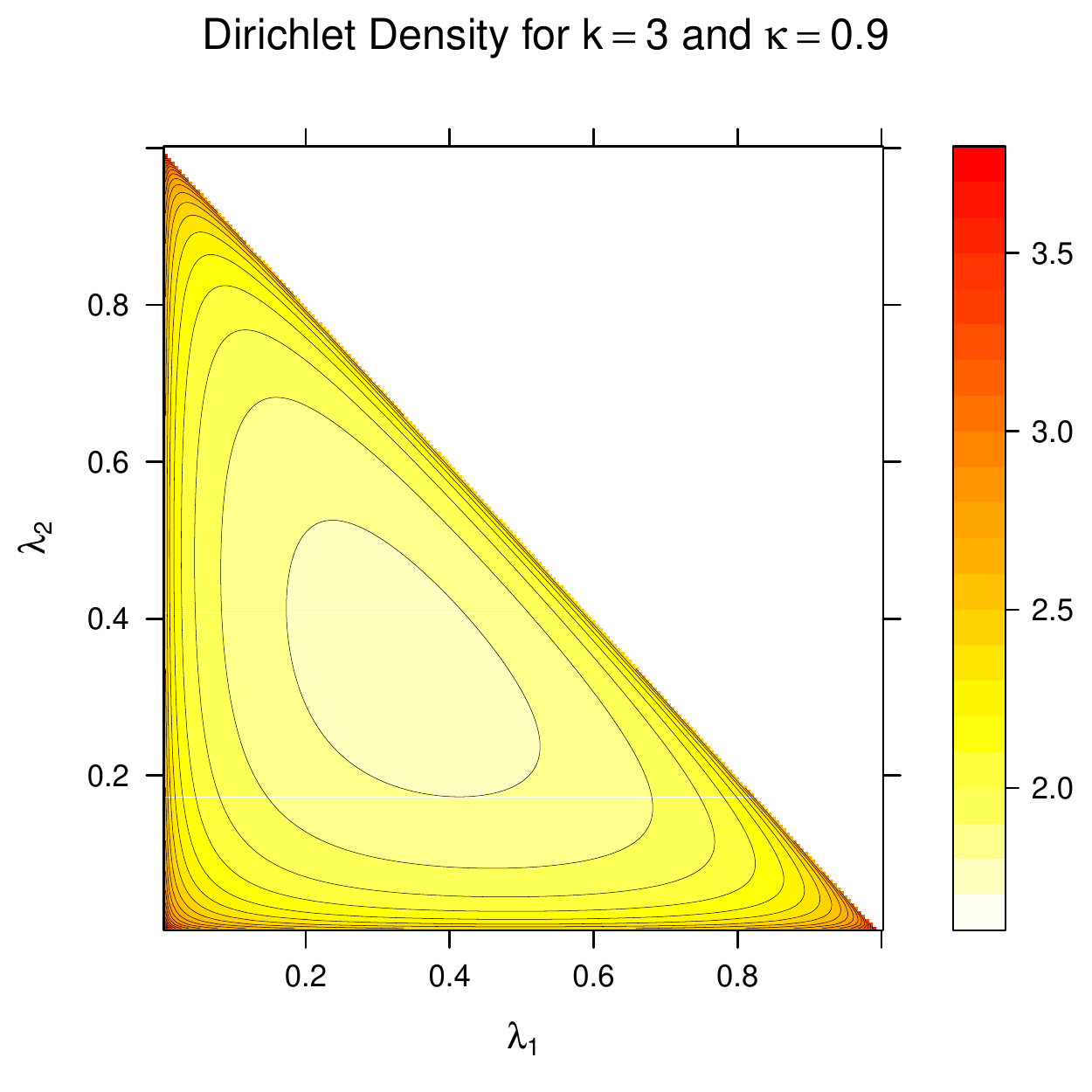}
  \caption{}
  \label{fig:dirichletk3-kappa0.9}
\end{subfigure}
\begin{subfigure}[b]{.40\textwidth}
  \centering
  \includegraphics[width=1.0\textwidth]{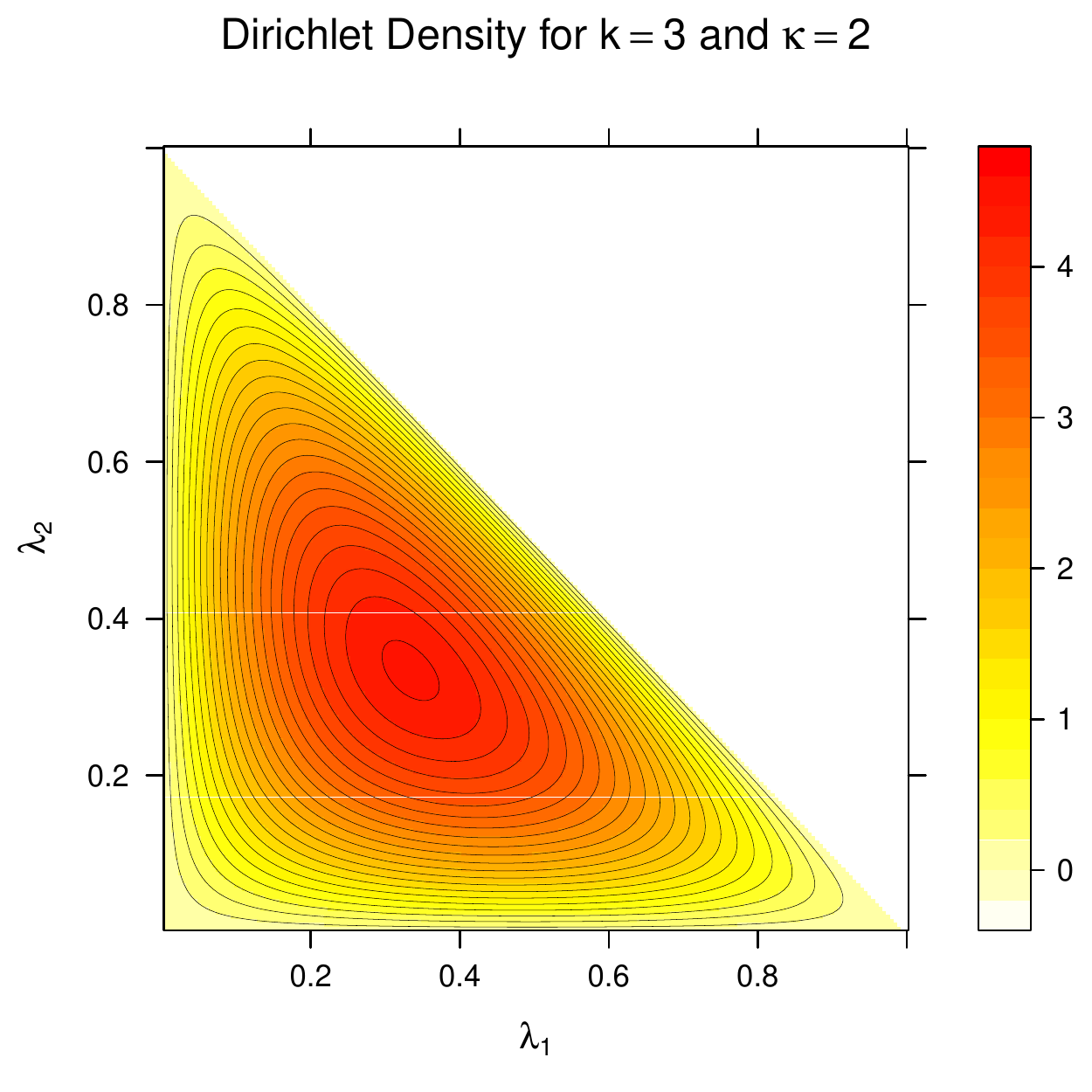}
  \caption{}
  \label{fig:dirichletk3-kappa2.0}
\end{subfigure}
\caption{The $\text{Dirichlet}_3(\vec{\lambda} \mid \kappa \vec{1})$ density for several settings of $\kappa$. Only $\lambda_1$ and $\lambda_2$ are plotted since $\lambda_3 = 1 - \lambda_1 - \lambda_2$.}
\label{fig:dirichletk3}
\end{figure}

Now, to obtain a Mixture Link density based on the more practical Beta distribution, define $\ell_{ij}$ and $u_{ij}$ as the smallest and largest elements respectively of the $j$th row $\vec{V}^{(i)}$; then $(\ell_{ij}, u_{ij})$ forms the support of $\mu_{ij}$. The Beta formulation of the model is
\begin{align}
&Y_i \indep \sum_{j=1}^J \pi_j g(y_i \mid \mu_{ij}, \vec{\phi}_{ij}), \label{eqn:interval-hierarchy-beta} \\
&\mu_{ij} = (u_{ij} - \ell_{ij})\psi_{ij} + \ell_{ij}, \quad \text{$j = 1, \ldots, J$}, \nonumber \\
&\psi_{ij} \sim \text{Beta}(a_{ij}, b_{ij}). \nonumber
\end{align}
To obtain $a_{ij}$ and $b_{ij}$, we first compute 
\begin{align*}
\E(\mu_{ij}) = (u_{ij} - \ell_{ij})\frac{a_{ij}}{a_{ij}+b_{ij}} + \ell_{ij},
\quad \text{and} \quad
\Var(\mu_{ij}) = \frac{ (u_{ij} - \ell_{ij})^2 a_{ij} b_{ij}}{(a_{ij}+b_{ij})^2 (a_{ij}+b_{ij}+1)}.
\end{align*}
Next, for $\vec{\lambda} \sim \text{Dirichlet}_{k_i}(\kappa \vec{1})$ and $\vec{v}_{j.}^{(i)T}$ denoting the $j$th row of $\vec{V}^{(i)}$, we can obtain
\begin{align*}
\E(\vec{v}_{j.}^{(i)T} \vec{\lambda}) = \bar{v}_{j.}^{(i)}
\quad \text{and} \quad
\Var(\vec{v}_{j.}^{(i)T} \vec{\lambda}) = \frac{\vec{v}_{j.}^{(i)T} \vec{v}_{j.}^{(i)} - k_i (\bar{v}_{j.}^{(i)})^2 }{ k_i (1 + k_i \kappa) },
\end{align*}
where $\bar{v}_{j.}^{(i)}$ denotes the mean of $\vec{v}_{j.}^{(i)T}$. Equating $\E(\mu_{ij})$ to $\E(\vec{v}_{j.}^{(i)T} \vec{\lambda})$ and $\Var(\mu_{ij})$ to  $\Var(\vec{v}_{j.}^{(i)T} \vec{\lambda})$ and solving for $a_{ij}$ and $b_{ij}$, we obtain that
\begin{align}
a_{ij} &= (\bar{v}_{j.}^{(i)} - \ell_{ij})^2 \left[ \frac{ k_i (1 + k_i \kappa) }{ \vec{v}_{j.}^{(i)T} \vec{v}_{j.}^{(i)}  - k_i (\bar{v}_{j.}^{(i)})^2 } \right] \frac{u_{ij} - \bar{v}_{j.}^{(i)}}{u_{ij} - \ell_{ij}} - \frac{\bar{v}_{j.}^{(i)} - \ell_{ij}}{u_{ij} - \ell_{ij}}, \\
b_{ij} &= a_{ij} \left( \frac{u_{ij} - \bar{v}_{j.}^{(i)}}{\bar{v}_{j.}^{(i)} - \ell_{ij}} \right).
\label{eqn:beta-moment-match}
\end{align}
In the special case that $k=2$, we have
\begin{align*}
&\bar{v}_{j.}^{(i)}
= \frac{1}{2} \left[ \min_{\ell \in \{1,2\}} v_{j\ell}^{(i)} +  \max_{\ell \in \{1,2\}} v_{j\ell}^{(i)} \right]
= \frac{1}{2} \left[ \ell_{ij} + u_{ij} \right], \\
&\bar{v}_{j.}^{(i)} - \ell_{ij} = u_{ij} - \bar{v}_{j.}^{(i)}, \\
&\vec{v}_{j.}^{(i)T} \vec{v}_{j.}^{(i)}
= u_{ij}^2 + \ell_{ij}^2,
\end{align*}
from which it can be shown that $a_{ij} = \kappa$ and $b_{ij} = \kappa$.

\citet{RaimPhDThesis2014} observes through simulation that, although the linear-combination-of-Dirichlet density can differ substantially from the moment-matched Beta density, the density of model~\eqref{eqn:interval-hierarchy-beta} is a close approximation to the density of model~\eqref{eqn:interval-hierarchy-dirichlet}. We have paid specific attention to the marginal distributions of the coordinates of $\vec{\mu}_{i}$ rather than the full joint distribution; it is seen from \eqref{eqn:mixture-link-density} that only the marginals influence the overall Mixture Link distribution. The density of model~\eqref{eqn:interval-hierarchy-beta} is now given by
\begin{align}
f(y_i \mid \bbeta, \vec{\pi}, \vec{\phi}_i, \kappa) = \sum_{j=1}^J \pi_j
\int_0^1 g(y_i \mid H_{ij}(w), \vec{\phi}_{ij}) \cdot
\mathcal{B}\left(w \given a_{ij}, b_{ij} \right) dw,
\label{eqn:mixture-link-density-probs}
\end{align}
where $\mathcal{B}(x \mid a, b)$ denotes the Beta density and $H_{ij}(x) = (u_{ij} - \ell_{ij}) x + \ell_{ij}$.

Computation of the Mixture Link density and its moments depends on the vertices of the set $A$. For the case $J=2$, it is easy to identify the vertices of $A$ graphically by plotting the line $\mu_1 \pi_1 + \mu_2 \pi_2 = \vartheta$, and visually identifying the points at which it intersects the unit rectangle. An illustration is given in Figure~\ref{fig:hyperplane-J2}. Formulas for the vertices in this case are stated now as a lemma.
\begin{lemma}
\label{result:vertices-J2-probs}
Suppose $J=2$ and $A$ has two distinct vertices $\vec{v}_1, \vec{v}_2$. Then the vertices are given by
\begin{align*}
&\vec{v}_1 =
\begin{cases}
\left( \frac{1}{\pi_1} \vartheta, 0 \right), \quad &\text{if $\frac{1}{\pi_1} \vartheta \leq 1$} \\
\left( 1, \frac{1}{\pi_2} (\vartheta - \pi_1) \right), \; \quad &\text{otherwise},
\end{cases}
\\
&\vec{v}_2 =
\begin{cases}
\left( \frac{1}{\pi_1} (\vartheta - \pi_2), 1 \right), \quad &\text{if $\frac{1}{\pi_1} (\vartheta - \pi_2) \geq 0$} \\
\left( 0, \frac{1}{\pi_2} \vartheta \right), \quad &\text{otherwise},
\end{cases}
\end{align*}
where $\pi_2 = 1 - \pi_1$.
\end{lemma}

\begin{proof}
Using $\mu_1 \pi_1 + \mu_2 \pi_2 = \vartheta$ we have
\begin{align}
\mu_1 = \frac{1}{\pi_1}(\vartheta - \mu_2 \pi_2 ) \quad \text{and} \quad
\mu_2 = \frac{1}{\pi_2}(\vartheta - \mu_1 \pi_1 ),
\label{eqn:solnsJ2}
\end{align}
where $\mu_1 \in [0,1]$ and $\mu_2 \in [0,1]$ must hold.
To obtain $\vec{v}_1$, take $\mu_1$ as large as possible noting expressions \eqref{eqn:solnsJ2}. If $\mu_1 = 1$ is a valid solution (i.e. a point in $A$), then $\mu_2 = \frac{1}{\pi_1} (\vartheta - \pi_2)$. Otherwise, take $\mu_2$ as small as possible to maximize $\mu_1$; this yields $\mu_1 = \frac{1}{\pi_1} \vartheta$ and $\mu_2 = 0$. A similar argument taking $\mu_1$ as small as possible yields $\vec{v}_2$.
\end{proof}

\begin{figure}
\centering
\includegraphics[width = 0.45\textwidth]{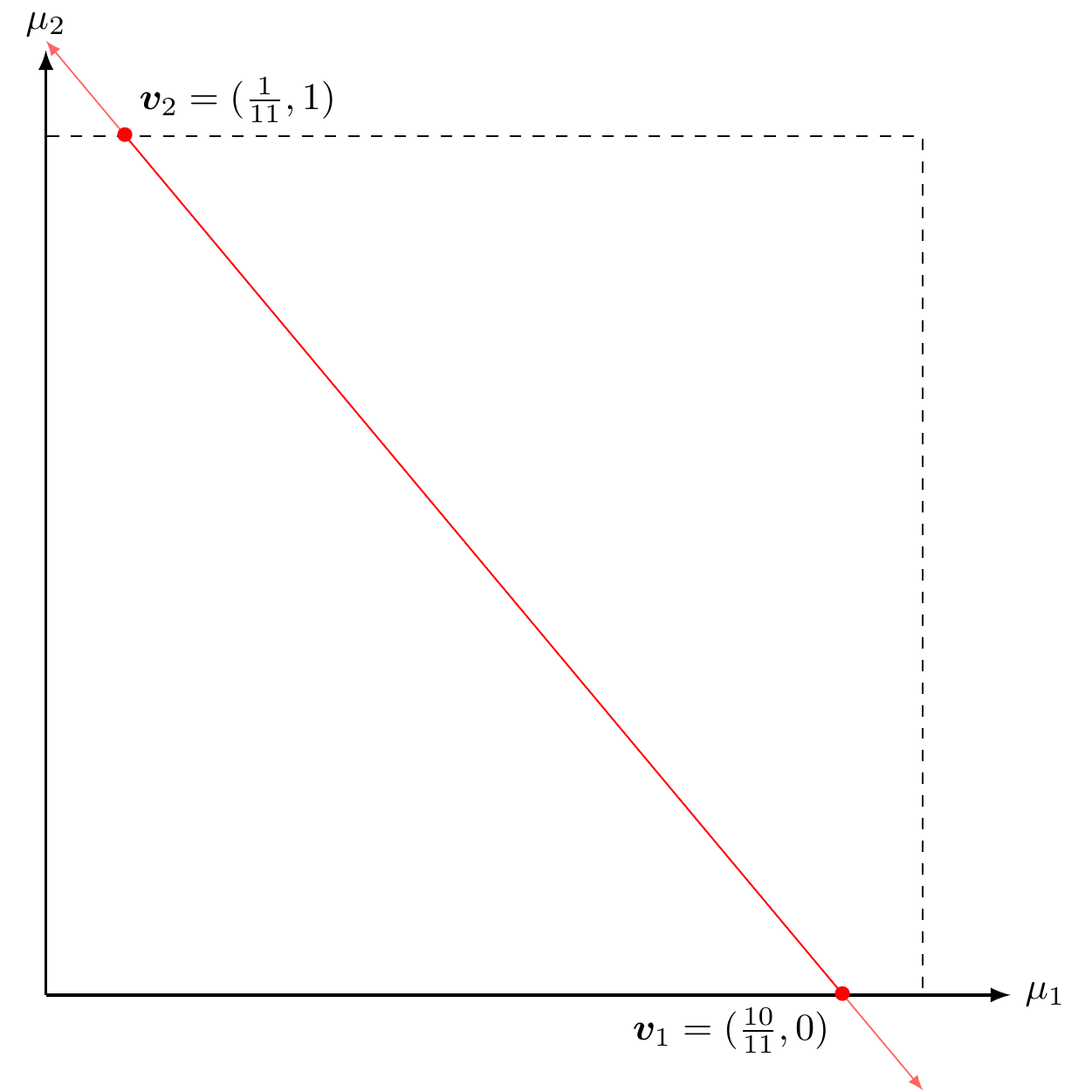}
\caption{An illustration of the set $A(\vartheta, \vec{\pi}) = \{ \vec{\mu} \in [0, 1]^J : \vec{\mu}^T \vec{\pi} = \vartheta \}$. Here we have selected $\vec{\pi} = (\frac{11}{20}, \frac{9}{20})$ and $\vartheta = \frac{1}{2}$.}
\label{fig:hyperplane-J2}
\end{figure}

\noindent
We may also locate the vertices $\vec{v}_1, \vec{v}_2$ systematically in the following way. Fix $\mu_2 = 0$ and solve for $\mu_1$ so that $\vec{\mu}^T \vec{\pi} = \vartheta$. Then fix $\mu_2 = 1$ and solve for $\mu_1$. Then fix $\mu_1$ at the values 0 and 1 and solve for $\mu_2$. At most two of these four solutions are contained in $A$; these are the vertices. We will soon see that this idea generalizes to $J > 2$. Note that it is also possible to have $k=1$ vertices when $J=2$. For example, if $\vec{\pi} = (1/2, 1/2)$ and $\vartheta = 1$, then $\mu_1 = 1, \mu_2 = 1$ is the only solution to $\mu_1 \pi_1 + \mu_2 \pi_2 = \vartheta$ in $[0,1]^2$, and therefore $A$ is a singleton set.

For the general ($J \geq 2$) case, Lemma~\ref{result:extreme-points-probs} characterizes points in $A$ which need to be considered when searching for the extreme points. In searching for extreme points, we must only consider those with at most one component not equal to 0 or 1.

\begin{lemma}[Characterization of Extreme Points of $A$]
\label{result:extreme-points-probs}
Suppose $\vec{v} = (v_1, \ldots, v_J)$ is a point in $A$ with two or more components strictly between 0 and 1. Then $\vec{v}$ is not an extreme point of $A$.
\end{lemma}

\begin{proof}
Suppose without loss of generality that $\vec{v} \in A$ with $v_1 \in (0,1)$ and $v_2 \in (0,1)$. We have that
\begin{align*}
\vec{v}^T \vec{\pi} = \vartheta
\quad &\iff \quad
v_1 \pi_1 + v_2 \pi_2 + (v_3 \pi_3 + \cdots + v_J \pi_J) = \vartheta \\
\quad &\iff \quad
v_1 \pi_1 + v_2 \pi_2 = \vartheta^*,
\end{align*}
where $\vartheta^* = \vartheta - (v_3 \pi_3 + \cdots + v_J \pi_J)$. We can now use Lemma~\ref{result:vertices-J2-probs} to obtain vertices, say $\vec{a}$ and $\vec{b}$, of the line segment
\begin{align*}
L = \left\{ (\mu_1, \mu_2, v_3, \ldots, v_J) \in [0,1]^J : \mu_1 \pi_1 + \mu_2 \pi_2 = \vartheta^* \right\},
\end{align*}
where $(v_3, \ldots, v_J)$ are held fixed and only $(\mu_1, \mu_2)$ may vary. Explicitly, we have
\begin{align*}
&\vec{a} =
\begin{cases}
\left( \frac{1}{\pi_1} \vartheta^*, 0, v_3, \ldots, v_J \right), \quad &\text{if $\frac{1}{\pi_1} \vartheta^* \leq 1$} \\
\left( 1, \frac{1}{\pi_2} (\vartheta^* - \pi_1), v_3, \ldots, v_J \right), \; \quad &\text{otherwise},
\end{cases}
\\
&\vec{b} =
\begin{cases}
\left( \frac{1}{\pi_1} (\vartheta^* - \pi_2), 1, v_3, \ldots, v_J \right), \quad &\text{if $\frac{1}{\pi_1} (\vartheta^* - \pi_2) \geq 0$} \\
\left( 0, \frac{1}{\pi_2} \vartheta^*, v_3, \ldots, v_J \right), \quad &\text{otherwise}.
\end{cases}
\end{align*}
By construction, we have that $\vec{v}$ is in the line segment strictly between $\vec{a}$ and $\vec{b}$, with $\vec{a} \neq \vec{b}$. Furthermore, since $L \subseteq A$, we have that $\vec{a}, \vec{b} \in A$. Therefore, $\vec{v}$ can not be an extreme point of $A$.
\end{proof}

This can be used to formulate a simple procedure to identify all extreme points of $A$, which is given as Algorithm~\ref{alg:find-vertices-probs}. Notice that it considers $J \cdot 2^{J-1}$ points; this would be impractical for large $J$, but is manageable for smaller values of $J$ that are commonly used in finite mixtures.

\begin{algorithm}
\caption{Find vertices of the set $A(\vartheta, \vec{\pi})$.}
\label{alg:find-vertices-probs}
\begin{algorithmic}
\Function{FindVertices}{$\vartheta, \vec{\pi}$}
\State $\mathcal{V} \gets \varnothing$
\For{$j = 1, \ldots, J$}
\If{$\pi_j > 0$}
\ForAll{$\vec{\mu}_{-j} \in \{0,1\}^{J-1}$}
\State $\mu_j^* \gets \pi_j^{-1} \left[ \vartheta - \vec{\mu}_{-j}^T \vec{\pi}_{-j} \right]$
\State $\vec{v}^* \gets (\mu_1, \ldots, \mu_{j-1}, \mu_j^*, \mu_{j+1}, \ldots, \mu_J)$
\State $\mathcal{V} \gets \mathcal{V} \cup \vec{v}^*$ if $\vec{v}^* \in A(\vartheta, \vec{\pi})$
\EndFor
\EndIf
\EndFor
\State \Return Matrix $\vec{V}$ with columns $\vec{v}^* \in \mathcal{V}$
\EndFunction
\end{algorithmic}
\end{algorithm}

We will now formulate a Mixture Link Binomial distribution. Suppose $g(y_i \mid w, \vec{\phi}_{ij}) = \text{Bin}(y_i \mid m_i, w)$ so that $y_i$ represents a count of successes out of $m_i$ independent trials. Model \eqref{eqn:interval-hierarchy-beta} becomes 
\begin{align}
Y_i &\indep \sum_{j=1}^J \pi_j \binom{m_i}{y_i} \mu_{ij}^{y_i} (1 - \mu_{ij})^{m_i-y_i}, \label{eqn:binomial-mixlink} \\
\mu_{ij} &= (u_{ij} - \ell_{ij})\psi_{ij} + \ell_{ij}, \quad \text{$j = 1, \ldots, J$}, \nonumber \\
\psi_{ij} &\sim \text{Beta}(a_{ij}, b_{ij}). \nonumber
\end{align}
To draw from this distribution,
\begin{enumerate*}
\item Compute matrix $\vec{V}$ given $\vec{x}$, $\bbeta$, and $\vec{\pi}$.
\item Compute $a_j$ and $b_j$ for $j = 1, \ldots, J$ according to \eqref{eqn:beta-moment-match}, and let $(\ell_j, u_j)$ be the minimum and maximum element, respectively, of the $j$th row of $\vec{V}$.
\item Let $\mu_j = (u_j - \ell_j) \psi_j + \ell_j$ with $\psi_j \sim \text{Beta}(a_j, b_j)$, for $j = 1, \ldots, J$.
\item Draw $Z \sim \text{Discrete}(1, \ldots, J; \vec{\pi})$.
\item Draw $Y \sim \text{Binomial}(m, \mu_Z)$.
\end{enumerate*}
Here, $\text{Discrete}(1, \ldots, k; \vec{p})$ denotes the discrete distribution with values $1, \ldots, k$ and corresponding probabilities $\vec{p} = (p_1, \ldots, p_k)$. Moments of $Y$ can be computed using moments of $\mu_j$ for $j = 1, \ldots, J$. In particular, after some algebra, we obtain
\begin{align*}
\Var(Y) = m \vartheta \left( 1 - m \vartheta \right)
+ m(m-1) \sum_{j=1}^J \pi_j \frac{ \vec{v}_{j.}^T \vec{v}_{j.} + \kappa (k \bar{v}_{j.})^2 }{ k (1 + \kappa k) }.
\end{align*}
Some remarks about the Mixture Link Binomial distribution follow.%
\footnote{Analogous statements for some of these remarks can be made about the Mixture Link Poisson and Mixture Link Normal distributions, discussed in Sections~\ref{sec:mixlink-positive} and \ref{sec:mixlink-real}. We have focused on the Binomial case for brevity.}

\begin{remark}
For the case $m = 1$ where $y$ represents a single success or failure, $\E(Y) = \vartheta$ implies $\Prob(Y = 1) = \vartheta^{y} (1 - \vartheta)^{1-y}$, and Mixture Link simplifies to the usual Bernoulli regression model. In this case, the distribution depends only on its $\bbeta$ parameter. When $m > 1$, this trivial simplification does not take place.
\end{remark}

\begin{remark}
Note that because $\vec{v}_{j.}^T \vec{v}_{j.} \leq k$ and $\bar{v}_{j.} \leq 1$, we have $\sum_{j=1}^J \pi_j \vec{v}_{j.}^T \vec{v}_{j.} + \kappa (k \bar{v}_{j.})^2 \leq k (1 + \kappa k)$, yielding the bound $\Var(Y) \leq m(m-1) - m \vartheta (m \vartheta - 1)$, which is free of $\vec{\pi}$ and $\kappa$.
\end{remark}

\begin{remark}
The expression $\Var(Y)$ is non-increasing in $\kappa$. This can be seen from
\begin{align*}
\frac{\partial}{\partial \kappa} \Var(Y) =
-\frac{m(m-1)}{(1 + \kappa k)^2} \sum_{j=1}^J \pi_j \sum_{\ell=1}^k (v_{j\ell} - \bar{v}_{j.})^2 \leq 0.
\end{align*}
\end{remark}

\begin{remark}
\label{remark:binomial-case}
$\text{Binomial}(m, \vartheta)$ is a special case of Mixture Link Binomial, when $\vec{\pi} = (\frac{1}{J}, \ldots, \frac{1}{J})$ and $\kappa \rightarrow \infty$. This can be seen directly from the Dirichlet formulation of Mixture Link \eqref{eqn:interval-hierarchy-dirichlet}. Let $\vec{\pi} = (\frac{1}{J}, \ldots, \frac{1}{J})$ so that $A(\vec{\pi}, \vartheta) = \{ \vec{\mu} \in [0,1]^J : \mu_1 + \cdots + \mu_J = J \vartheta \}$. A vertex $\vec{v}^*$ of $A(\vec{\pi}, \vartheta)$ is obtained by taking, say, the first $v_1^*, \ldots, v_{[J \vartheta]}^*$ to be 1, $v_{[J \vartheta] + 1}^* = J \vartheta - [J \vartheta]$, and the remaining elements of $\vec{v}^*$ to be zero. Here, $[x]$ represents the integer part of a real number $x$. By Lemma~\ref{result:extreme-points-probs}, $\vec{v}^*$ is a vertex of $A(\vec{\pi}, \vartheta)$. The remaining vertices can be obtained by permuting the elements of $\vec{v}^*$. If $\tilde{v}_1^*, \ldots, \tilde{v}_s^*$ are the unique elements of $\vec{v}^*$ with multiplicities $|\tilde{v}_1^*|, \ldots, |\tilde{v}_s^*|$, then there are $k = J! / \{ |\tilde{v}_1^*|! \cdots |\tilde{v}_s^*|! \}$ unique permutations of $\vec{v}^*$ to use as columns in the matrix $\vec{V}$. Notice that, for any $a, j \in \{ 1, \ldots, J \}$, the element $\tilde{v}_a^*$ appears in the $j$th row $\vec{v}_{j.}^T$ of $\vec{V}$ exactly $(J-1)! / \{ |\tilde{v}_a^* - 1|! \prod_{\ell \neq a} |\tilde{v}_\ell^*|! \}$ times.%
\footnote{This is the number of unique permutations of $\{ v_1^*, \ldots, v_J^* \}$, keeping one of the elements fixed.}
Then we have
\begin{align}
\vec{v}_{j.}^T \vec{1} = \sum_{a=1}^s \tilde{v}_a^* \frac{(J-1)! }{ |\tilde{v}_a^* - 1|! \prod_{\ell \neq a} |\tilde{v}_\ell^*|! }
= \sum_{a=1}^s \tilde{v}_a^* \frac{J! |\tilde{v}_a^*|}{ \prod_{\ell=1}^a |\tilde{v}_\ell^*|! } \frac{1}{J}
= \frac{k}{J} \sum_{a=1}^s \tilde{v}_a^* \cdot |\tilde{v}_a^*|
= \frac{k}{J} J \vartheta = k \vartheta.
\label{eqn:row-simplification}
\end{align}
When $\kappa \rightarrow \infty$, a draw $\vec{\lambda} \sim \text{Dirichlet}_k(\kappa \vec{1})$ becomes a point mass at its expected value $\frac{1}{k} \vec{1}$ so that \eqref{eqn:row-simplification} gives $\vec{\mu} = \vec{V} \vec{\lambda} = \frac{1}{k} \vec{V} \vec{1} = \vartheta \vec{1}$. It can now be seen that
\begin{align*}
f(y) = \sum_{j=1}^J \pi_j \binom{m}{y} \mu_{j}^{y} (1 - \mu_{j})^{m-y}
= \sum_{j=1}^J \frac{1}{J} \binom{m}{y} \vartheta^{y} (1 - \vartheta)^{m-y}
\end{align*}
is the $\text{Binomial}(m, \vartheta)$ distribution.
\end{remark}

\begin{remark}
\label{remark:zi-case}
Mixture Link Binomial becomes a zero- and/or $m$-inflated Binomial model when $\kappa \rightarrow 0$. As in Remark~\ref{remark:binomial-case}, we will work directly from the Dirichlet formulation. As $\kappa \rightarrow 0$, a draw $\vec{\lambda} \sim \text{Dirichlet}_k(\kappa \vec{1})$ behaves as a discrete uniform random variable on $\{ \vec{e}_1, \ldots, \vec{e}_k \}$, the columns of the $k \times k$ identity matrix which represent the vertices of the simplex $\mathcal{S}^k$. Here, the Mixture Link distribution becomes
\begin{align*}
f(y) &= \sum_{j=1}^J \pi_j \sum_{\ell=1}^k \frac{1}{k} \cdot \text{Bin}(y \mid m, \vec{v}_{j.}^T \vec{e}_\ell) \\
&= \sum_{j=1}^J \sum_{\ell=1}^k \frac{\pi_j}{k} \binom{m}{y} v_{j\ell}^{y} (1 - v_{j\ell})^{m-y}.
\end{align*}
Recall from Lemma~\ref{result:extreme-points-probs} that, for each $\ell = 1, \ldots, k$, at most one of $\{ v_{1 \ell}, \ldots, v_{J \ell} \}$ can take on a value outside of $\{0, 1\}$. Terms with $v_{J \ell} = 0$ represent a point mass at zero, while terms with $v_{J \ell} = 1$ represent a point mass at $m$.
\end{remark}

\begin{remark}
Mixture Link Binomial is closely related to two other Binomial models for overdispersion. Starting from \eqref{eqn:mixture-link-density-probs}, if we could take $\ell_{ij} = 0$ and $u_{ij} = 1$, we would have
\begin{align*}
f(y_i \mid \bbeta, \vec{\pi}, \vec{\phi}_i, \kappa) &= \sum_{j=1}^J \pi_j
\int_0^1 \text{Bin}(y_i \mid (u_{ij} - \ell_{ij})w + \ell_{ij}, \vec{\phi}_{ij}) \cdot
\mathcal{B}\left(w \given a_{ij}, b_{ij} \right) dw, \\
&= \sum_{j=1}^J \pi_j \binom{m_i}{y_i}
\frac{ B(a_{ij} + y_i, b_{ij} + m_i - y_i) }{ B(a_{ij}, b_{ij}) }.
\end{align*}
Therefore, Mixture Link Binomial can be seen as a constrained form of a finite mixture of $J$ Beta-Binomial densities. Also, recall the Random-Clumped Binomial (RCB) distribution \citep{MorelNagaraj1993}, whose density is given by
\begin{align*}
f(y \mid \pi, \rho) = \pi_1 \text{Bin}(y \mid \pi, \mu_1) + \pi_2 \text{Bin}(y \mid \pi, \mu_2),
\end{align*}
where $\pi_1 = \pi$, $\pi_2 = 1-\pi$, and $\mu_1 = (1-\rho)\pi + \rho$, $\mu_2 = (1-\rho)\pi$. The free parameters of the distribution are $\pi \in (0,1)$ and $\rho \in (0,1)$. Notice that $\pi_1 \mu_1 + \pi_2 \mu_2 = \pi$, so that this particular choice of $(\mu_1, \mu_2)$ is in the set $A(\pi_1, \vec{\pi})$. Therefore, RCB can be seen as a special case of Mixture Link Binomial.
\end{remark}

\section{Positive Means}
\label{sec:mixlink-positive}

The setting $\mathcal{M} = [0, \infty)$ is commonly required for count data and time-to-event data. Just as in Section~\ref{sec:mixlink-probs}, the set $A(\vartheta, \vec{\pi}) = \{ \vec{\mu} \in [0, \infty)^J : \vec{\mu}^T \vec{\pi} = \vartheta \}$ is a closed convex hyperplane segment within $\mathbb{R}^J$. Therefore, the decomposition \eqref{eqn:mink-weyl} also applies but the procedure to compute vertices is much simpler. First note that for $J=2$, $\vec{v}_1 = (\vartheta / \pi_1, 0)$ and $\vec{v}_2 = (0, \vartheta / \pi_2)$ are the vertices of $A$. To see this, suppose $\vec{\mu}^*$ is an arbitrary point in $A$. Then we must have, for some $\lambda \in [0,1]$,
\begin{align*}
\begin{pmatrix}
\mu_1^* \\
\mu_2^*
\end{pmatrix}
= \lambda \vec{v}_1 + (1-\lambda) \vec{v}_2 =
\begin{pmatrix}
\lambda \vartheta / \pi_1 \\
(1-\lambda) \vartheta / \pi_2
\end{pmatrix}.
\end{align*}
Taking $\lambda = \mu_1^* \pi_1 / \vartheta$ satisfies the first equation $\mu_1^* = \lambda \vartheta / \pi_1$, and also gives $(1-\lambda) \vartheta / \pi_2 = (\vartheta - \mu_1^* \pi_1) / \pi_2 = \mu_2^*$ to satisfy the second equation. Similarly to Lemma~\ref{result:extreme-points-probs}, we characterize the extreme points of $A$ for the case of positive means by Lemma~\ref{result:extreme-points-positive}. The proof is similar to that of Lemma~\ref{result:extreme-points-probs}, and therefore omitted.

\begin{lemma}[Characterization of Extreme Points of $A$]
\label{result:extreme-points-positive}
Suppose $\vec{v} = (v_1, \ldots, v_J)$ is a point in $A$ with two or more components which are strictly positive. Then $\vec{v}$ is not an extreme point of $A$.
\end{lemma}

Now, if $\vec{v} = (0, \ldots, 0, v_j, 0, \ldots, 0)$ is a point in $A$, $\vec{v}^T \vec{\pi} = \vartheta$ implies $v_j \pi_j = \vartheta$. There are exactly $J$ such points in $A$, yielding $\vec{V} = \Diag(\vartheta / \pi_1, \ldots, \vartheta / \pi_J)$. Poisson Mixture Link can now be formulated similarly as in Section~\ref{sec:mixlink-probs}. Note that, in this case, the Dirichlet and Beta assumptions on $\mu_{i}$ lead to exactly the same model. Taking $g(y_i \mid w, \vec{\phi}_{ij}) = \text{Poisson}(y_i \mid w)$, the model becomes
\begin{align*}
Y_i &\indep \sum_{j=1}^J \pi_j \frac{e^{-\mu_{ij}} \mu_{ij}^{y_i}}{y_i!} \\
&\vec{\mu}_i = \vec{V}^{(i)} \vec{\lambda}^{(i)}, \\
&\vec{\lambda}^{(i)} \indep \text{Dirichlet}_{k_i}(\kappa \vec{1}).
\end{align*}
Expressions involving the vertices simplify in the case of positive means, with $J = k_i$, $\ell_{ij} = 0$, $u_{ij} = v_{jj}^{(i)}$, $\bar{v}_{j.}^{(i)} = v_{jj}^{(i)} / J$, $\vec{v}_{j.}^{(i)T} \vec{v}_{j.}^{(i)} = (v_{jj}^{(i)})^2$, $H_{ij}(w) = v_{jj}^{(i)} w$, $a_{ij} = \kappa$, and $b_{ij} = \kappa (J-1)$. Recalling that the marginal distribution of a single coordinate of $\text{Dirichlet}_{J}(\kappa \vec{1})$ is $\text{Beta}(\kappa, \kappa (J-1))$, the Mixture Link density becomes
\begin{align*}
f(y_i \mid \bbeta, \vec{\pi}, \kappa)
&= \sum_{j=1}^J \pi_j
\int_0^1 \frac{e^{-H_{ij}(w)} H_{ij}(w)^{y_i}}{y_i!} \cdot
\mathcal{B}\left(w \given \kappa, \kappa (J-1) \right) dw \\
&= \sum_{j=1}^J \pi_j
\int_0^1 \frac{e^{-v_{jj}^{(i)} w} [v_{jj}^{(i)} w]^{y_i}}{y_i!} \cdot
\frac{w^{\kappa-1} (1-w)^{\kappa (J-1) - 1}}{B(\kappa, \kappa (J-1))} dw \\
&= \frac{\vartheta_i^{y_i} \Gamma(y_i + \kappa) \Gamma(\kappa J) }{ \Gamma(y_i + \kappa J) \Gamma(\kappa) \Gamma(y_i+1) }
\sum_{j=1}^J \pi_j^{1 - y_i} \cdot
\mathcal{F}\left( -\frac{\vartheta_i}{\pi_j}; y_i + \kappa, y_i + J\kappa \right)
\end{align*}
where $\mathcal{F}(x; a, b) = [B(a, b-a)]^{-1} \int_0^1 w^{a-1} (1-w)^{b-a-1} e^{xw} dw$ is the confluent hypergeometric function of the first order and $B(a, b) = \Gamma(a) \Gamma(b) / \Gamma(a+b)$ is the beta function \citep[Chapter 1]{JohnsonKotzKemp2005}. Implementations of $\mathcal{F}(x; a, b)$ are available in computing packages such as the GNU Scientific Library.%
\footnote{\url{www.gnu.org/software/gsl}} The variance of $Y$ becomes
\begin{align*}
\Var(Y) &= \vartheta + \left[ \sum_{j=1}^J \pi_j \bar{v}_{j.}^2 - \vartheta^2 \right] + 
\sum_{j=1}^J \pi_j \frac{\vec{v}_{j.}^T \vec{v}_{j.} - k(\bar{v}_{j.})^2}{k (1 + \kappa k)} \\
&= \vartheta + \vartheta^2 \left[ \frac{ \kappa+1}{J (1 + J \kappa)} \sum_{j=1}^J \frac{1}{\pi_j} - 1 \right].
\end{align*}
Drawing random variables from Mixture Link Poisson is similar to the method given in Section~\ref{sec:mixlink-probs} for Mixture Link Binomial:
\begin{enumerate*}
\item Compute matrix of vertices $\vec{V}$ given $\vec{x}$, $\bbeta$, and $\vec{\pi}$.
\item Let $\mu_j = \psi_j \cdot \vartheta / \pi_j$ with $\psi_j \sim \text{Beta}(\kappa, \kappa (J-1))$, for $j = 1, \ldots, J$.
\item Draw $Z \sim \text{Discrete}(1, \ldots, J; \vec{\pi})$.
\item Draw $Y \sim \text{Binomial}(m, \mu_Z)$.
\end{enumerate*}

\begin{remark}
The expression $\Var(Y)$ is decreasing in $\kappa$ since
\begin{align*}
\frac{\partial}{\partial \kappa} \Var(Y) = -\frac{\vartheta (J-1)}{J(1 + J\kappa)} \sum_{j=1}^J \frac{1}{\pi_j} < 0.
\end{align*}
\end{remark}

\section{Real-valued Means}
\label{sec:mixlink-real}
In the case $\mathcal{M} = \mathbb{R}$, the set $A(\vartheta, \vec{\pi}) = \{ \vec{\mu} \in \mathbb{R}^J : \vec{\mu}^T \vec{\pi} = \vartheta \}$ forms a hyperplane in $\mathbb{R}^J$ and can be decomposed as $A(\vartheta, \vec{\pi}) = \{ \bar{\vec{\mu}} \in \mathbb{R}^J : \bar{\vec{\mu}}^T \vec{\pi} = 0 \} + \vartheta \vec{1}$. For any $\bar{\vec{\mu}}$ in the subspace $\{ \bar{\vec{\mu}} \in \mathbb{R}^J : \bar{\vec{\mu}}^T \vec{\pi} = 0 \}$, we can write $\bar{\mu}_J = -\pi_J^{-1} (\pi_1 \bar{\mu}_1 + \cdots + \pi_{J-1} \bar{\mu}_{J-1})$ with $\bar{\mu}_j$ unrestricted for $j = 1, \ldots, J-1$. Therefore a basis for the subspace is given by the $J \times (J-1)$ matrix
\begin{align*}
\vec{V} =
\begin{pmatrix}
1 & 0 & \cdots & 0 \\
0 & 1 & \cdots & 0 \\
  &   & \ddots &   \\
0 & 0 & \cdots & 1 \\
-\pi_1/\pi_J & -\pi_2/\pi_J & \cdots & -\pi_{J-1}/\pi_J 
\end{pmatrix}.
\end{align*}
We can therefore represent any $\vec{\mu} \in A(\vartheta, \vec{\pi})$ as
\begin{align*}
\vec{\mu} = \vec{V} \vec{\lambda} + \vartheta \vec{1} \quad \text{for some $\vec{\lambda} \in \mathbb{R}^{J-1}$.}
\end{align*}
A natural choice for a random effects distribution on $A(\vartheta, \vec{\pi})$ is to take $\lambda_j \iid \text{N}(0, \kappa^2)$ for $j = 1, \ldots, J-1$. This leads to
\begin{align*}
\vec{\mu} \sim \text{N}(\vartheta \vec{1}, \kappa^2 \vec{V} \vec{V}^T), \quad \text{where} \quad
\vec{V} \vec{V}^T =
\begin{pmatrix}
\vec{I}                      & -\pi_J^{-1} \vec{\pi}_{-J} \\
-\pi_J^{-1} \vec{\pi}_{-J}^T &  \pi_J^{-2} \vec{\pi}_{-J}^T \vec{\pi}_{-J}
\end{pmatrix},
\end{align*}
$\vec{I}$ denotes the $(J-1) \times (J-1)$ identity matrix, and $\vec{\pi}_{-J} = (\pi_1, \ldots, \pi_{J-1})$. The Mixture Link density depends only on the diagonal terms of the random effect variance,
\begin{align}
f(y_i \mid \bbeta, \vec{\pi}, \vec{\phi}_i, \kappa)
&= \sum_{j=1}^J \pi_j \int g(y_i \mid w, \vec{\phi}_{ij}) \cdot \text{N}(w \mid \vartheta_i, \kappa^2 a_{ij}) dw,
\label{eqn:mixture-link-density-unconstr}
\end{align}
where $a_{ij} = 1$ for $j = 1, \ldots, J-1$ and $a_{iJ} = \pi_J^{-2} \vec{\pi}_{-J}^T \vec{\pi}_{-J}$.

To obtain a Mixture Link analogue to the commonly used ordinary least squares model, suppose $g(y_i \mid w, \phi_{ij}) = \text{N}(y_i \mid w, \sigma_j^2)$. In this case, it can be shown that \eqref{eqn:mixture-link-density-unconstr} simplifies to the finite mixture
\begin{align}
f(y_i \mid \bbeta, \vec{\pi}, \sigma_1^2, \ldots, \sigma_J^2, \kappa)
&= \sum_{j=1}^J \pi_j \text{N}(y_i \mid \vartheta_i, \kappa^2 a_{ij} + \sigma_j^2),
\label{eqn:mixture-link-density-normal}
\end{align}
where each of the subpopulations has a common mean. If the $J$ subpopulations are assumed to be homoskedastic, \eqref{eqn:mixture-link-density-normal} further simplifies to a finite mixture of two densities,
\begin{align*}
f(y_i \mid \bbeta, \vec{\pi}, \sigma^2, \kappa)
= (1-\pi_J) \text{N}(y_i \mid \vartheta_i, \kappa^2 + \sigma^2)
+ \pi_J \text{N}(y_i \mid \vartheta_i, \kappa^2 \pi_J^{-2} (1 - \pi_J)^2 + \sigma^2).
\end{align*}
Focusing on the homoskedastic model, it is straightforward to draw from the distribution:
\begin{enumerate*}
\item Draw $Z_i \sim \text{Discrete}(1, 2; (1-\pi_J, \pi_J))$,
\item Draw $Y_i$ from $\text{N}(y_i \mid \vartheta_i, \kappa^2 a_{ij} + \sigma^2)$ where $Z_i = j$.
\end{enumerate*}
An expression for the variance is given by
\begin{align*}
\Var(Y_i) = \kappa^2 \frac{1 - \pi_J}{\pi_J} + \sigma^2.
\end{align*} 

\section{Data Analysis Examples}
\label{sec:data-examples}
We now present two examples of data analysis with the Mixture Link distribution. The Hiroshima data discussed in Section~\ref{sec:data-hiroshima} features a Binomial outcome. The Arizona Medpar data has a count outcome, and is discussed in Section~\ref{sec:data-azpro}.

For a complete Bayesian specification of Mixture Link Binomial and Mixture Link Poisson, we assume priors
\begin{align*}
\bbeta &\sim \text{N}(\vec{0}, \Vbeta), \\
\vec{\pi} &\sim \text{Dirichlet}(\vec{\gamma}), \\
\kappa &\sim \text{Gamma}(a_\kappa, b_\kappa),
\end{align*}
where the parameterization of Gamma is taken to have $\E(\kappa) = a_\kappa / b_\kappa$. In the absence of a-priori knowledge, a somewhat vague choice of hyperparameters is $\Vbeta = 1000 \vec{I}_{d}$, $\vec{\gamma} = \vec{1}$, and $a_\kappa = 1, b_\kappa = 2$.

To diagnose the fit of models with non-Normal outcomes, we make use of the randomized quantile residuals \citep{DunnSmyth1996}. Interpretation of quantile residuals is similar to the routine residual analysis from ordinary least squares regression. Quantile residuals from an adequate model fit appear to behave as an independent sample from the standard Normal distribution. For $y_i$ drawn independently from a continuous distribution $F(\cdot \mid \btheta)$ with estimate $\hat{\btheta}$, the quantile residual is defined as $r_i = \Phi^{-1}\{ F(y_i \mid \hat{\btheta}) \}$. For $y_i$ drawn independently from a discrete distribution, there is an additional randomization where the residual is defined by $r_i = \Phi^{-1}\{ u_i \}$, using $u_i$ drawn uniformly on the interval between $\lim_{\varepsilon \downarrow 0} F(y_i - \varepsilon \mid \hat{\btheta})$ and $F(y_i \mid \hat{\btheta})$. A Bayesian version of the quantile residual using draws $\btheta^{(1)}, \ldots, \btheta^{(R)}$ from the posterior distribution $f(\btheta \mid \vec{y})$ is $r_i = \frac{1}{R} \sum_{r=1}^R \Phi^{-1}\{ u_i^{(r)} \}$, where each $u_i^{(r)}$ is drawn uniformly on the interval between $\lim_{\varepsilon \downarrow 0} F(y_i - \varepsilon \mid \btheta^{(r)})$ and $F(y_i \mid \btheta^{(r)})$.

We will also evaluate models using prediction intervals computed from the posterior predictive distribution. Recall that the posterior predictive distribution for a new sample $\tilde{\vec{y}}$ given the observed sample $\vec{y}$ is
\begin{align*}
f(\tilde{\vec{y}} \mid \vec{y})
= \int f(\tilde{\vec{y}} \mid \btheta, \vec{y}) f(\btheta \mid \vec{y}) d\nu(\btheta)
= \int f(\tilde{\vec{y}} \mid \btheta) f(\btheta \mid \vec{y}) d\nu(\btheta),
\end{align*}
where $\nu$ denotes an appropriate dominating measure. Then to sample from $f(\tilde{\vec{y}} \mid \vec{y})$:
\begin{enumerate*}
\item Draw $\btheta^{(1)}, \ldots, \btheta^{(R)}$ from posterior $f(\btheta \mid \vec{y})$.
\item Draw $\tilde{\vec{y}}^{(r)}$ from $f(\tilde{\vec{y}} \mid \btheta^{(r)})$ for $r = 1, \ldots, R$.
\end{enumerate*}
Now $(\tilde{\vec{y}}^{(1)}, \ldots, \tilde{\vec{y}}^{(R)})$ is a draw from the posterior predictive distribution. A prediction for the $i$th observation is given by $\frac{1}{R} \sum_{r=1}^R \tilde{y}_i^{(r)}$, and a prediction interval with coverage probability $1-\alpha$ for the $i$th observation is given by the $\alpha/2$ and $1-\alpha/2$ quantiles of $(\tilde{y}_i^{(1)}, \ldots, \tilde{y}_i^{(R)})$.

Label switching is a common issue in Bayesian analysis of finite mixtures \citep{JasraEtAl2005}. For Mixture Link, the $\vec{\pi}$ parameters are susceptible to this problem. Because finite mixtures are invariant to permutation of the labels, the parameters corresponding to labels $\{ 1, \ldots, J \}$ can change during the course of an MCMC computation. Therefore, special care must be taken when summarizing parameters using MCMC draws. In this work, we take the simple approach of reordering the components within each draw $\vec{\pi}^{(r)}$, in ascending order, for each $r = 1, \ldots, R$.

\subsection{Hiroshima Data}
\label{sec:data-hiroshima}

\citet{AwaEtAl1971} and \citet{SofuniEtAl1978} study the effects of radiation exposure on chromosome aberrations in survivors of the atomic bombs that were used in Hiroshima and Nagasaki. We consider a subset of the data, as presented in \citet{MorelNeerchalSAS2012}, on $n = 648$ subjects in Hiroshima. For the $i$th subject, a chromosome analysis has been carried out on $m_i$ circulating lymphocytes to determine the number $y_i$ containing chromosome aberrations. Neutron and gamma radiation exposure (measured in rads) are available as potential covariates. As in \citet{MixLinkJSM2015}, we consider the regression
\begin{align}
\vartheta_i = G(\beta_0 + \beta_1 x_i + \beta_2 x_i^2),
\label{eqn:hiroshima-regression}
\end{align}
where $x_i$ is a normalized sum of neutron and gamma doses, and we take $G$ to be the logistic CDF (as in logistic regression).

We compare six Binomial-type models with \eqref{eqn:hiroshima-regression} as the regression function: Binomial, Random-Clumped Binomial (RCB), Beta-Binomial (BB), and Mixture Link with $J = 2, 3, 4$ mixture components (MixLinkJ2, MixLinkJ3, MixLinkJ4). Because of the complicated manner in which parameters enter the Mixture Link Binomial likelihood, conjugate priors leading to closed-form Gibbs samplers do not appear possible. We considered a simple Random Walk Metropolis-Hastings (RWMH) sampler \citep[Section 7.5]{RobertCasella2010}; however, sampling with RWMH is time consuming because it requires computation of the likelihood to determine whether each proposed jump will be accepted. Recall that, for Mixture Link Binomial, evaluation of the likelihood consists of evaluating $J$ integrals numerically for each of the $n$ observations. Alternatively, Appendix~\ref{sec:appendix-mcmc-binomial} proposes a Metropolis-within-Gibbs (MWG) sampler \citep[Section 10.3]{RobertCasella2010} where $\vec{\psi}_{i}$ are taken as augmented data \citep{TannerWong1987} to avoid the expensive integration.

An RWMH sampler was used to obtain posterior draws under the Binomial, RCB, and BB models, while the MWG sampler from Appendix~\ref{sec:appendix-mcmc-binomial} was used for Mixture Link. For each Mixture Link model, we carried out a preliminary ``pilot'' MCMC, which was used to tune the proposal distribution for a final MCMC run and achieve satisfactory mixing. Mixing was assessed primarily through trace plots and autocorrelation plots of the saved draws. Trace plots for the selected Mixture Link model are shown in Figure~\ref{fig:hiroshima-trace}. For all models, a multivariate Normal proposal distribution was selected by hand to achieve acceptance rates between about 15\% and 30\%. Final MCMC runs for Mixture Link were carried out for 55,000 iterations; the first 5,000 were discarded as a burn-in sample, and 1 of every 50 remaining draws from the chain were saved. For Binomial, BB, and RCB, we used 50,000 iterations overall with the first 5,000 discarded as burn-in and saved 1 of every 50 remaining.

Table~\ref{tab:hiroshima-model-selection} shows the Deviance Information Criterion (DIC) for these models. The three Mixture Link models fit best according to DIC; BB has a smaller DIC than RCB by a large margin, and Binomial gives the worst fit as expected. Table~\ref{tab:hiroshima-est} reports means, standard deviations, 2.5\% quantiles, and 97.5\% quantiles for each parameter from the posterior draws. Generally, signs and magnitudes of the $\bbeta$ estimates agree between models. Standard deviations and credible intervals are a bit larger for BB and MixLink models than RCB and Binomial. Figure~\ref{fig:hiroshoma-rqres} displays quantile residuals for the Binomial, BB, and MixLinkJ2 models. Residuals from BB and MixLinkJ2 are markedly closer to a $\text{N}(0,1)$ sample than Binomial residuals, as can be seen from the Q-Q plots. For all models, there is a systematic pattern in residuals vs.~predicted proportions, which is an indication that the mean is not fully explained by regression function~\eqref{eqn:hiroshima-regression}. Finally, Figure~\ref{fig:hiroshima-predict} plots $x_i$ against observed $y_i / m_i$, along with 95\% prediction intervals for Binomial, BB, and MixLinkJ2. The intervals computed by MixLinkJ2, and to a lesser extent BB, express variability from the observed data into wider prediction intervals.

\begin{table}
\centering
\small
\caption{DIC for Hiroshima models.}
\label{tab:hiroshima-model-selection}
\tt
\begin{tabular}{lr}
\multicolumn{1}{c}{Model} &
\multicolumn{1}{c}{DIC} \\
\hline
Binomial   & 3625.34 \\
RCB        & 3148.05 \\
BB         & 2984.49 \\
MixLinkJ2  & 2876.64 \\
MixLinkJ3  & 2878.01 \\
MixLinkJ4  & 2875.93 \\
\hline
\end{tabular}
\end{table}

\begin{table}
\centering
\small
\caption{Posterior summaries for Hiroshima models.}
\label{tab:hiroshima-est}
\tt
\begin{tabular}{l|rrrr}
\hline
\multicolumn{1}{l|}{Binomial} &
\multicolumn{1}{r}{mean} &
\multicolumn{1}{r}{SD} &
\multicolumn{1}{r}{2.5\%} &
\multicolumn{1}{r}{97.5\%} \\
\hline
intercept & -3.0241 & 0.0241 & -3.0695 & -2.9723 \\
$x$       &  0.9494 & 0.0244 &  0.9014 &  0.9938 \\
$x^2$     & -0.1611 & 0.0080 & -0.1762 & -0.1459 \\
\hline
\hline
\multicolumn{1}{l|}{BB} &
\multicolumn{1}{r}{mean} &
\multicolumn{1}{r}{SD} &
\multicolumn{1}{r}{2.5\%} &
\multicolumn{1}{r}{97.5\%} \\
\hline
intercept & -2.9437 & 0.0461 & -3.0368 & -2.8589 \\
$x$       &  0.8165 & 0.0395 &  0.7346 &  0.8950 \\
$x^2$     & -0.1416 & 0.0139 & -0.1681 & -0.1146 \\
$\rho$    &  0.1666 & 0.0079 &  0.1515 &  0.1823 \\
\hline
\hline
\multicolumn{1}{l|}{RCB} &
\multicolumn{1}{r}{mean} &
\multicolumn{1}{r}{SD} &
\multicolumn{1}{r}{2.5\%} &
\multicolumn{1}{r}{97.5\%} \\
\hline
intercept & -2.9761 & 0.0360 & -3.0449 & -2.9051 \\
$x$       &  0.8859 & 0.0298 &  0.8296 &  0.9430 \\
$x^2$     & -0.1817 & 0.0121 & -0.2052 & -0.1578 \\
$\rho$    &  0.1526 & 0.0081 &  0.1366 &  0.1678 \\
\hline
\hline
\multicolumn{1}{l|}{MixLinkJ2} &
\multicolumn{1}{r}{mean} &
\multicolumn{1}{r}{SD} &
\multicolumn{1}{r}{2.5\%} &
\multicolumn{1}{r}{97.5\%} \\
\hline
intercept & -3.0030 & 0.0440 & -3.0857 & -2.9110 \\
$x$       &  0.9989 & 0.0426 &  0.9155 &  1.0880 \\
$x^2$     & -0.1771 & 0.0167 & -0.2114 & -0.1450 \\
$\pi_1$   &  0.3336 & 0.0178 &  0.3004 &  0.3687 \\
$\pi_2$   &  0.6664 & 0.0178 &  0.6313 &  0.6996 \\
$\kappa$  &  1.6200 & 0.2489 &  1.2154 &  2.1959 \\
\hline
\end{tabular}
\end{table}

\begin{figure}
\begin{subfigure}[b]{0.32\textwidth}
  \centering
  \includegraphics[width=1.0\textwidth]{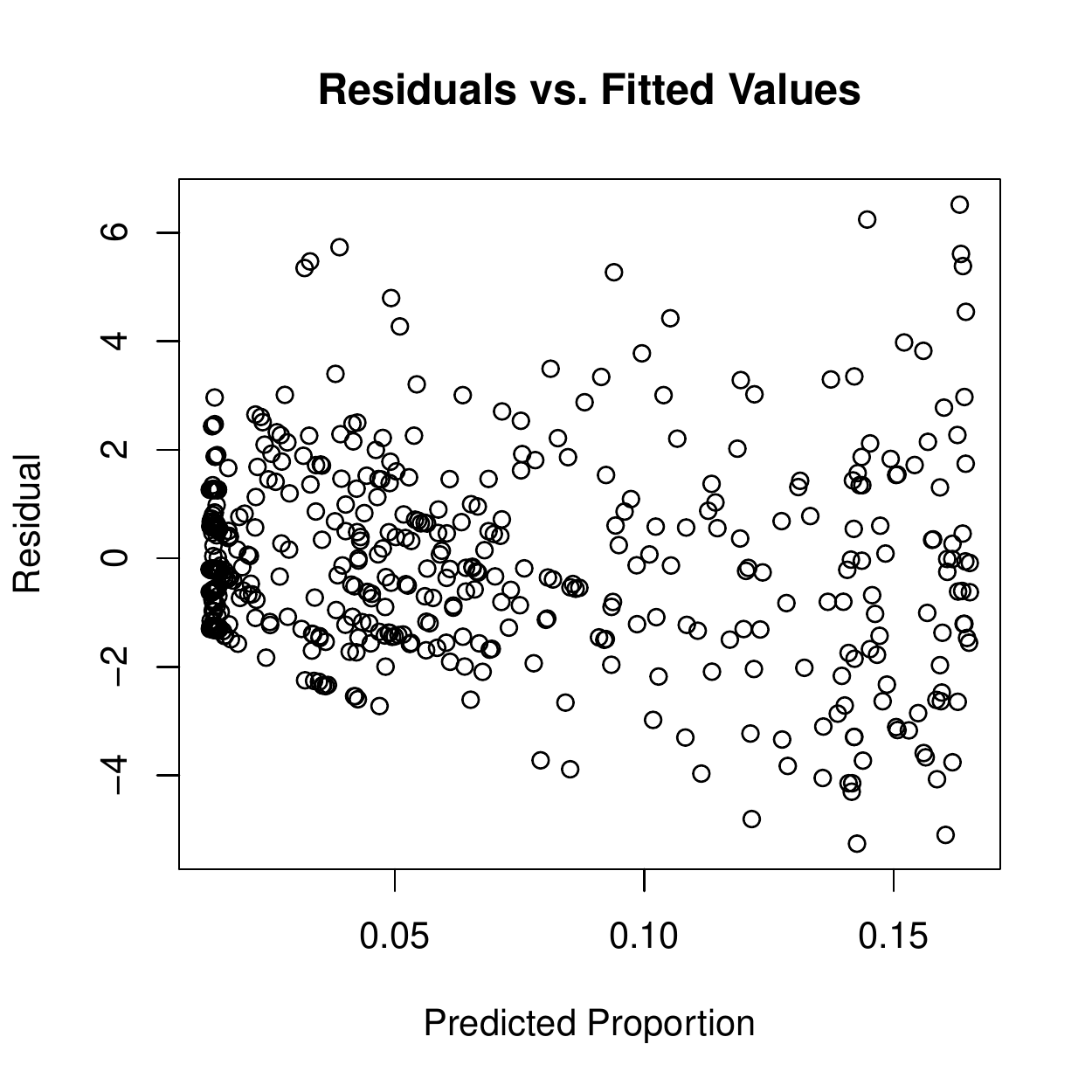}
  \caption{Binomial}
  \label{fig:pois-rqres-vsfit}
\end{subfigure}
\begin{subfigure}[b]{0.32\textwidth}
  \centering
  \includegraphics[width=1.0\textwidth]{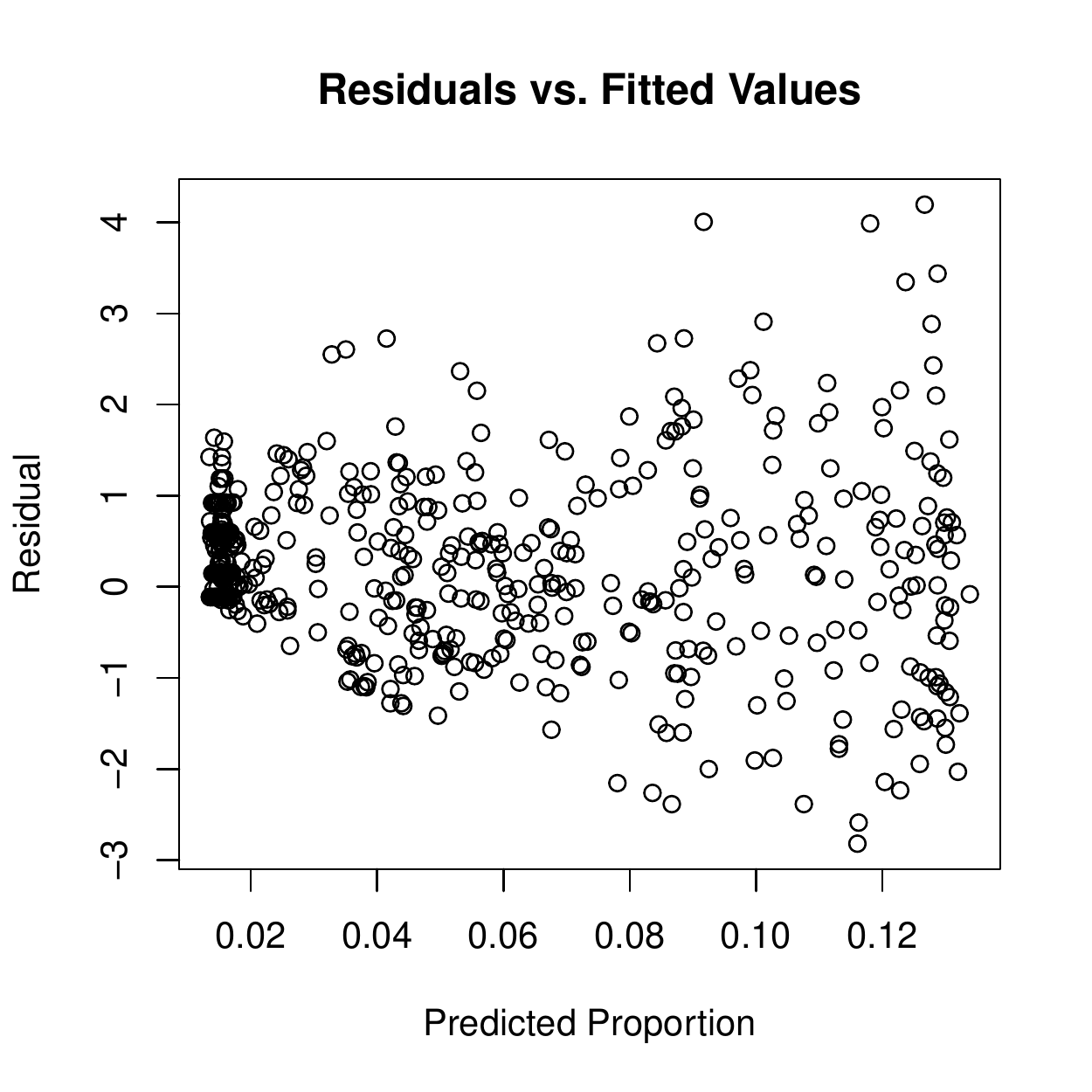}
  \caption{BB}
  \label{fig:negbin-rqres-vsfit}
\end{subfigure}
\begin{subfigure}[b]{0.32\textwidth}
  \centering
  \includegraphics[width=1.0\textwidth]{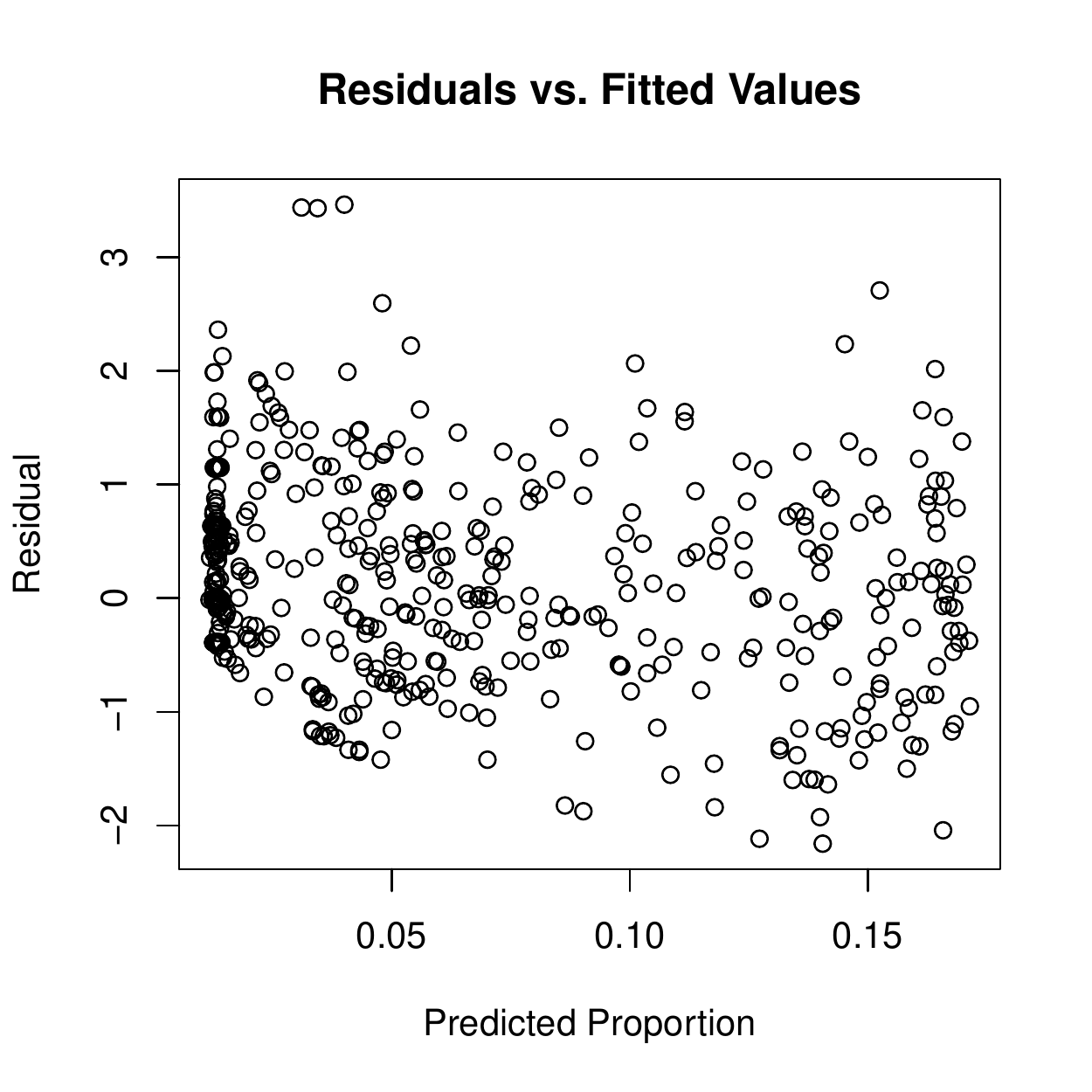}
  \caption{MixLinkJ2}
  \label{fig:mixlinkJ8-rqres-vsfit}
\end{subfigure}

\begin{subfigure}[b]{0.32\textwidth}
  \centering
  \includegraphics[width=1.0\textwidth]{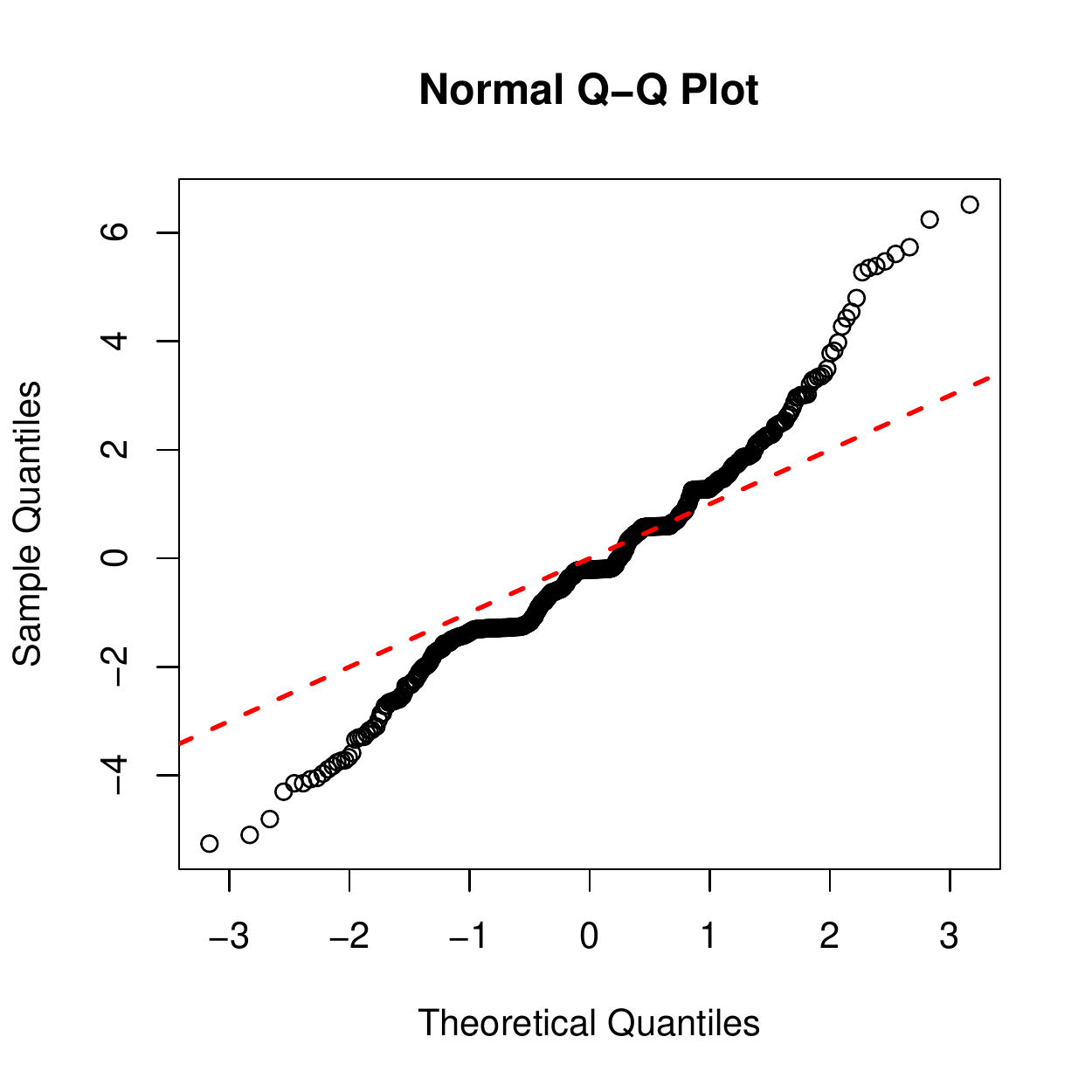}
  \caption{Binomial}
  \label{fig:binom-rqres-qq}
\end{subfigure}
\begin{subfigure}[b]{0.32\textwidth}
  \centering
  \includegraphics[width=1.0\textwidth]{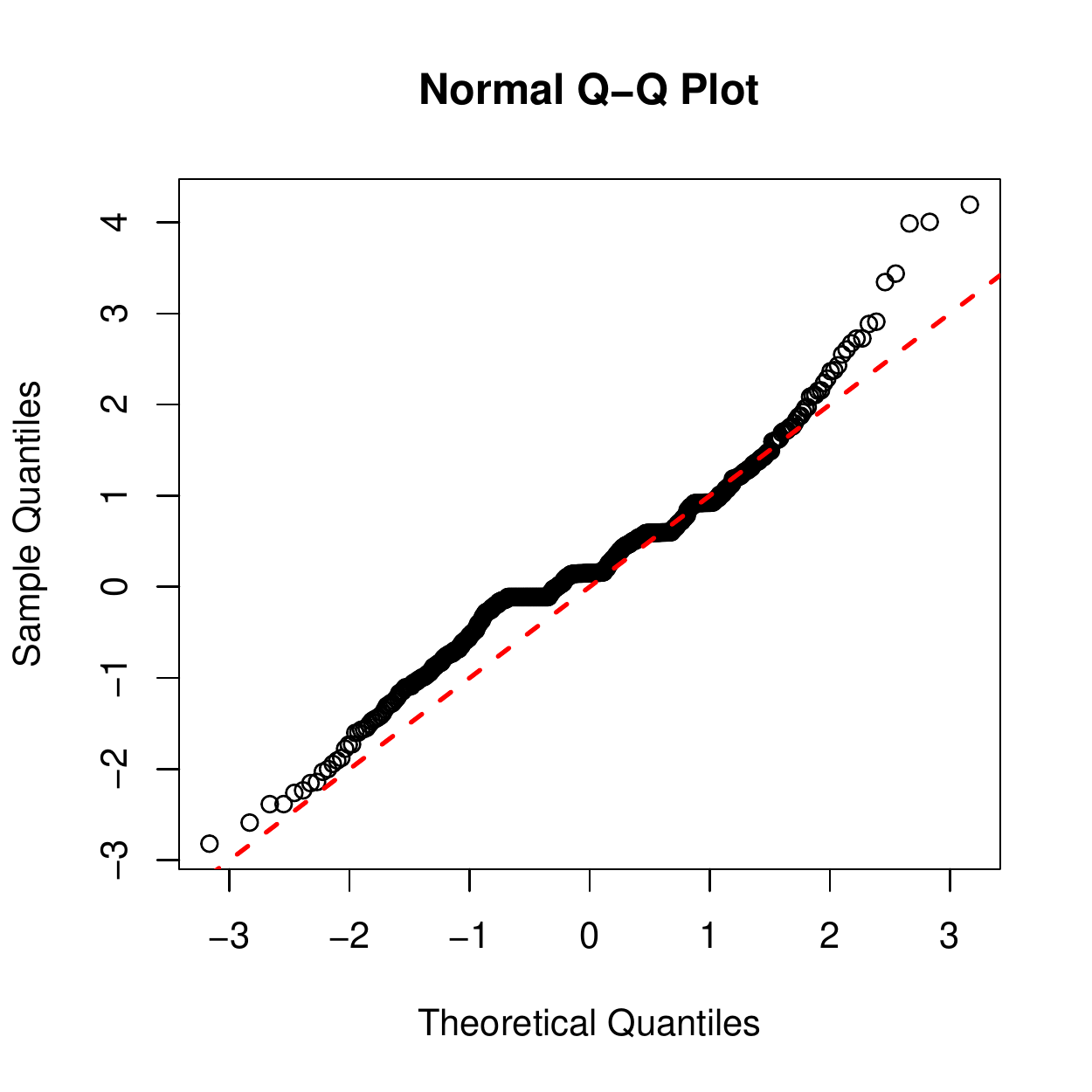}
  \caption{BB}
  \label{fig:bb-rqres-qq}
\end{subfigure}
\begin{subfigure}[b]{0.32\textwidth}
  \centering
  \includegraphics[width=1.0\textwidth]{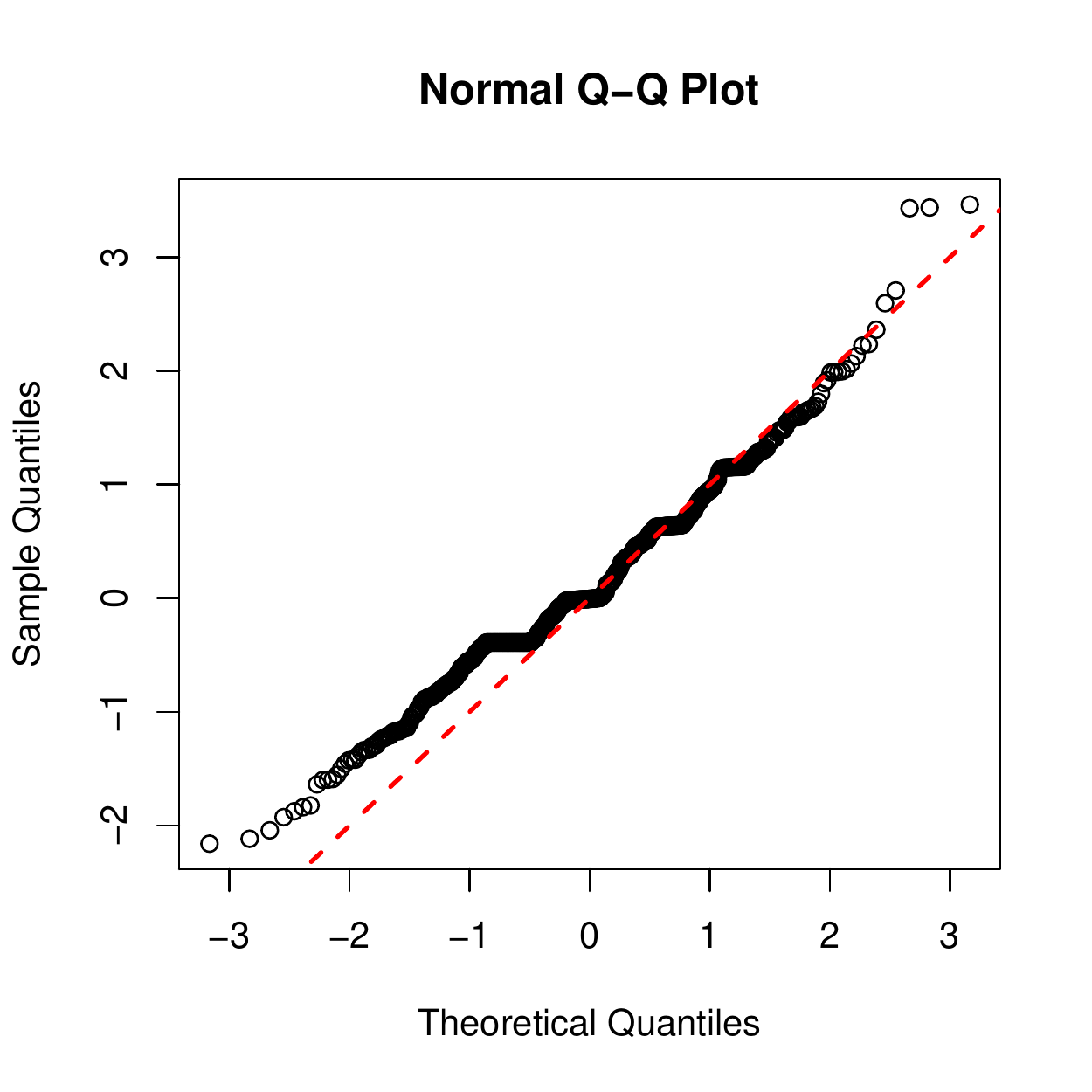}
  \caption{MixLinkJ2}
  \label{fig:mixlinkJ2-rqres-qq}
\end{subfigure}
\caption{Quantile residuals for Hiroshima models.}
\label{fig:hiroshoma-rqres}
\end{figure}

\begin{figure}
\begin{subfigure}[b]{0.32\textwidth}
  \centering
  \includegraphics[width=1.0\textwidth]{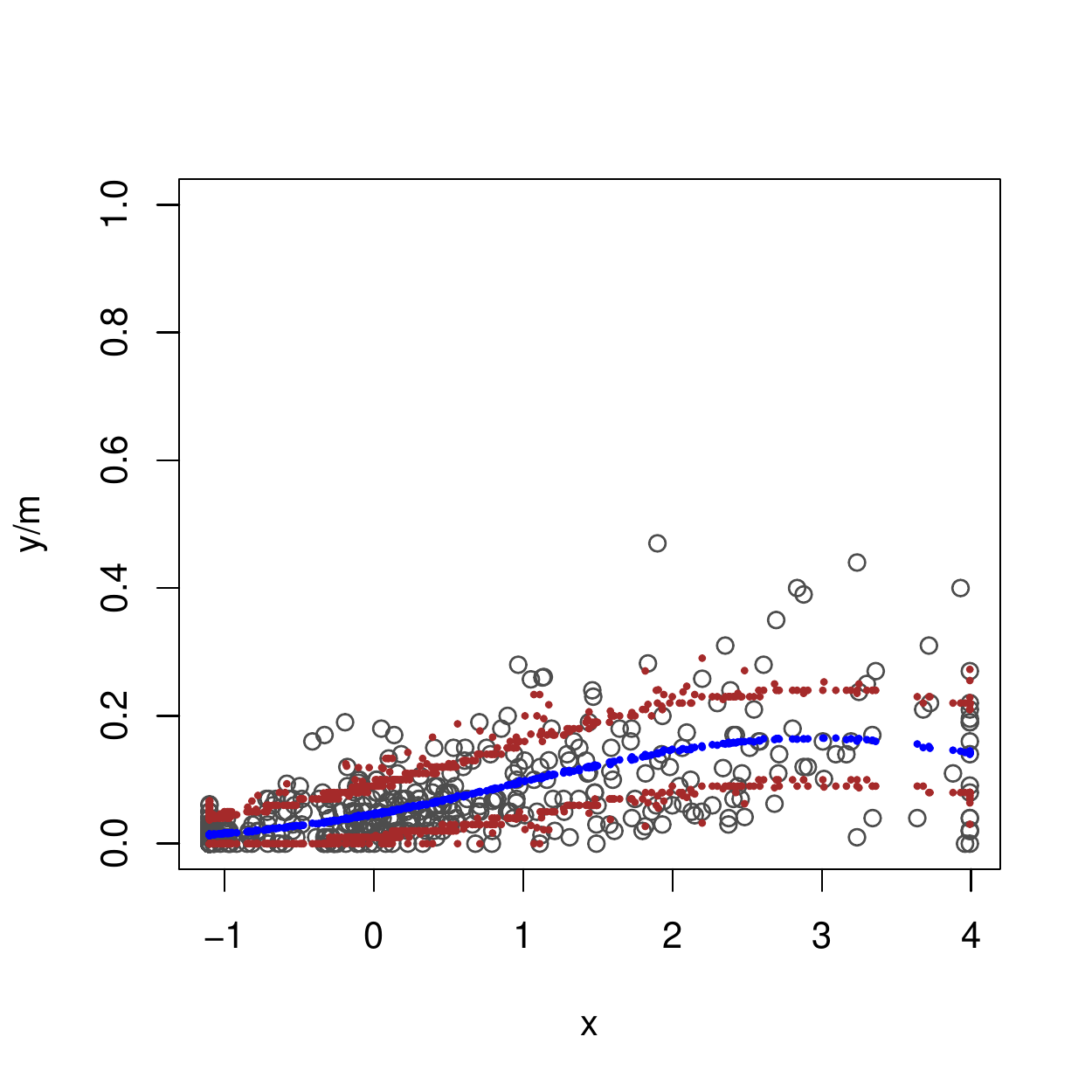}
  \caption{Binomial}
  \label{fig:hiroshima-binomial-predict}
\end{subfigure}
\begin{subfigure}[b]{0.32\textwidth}
  \centering
  \includegraphics[width=1.0\textwidth]{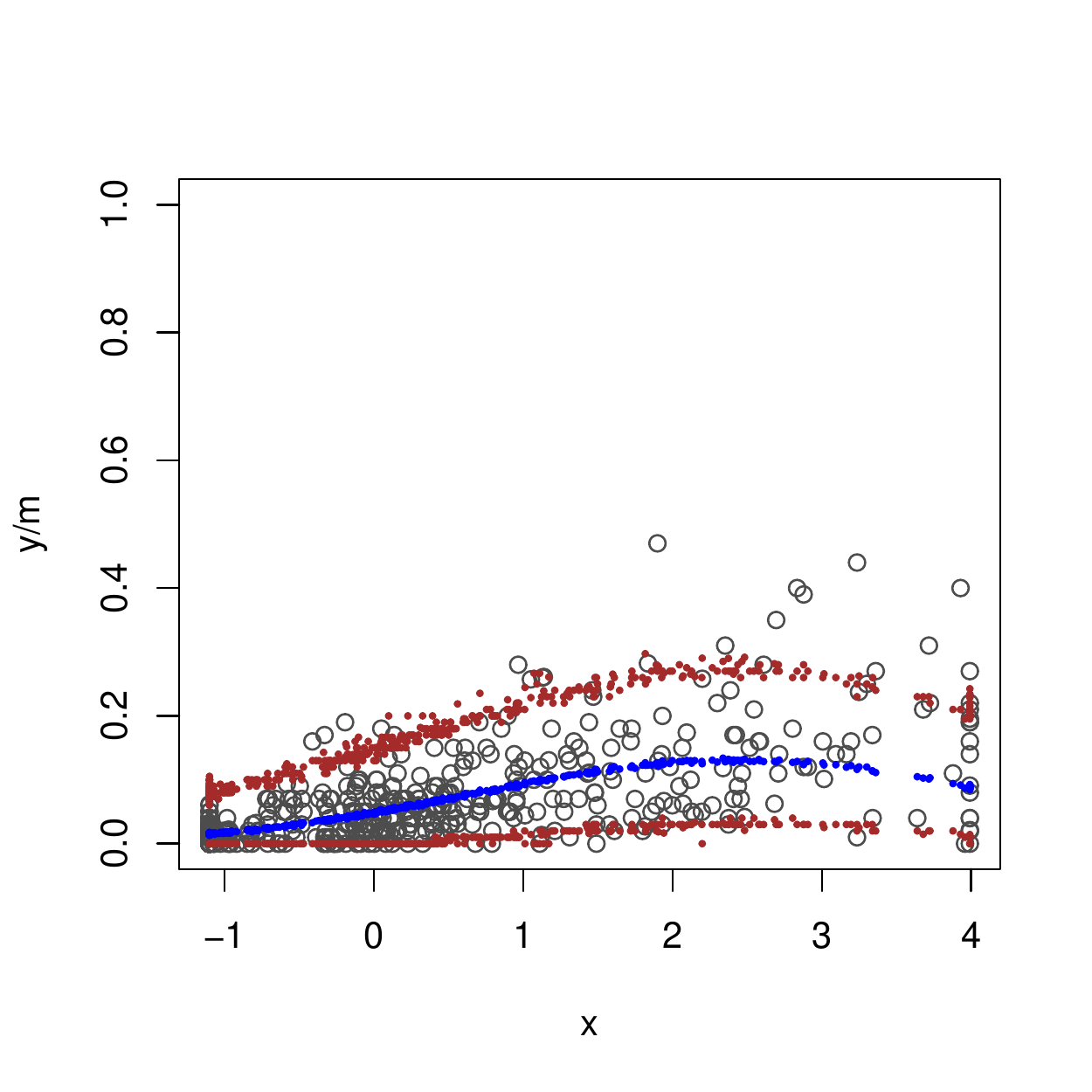}
  \caption{BB}
  \label{fig:hiroshima-bb-predict}
\end{subfigure}
\begin{subfigure}[b]{0.32\textwidth}
  \centering
  \includegraphics[width=1.0\textwidth]{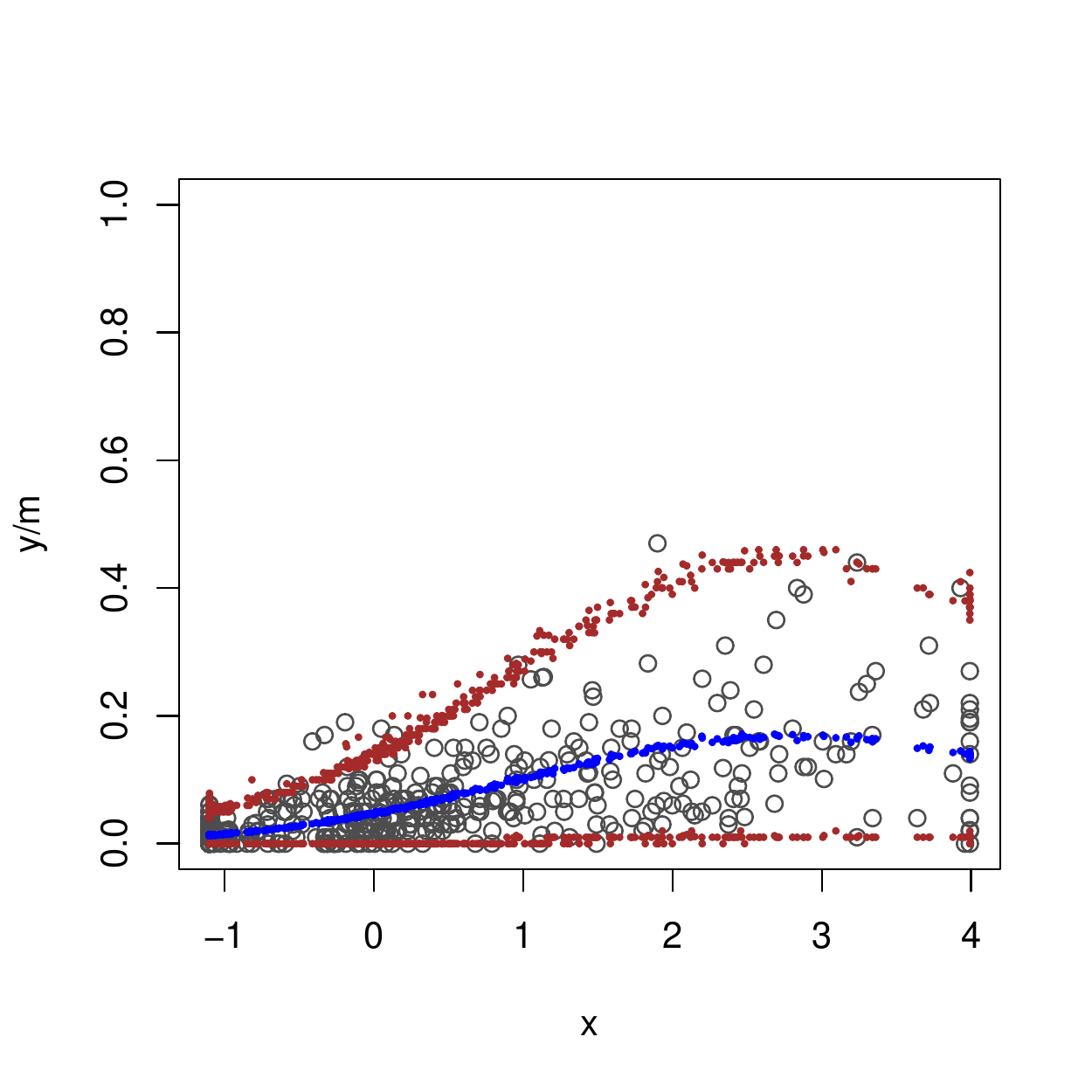}
  \caption{MixLinkJ2}
  \label{fig:hiroshima-mixlinkJ2-predict}
\end{subfigure}
\caption{Observed proportions $y_i / m_i$ vs.~$x_i$ for Hiroshima data are plotted as open circles. Smaller solid dots represent 95\% prediction intervals (upper and lower curves) and predictions (middle curve) from the respective model.}
\label{fig:hiroshima-predict}
\end{figure}

\begin{figure}
\begin{subfigure}[t]{0.32\textwidth}
  \centering
  \includegraphics[width=1.0\textwidth]{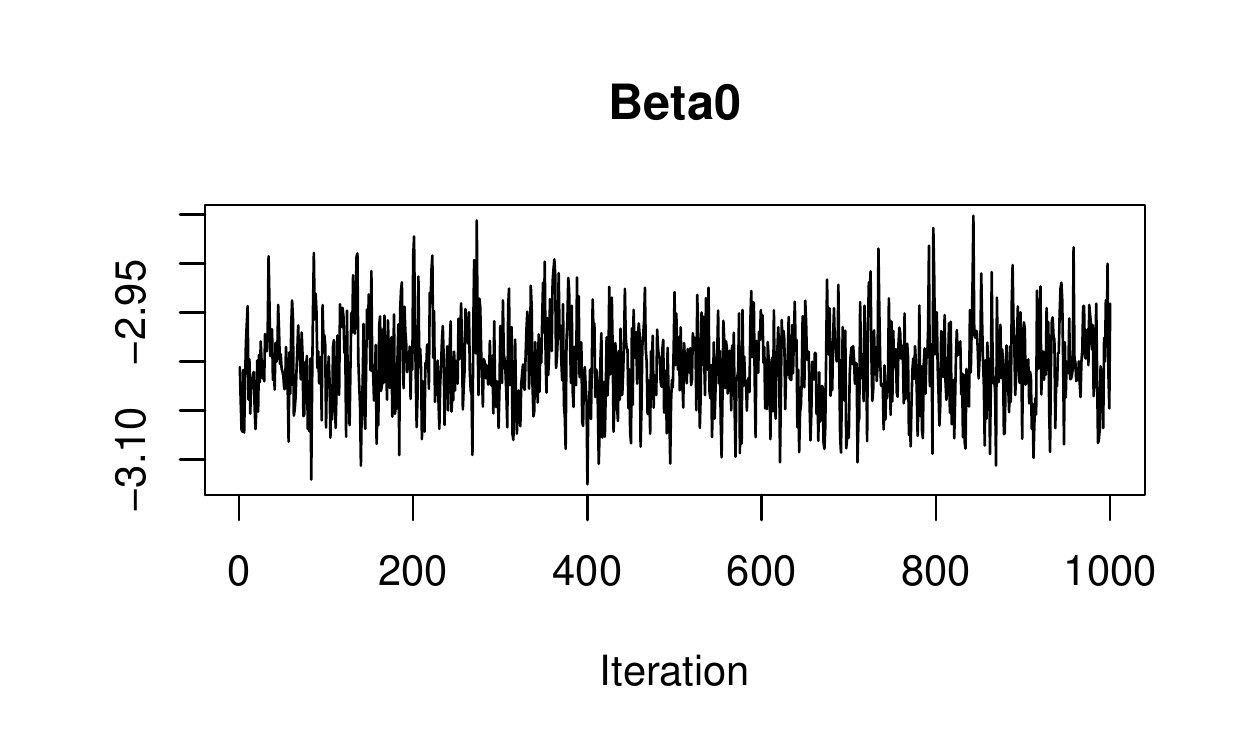}
  \label{fig:binom-trace-Beta0}
\end{subfigure}
\begin{subfigure}[t]{0.32\textwidth}
  \centering
  \includegraphics[width=1.0\textwidth]{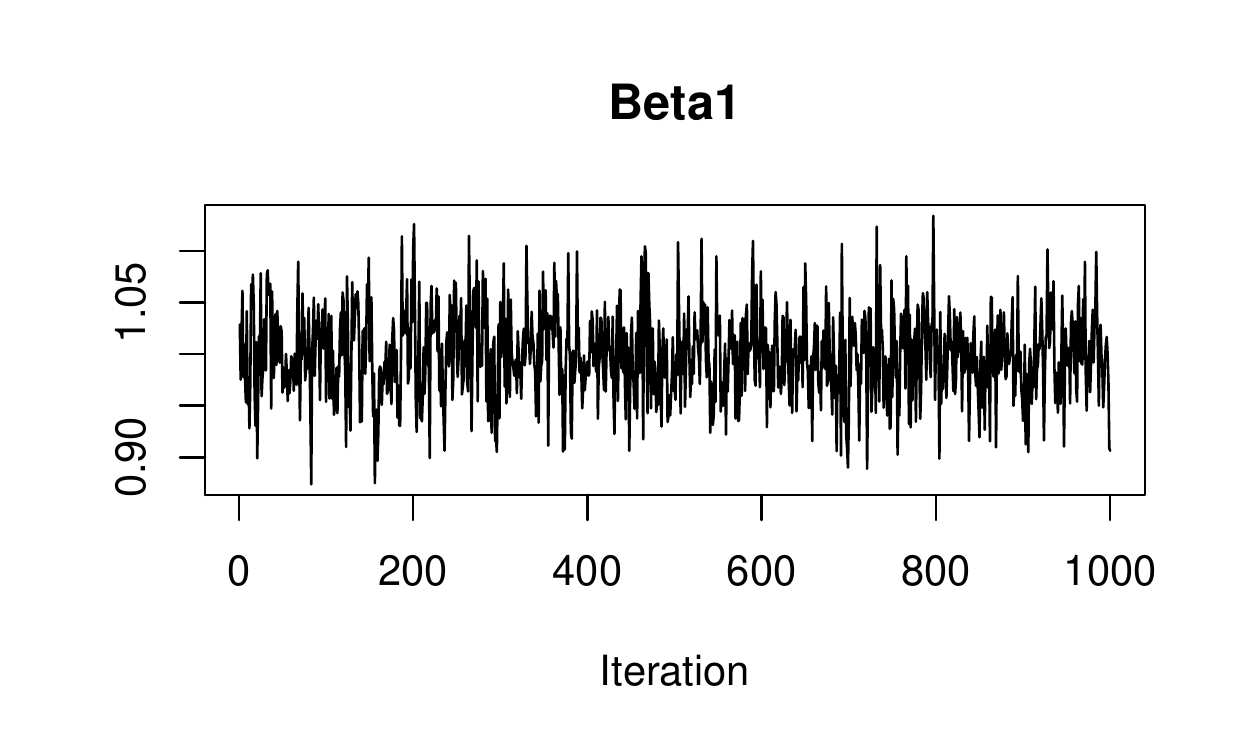}
  \label{fig:mixlinkJ2-trace-Beta1}
\end{subfigure}
\begin{subfigure}[t]{0.32\textwidth}
  \centering
  \includegraphics[width=1.0\textwidth]{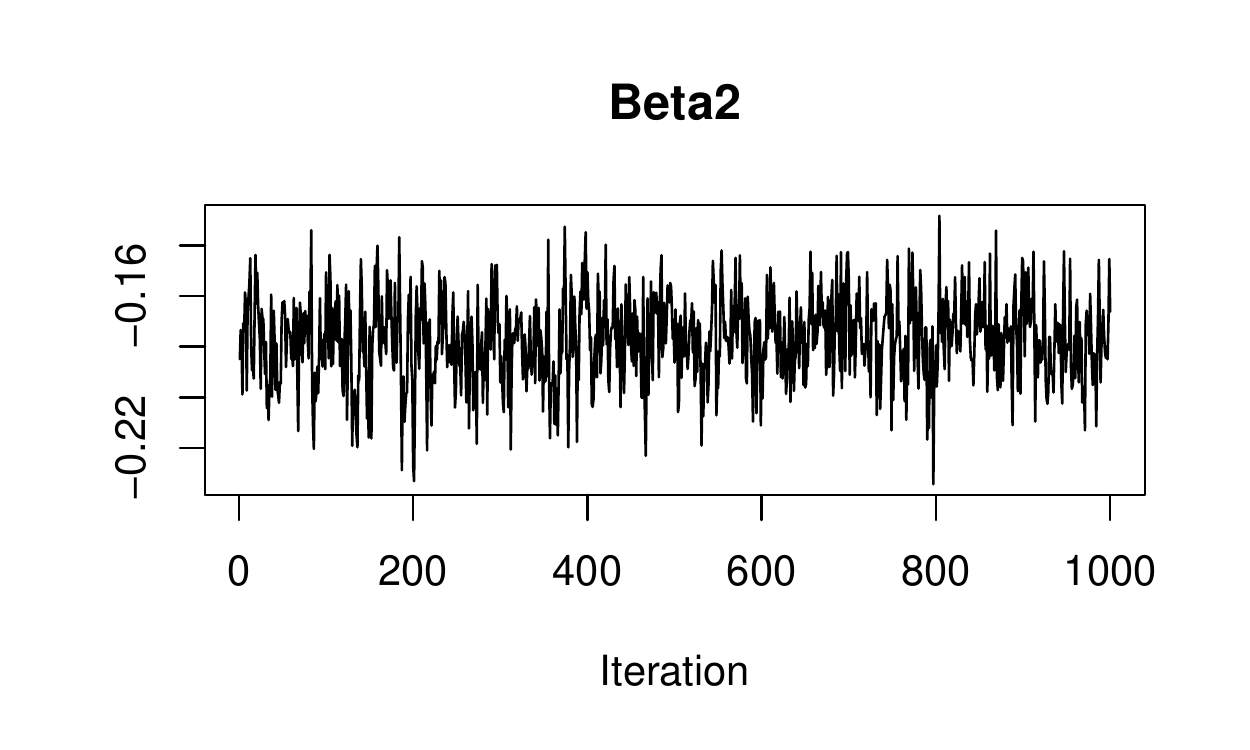}
  \label{fig:mixlinkJ2-trace-Beta2}
\end{subfigure}

\begin{subfigure}[t]{0.32\textwidth}
  \centering
  \includegraphics[width=1.0\textwidth]{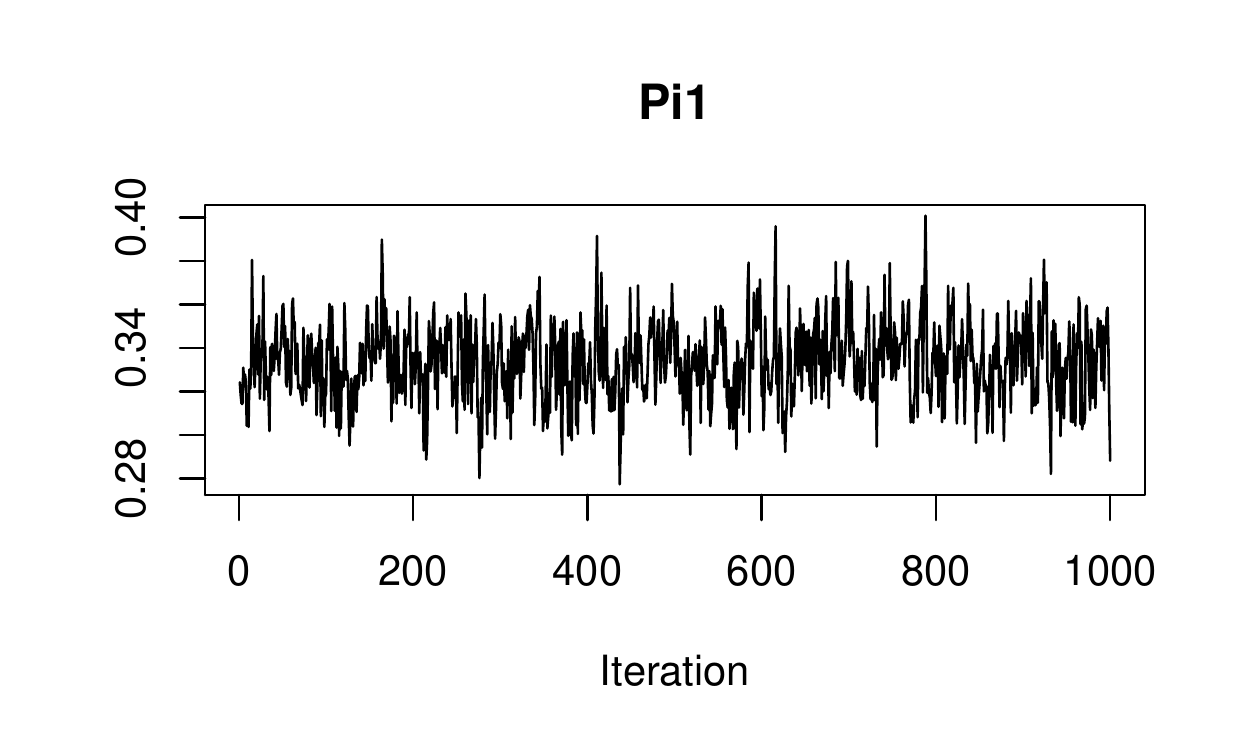}
  \label{fig:mixlinkJ2-trace-Pi1}
\end{subfigure}
\begin{subfigure}[t]{0.32\textwidth}
  \centering
  \includegraphics[width=1.0\textwidth]{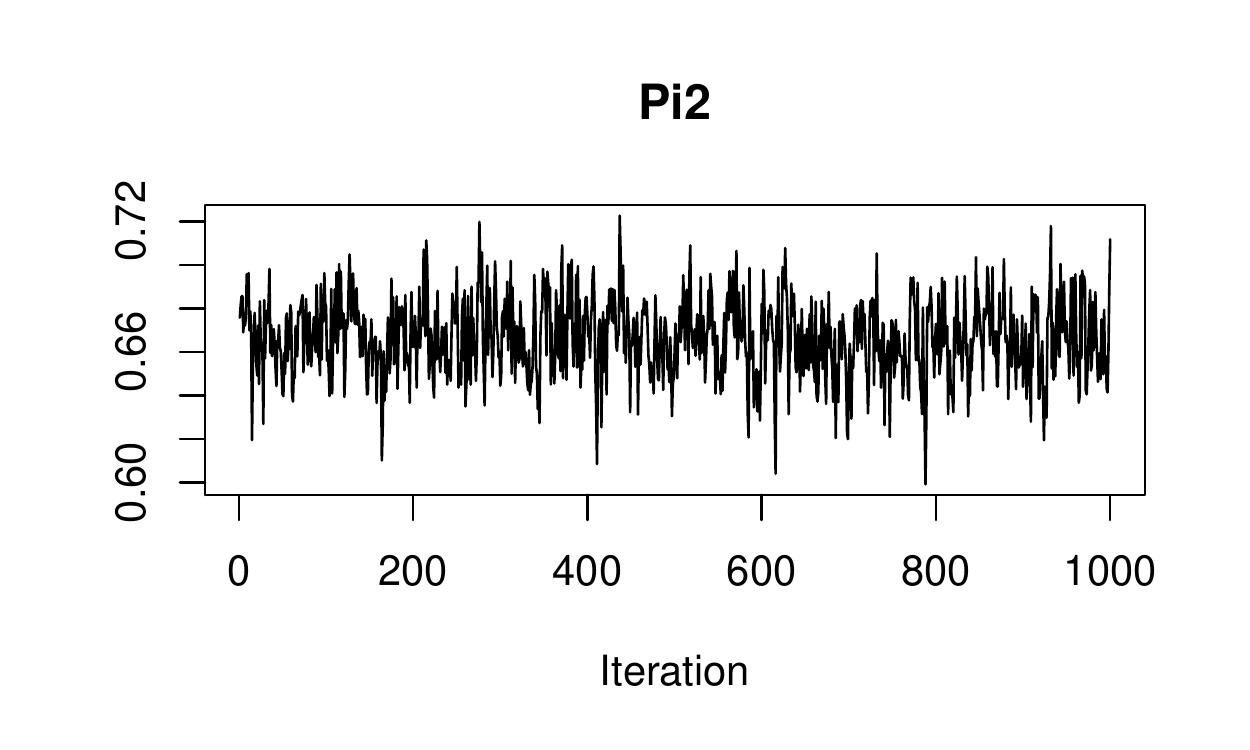}
  \label{fig:mixlinkJ2-trace-Pi2}
\end{subfigure}
\begin{subfigure}[t]{0.32\textwidth}
  \centering
  \includegraphics[width=1.0\textwidth]{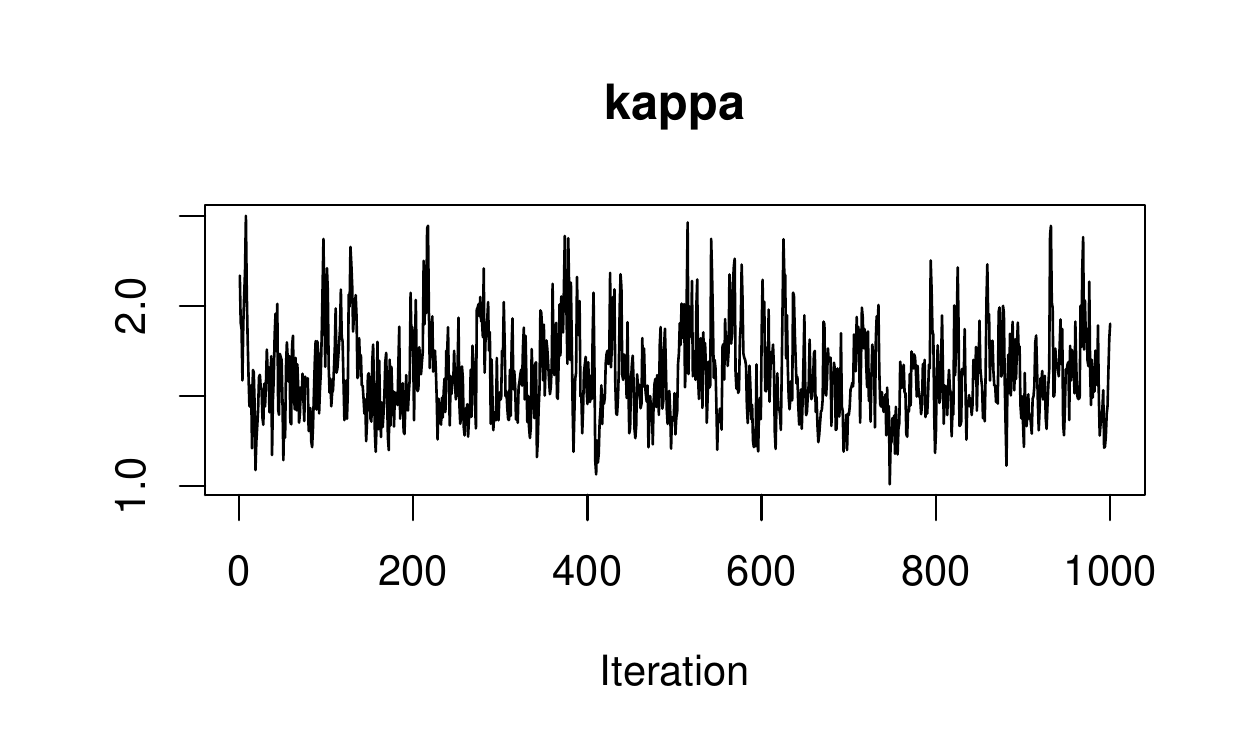}
  \label{fig:mixlinkJ2-trace-kappa}
\end{subfigure}
\caption{Trace plots for MixLinkJ2 fit to Hiroshima data.}
\label{fig:hiroshima-trace}
\end{figure}

\subsection{Arizona Medpar Data}
\label{sec:data-azpro}
The \code{azpro} data in the \code{COUNT} \code{R} package are taken from Arizona cardiovascular patient files in 1991. It contains 3,589 observations on subjects from 17 hospitals. The outcome of interest, length of hospital stay $y$, is a count. Several indicator variables are available as covariates: \code{procedure} takes values 1 for Coronary Artery Bypass Graft and 0 for Percutaneous Transluminal Coronary Angioplasty, \code{sex} is 1 for male and 0 for female, type of admission \code{admit} is 1 if emergency and 0 if elective, \code{age75} is 1 if patient's age is at least 75 and 0 otherwise, and \code{hospital} is a code to identify hospital. For this example, we consider only the 376 observations with \code{hospital = 6.5}, and take the regression function to be
\begin{align*}
\E(y_i) = \exp\{ \beta_0 + \beta_1 \cdot \text{procedure}_i + \beta_2 \cdot \text{sex}_i + \beta_3 \cdot \text{age75}_i \}.
\end{align*}
We compare count regression models based on Poisson, NegBin, and Mixture Link with $J = 2, \ldots, 8$ mixture components. All models used a simple RWMH sampler to obtain draws from the posterior. For Mixture Link models, proposals for $\btheta$ were drawn in a partitioned manner to improve mixing of the chain: a proposal for either $\bbeta$, $\vec{\pi}$, or $\kappa$ was drawn at a time, keeping other parameters fixed, and either accepted or rejected. In some cases where $J > 2$, the components of $\vec{\pi}$ were also drawn individually to further improve mixing. We assessed mixing primarily through trace plots and autocorrelation plots of the saved draws. For all models, the multivariate Normal proposal distribution was tuned by hand to achieve acceptance rates between about 15\% and 30\%. MCMC was carried out for 55,000 iterations; the first 5,000 were discarded as a burn-in sample, and 1 of every 20 remaining draws from the chain were saved.

Table~\ref{tab:azpro-model-selection} compares DIC across all fitted models. Because Poisson is a special case of NegBin, it is not surprising that the DIC of NegBin indicates a superior fit. It is interesting that the DIC of MixLink appears to improve gradually as the number of mixture components $J$ are increased. Taking $J > 2$ required additional hand-tuning of the sampler for some cases to yield acceptable diagnostics. Initial attempts to fit MixLink with $J = 9$ resulted in poor diagnostics, so these results are not shown. Figure~\ref{fig:azpro-trace} displays the trace plots for MixLinkJ8, which was selected among the seven Mixture Link models for further analysis.

We proceed by comparing the Poisson, NegBin, and MixLinkJ8 models. Table~\ref{tab:azpro-est} reports means, standard deviations, 2.5\% quantiles, and 97.5\% quantiles of each parameter computed from the posterior draws. Generally, the signs and magnitudes of the means of $\bbeta$ are similar. The standard deviations of $\bbeta$ are smallest for Poisson and largest for NegBin. The credible intervals based on the quantiles are correspondingly narrowest for Poisson and widest for NegBin. For MixLinkJ8, $\kappa$ takes on rather large values which effectively reduces $\Var(Y_i)$ over $i = 1, \ldots, n$.

Figure~\ref{fig:azpro-rqres} plots quantile residuals against predictions and also displays Q-Q plots to assess Normality. The predictions have been computed by taking means of draws from the posterior predictive distribution. Note that there are only 16 distinct values of the covariate $\vec{x}$ and observations with a common covariate are likely to obtain similar predictions. The residuals produced by MixLinkJ8 exhibit the best behavior of the three models, with the least departure from standard Normality. There is still a pattern where smaller predictions tend to have more variable residuals, which indicates that further refinement of the regression function may be needed.

Finally, Figure~\ref{fig:azpro-pi-boxplot} displays boxplots of $y$ for each of the 16 possible covariate values, with 95\% prediction intervals from both the Poisson and MixLinkJ8 models. These intervals were computed from 2.5\% and 97.5\% quantiles of the posterior predictive distribution. Intervals for the NegBin model are not shown because the upper limits are far above the range of the plots in all cases. In some cases, the Poisson intervals appear to be too narrow to capture the observed variability of the data, while MixLinkJ8 widens the intervals to reflect the variability.

\begin{table}
\centering
\small
\caption{DIC for Arizona Medpar models.}
\label{tab:azpro-model-selection}
\tt
\begin{tabular}{lr}
\multicolumn{1}{c}{Model} &
\multicolumn{1}{c}{DIC} \\
\hline
Poisson    & 2392.62 \\
NegBin     & 2125.11 \\
MixLinkJ2  & 2095.07 \\
MixLinkJ3  & 2096.85 \\
MixLinkJ4  & 2065.76 \\
MixLinkJ5  & 2061.04 \\
MixLinkJ6  & 2062.23 \\
MixLinkJ7  & 2059.73 \\
MixLinkJ8  & 2059.39 \\
\hline
\end{tabular}
\end{table}

\begin{table}
\centering
\small
\caption{Posterior summaries for Arizona Medpar models.}
\label{tab:azpro-est}
\tt
\begin{tabular}{l|rrrr}
\hline
\multicolumn{1}{l|}{Poisson} &
\multicolumn{1}{r}{mean} &
\multicolumn{1}{r}{SD} &
\multicolumn{1}{r}{2.5\%} &
\multicolumn{1}{r}{97.5\%} \\
\hline
intercept  &  1.4947 & 0.0541 &  1.3885 & 1.6012 \\
procedure  &  0.8447 & 0.0369 &  0.7713 & 0.9161 \\
sex        & -0.0292 & 0.0370 & -0.1024 & 0.0429 \\
admit      &  0.2813 & 0.0469 &  0.1896 & 0.3749 \\
age75      &  0.0366 & 0.0388 & -0.0402 & 0.1092 \\
\hline
\hline
\multicolumn{1}{l|}{NegBin} &
\multicolumn{1}{r}{mean} &
\multicolumn{1}{r}{SD} &
\multicolumn{1}{r}{2.5\%} &
\multicolumn{1}{r}{97.5\%} \\
\hline
intercept  &  1.4972 & 0.0861 &  1.3323 & 1.6698 \\
procedure  &  0.8492 & 0.0593 &  0.7333 & 0.9634 \\
sex        & -0.0422 & 0.0626 & -0.1651 & 0.0781 \\
admit      &  0.2889 & 0.0750 &  0.1391 & 0.4366 \\
age75      &  0.0335 & 0.0649 & -0.0960 & 0.1628 \\
$\kappa$   &  0.1938 & 0.0229 &  0.1519 & 0.2416 \\
\hline
\hline
\multicolumn{1}{l|}{MixLinkJ8} &
\multicolumn{1}{r}{mean} &
\multicolumn{1}{r}{SD} &
\multicolumn{1}{r}{2.5\%} &
\multicolumn{1}{r}{97.5\%} \\
\hline
intercept  &  1.5246 & 0.0759 &  1.3751 & 1.6759 \\
procedure  &  0.9451 & 0.0507 &  0.8452 & 1.0470 \\
sex        & -0.0974 & 0.0526 & -0.2013 & 0.0035 \\
admit      &  0.2578 & 0.0627 &  0.1390 & 0.3858 \\
age75      &  0.0849 & 0.0548 & -0.0266 & 0.1891 \\
$\pi_1$    &  0.0393 & 0.0055 &  0.0280 & 0.0495 \\
$\pi_2$    &  0.0631 & 0.0113 &  0.0458 & 0.0931 \\
$\pi_3$    &  0.1145 & 0.0158 &  0.0775 & 0.1376 \\
$\pi_4$    &  0.1364 & 0.0085 &  0.1181 & 0.1512 \\
$\pi_5$    &  0.1472 & 0.0069 &  0.1338 & 0.1609 \\
$\pi_6$    &  0.1562 & 0.0071 &  0.1431 & 0.1707 \\
$\pi_7$    &  0.1654 & 0.0081 &  0.1515 & 0.1828 \\
$\pi_8$    &  0.1779 & 0.0103 &  0.1601 & 0.2008 \\
$\kappa$   & 17.0029 & 3.5466 & 11.0783 & 24.6940 \\
\hline
\end{tabular}
\end{table}

\begin{figure}
\begin{subfigure}[b]{0.32\textwidth}
  \centering
  \includegraphics[width=1.0\textwidth]{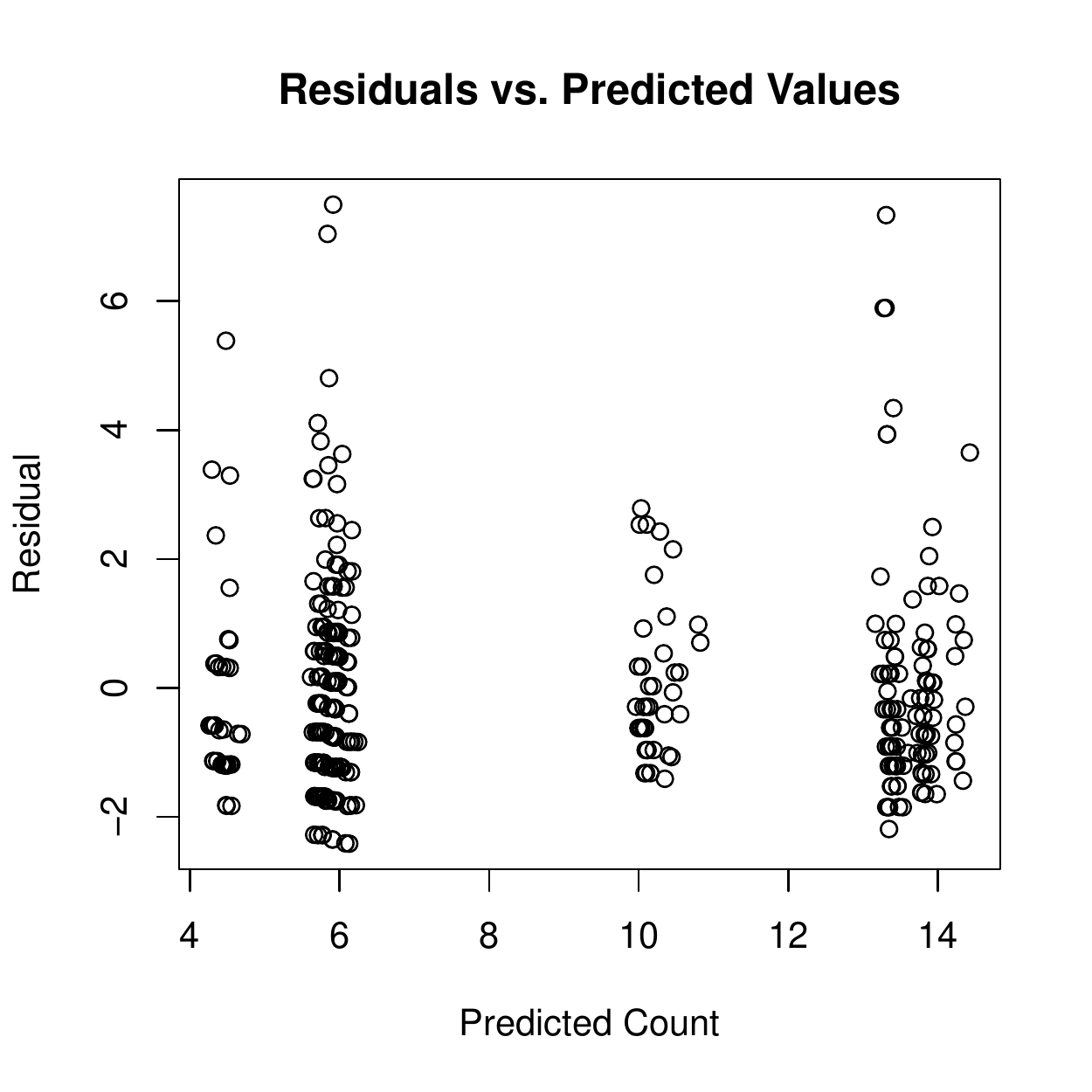}
  \caption{Poisson}
  \label{fig:azpro-pois-rqres-vsfit}
\end{subfigure}
\begin{subfigure}[b]{0.32\textwidth}
  \centering
  \includegraphics[width=1.0\textwidth]{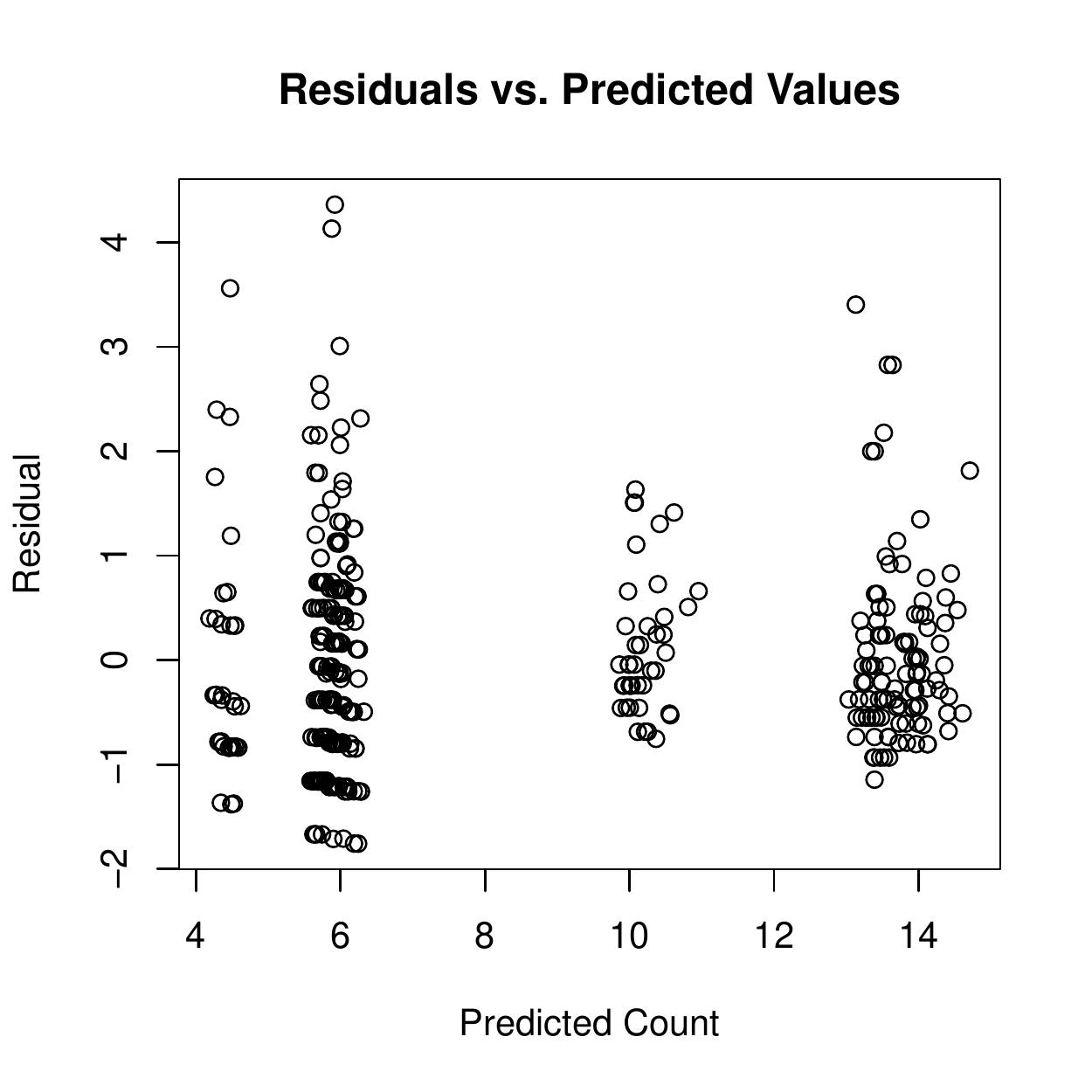}
  \caption{NegBin}
  \label{fig:azpro-negbin-rqres-vsfit}
\end{subfigure}
\begin{subfigure}[b]{0.32\textwidth}
  \centering
  \includegraphics[width=1.0\textwidth]{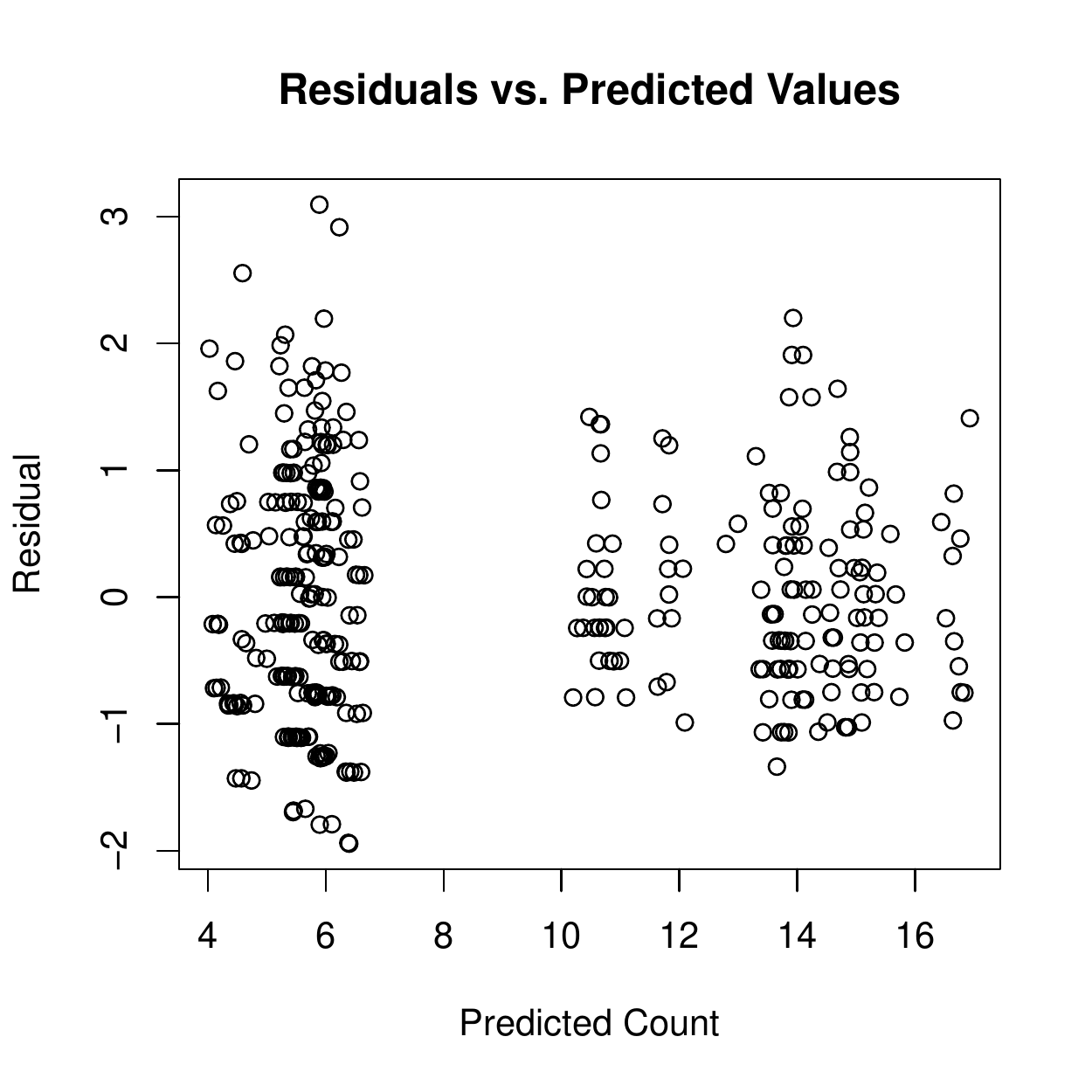}
  \caption{MixLinkJ8}
  \label{fig:azpro-mixlinkJ8-rqres-vsfit}
\end{subfigure}

\begin{subfigure}[b]{0.32\textwidth}
  \centering
  \includegraphics[width=1.0\textwidth]{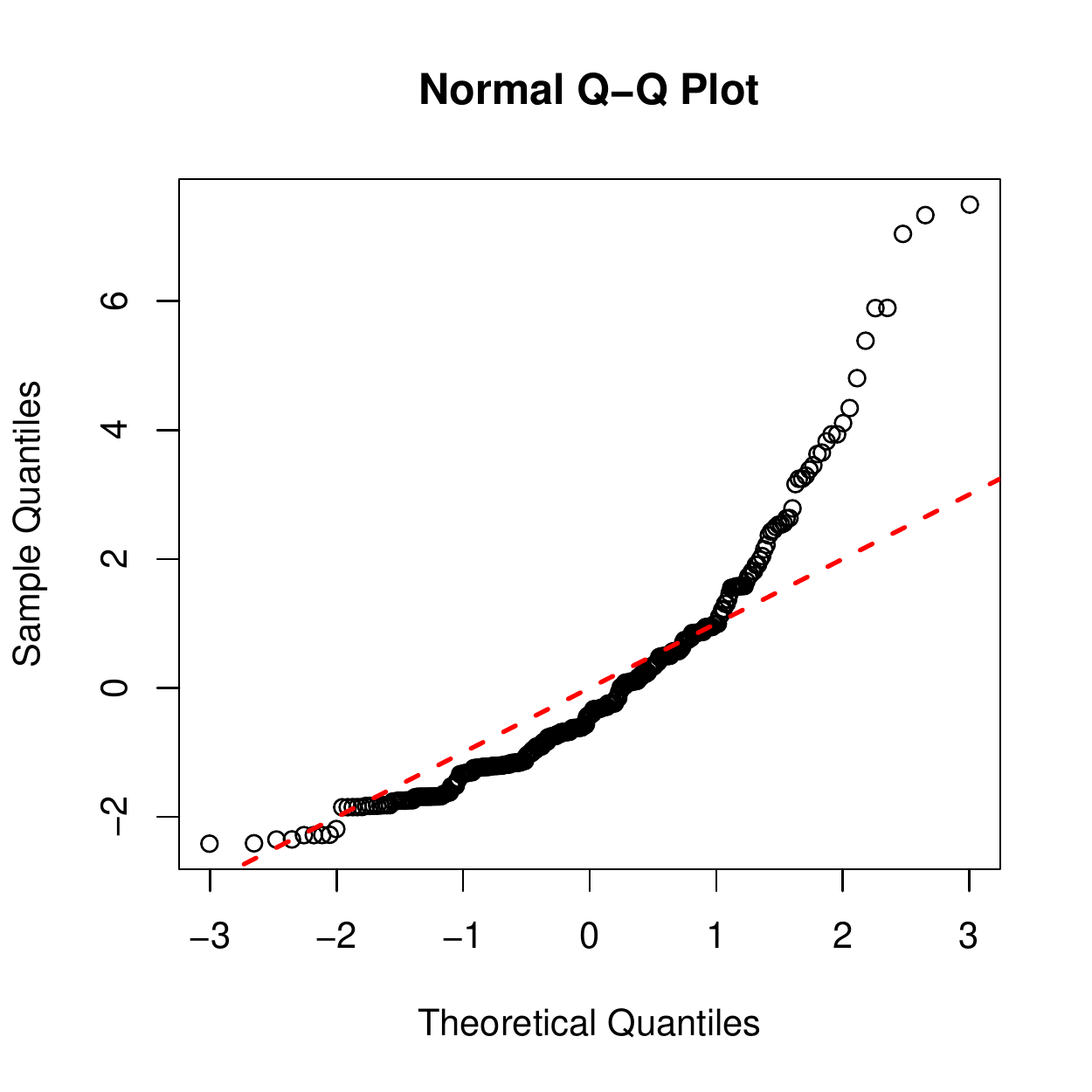}
  \caption{Poisson}
  \label{fig:azpro-pois-rqres-qq}
\end{subfigure}
\begin{subfigure}[b]{0.32\textwidth}
  \centering
  \includegraphics[width=1.0\textwidth]{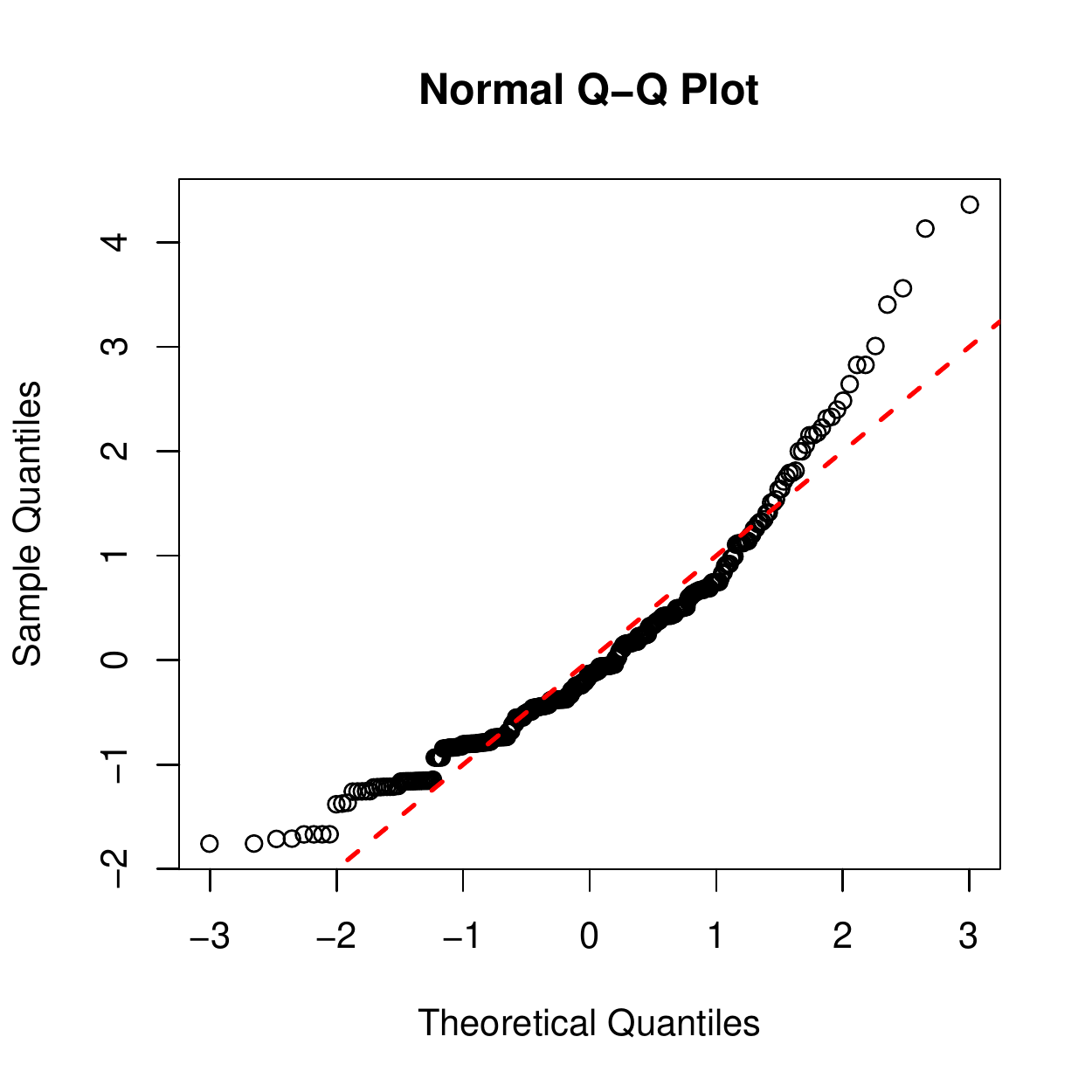}
  \caption{NegBin}
  \label{fig:azpro-negbin-rqres-qq}
\end{subfigure}
\begin{subfigure}[b]{0.32\textwidth}
  \centering
  \includegraphics[width=1.0\textwidth]{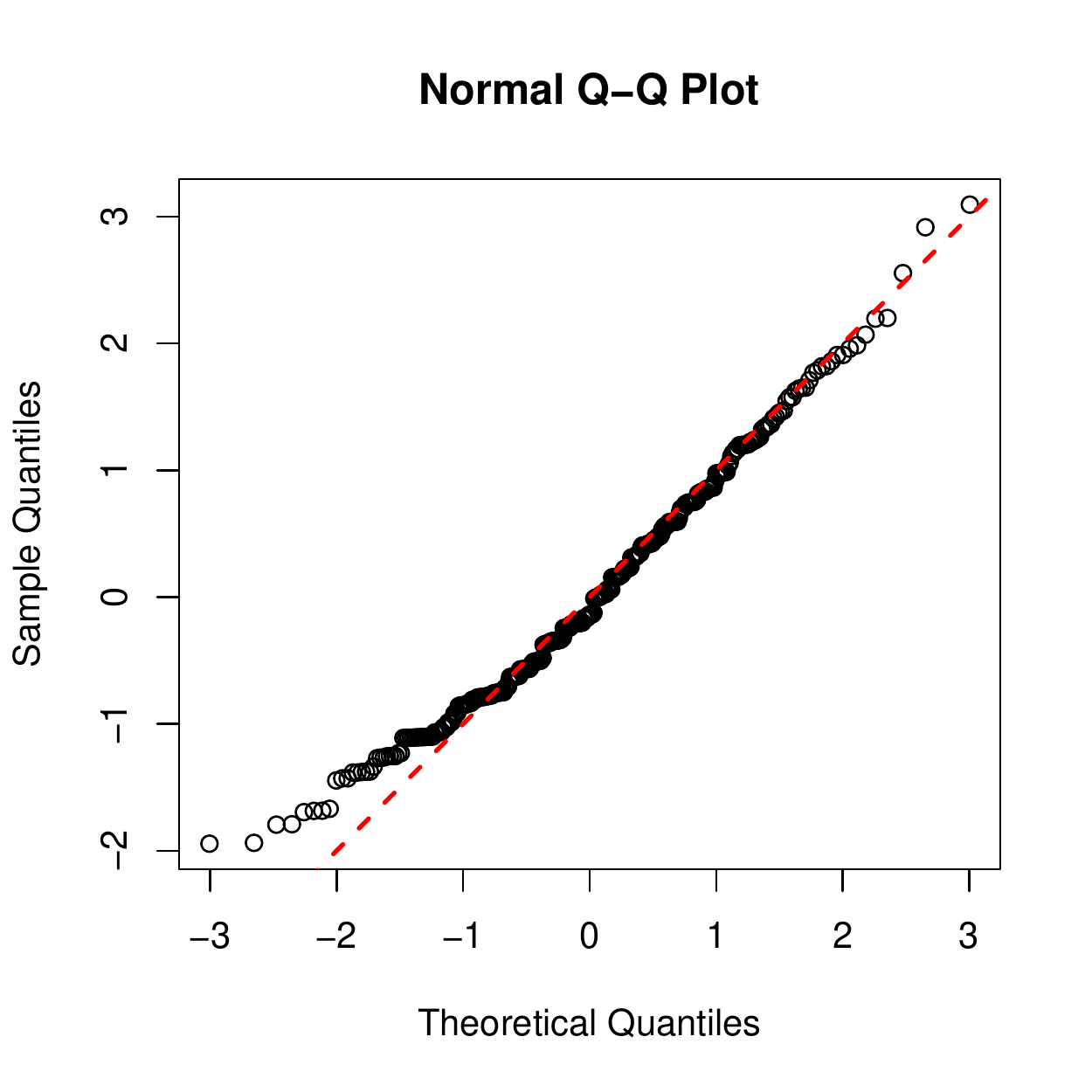}
  \caption{MixLinkJ8}
  \label{fig:azpro-mixlinkJ8-rqres-qq}
\end{subfigure}
\caption{Quantile residuals for Arizona Medpar data.}
\label{fig:azpro-rqres}
\end{figure}

\begin{figure}
\centering
\begin{subfigure}[b]{0.49\textwidth}
  \centering
  \includegraphics[width=1.0\textwidth]{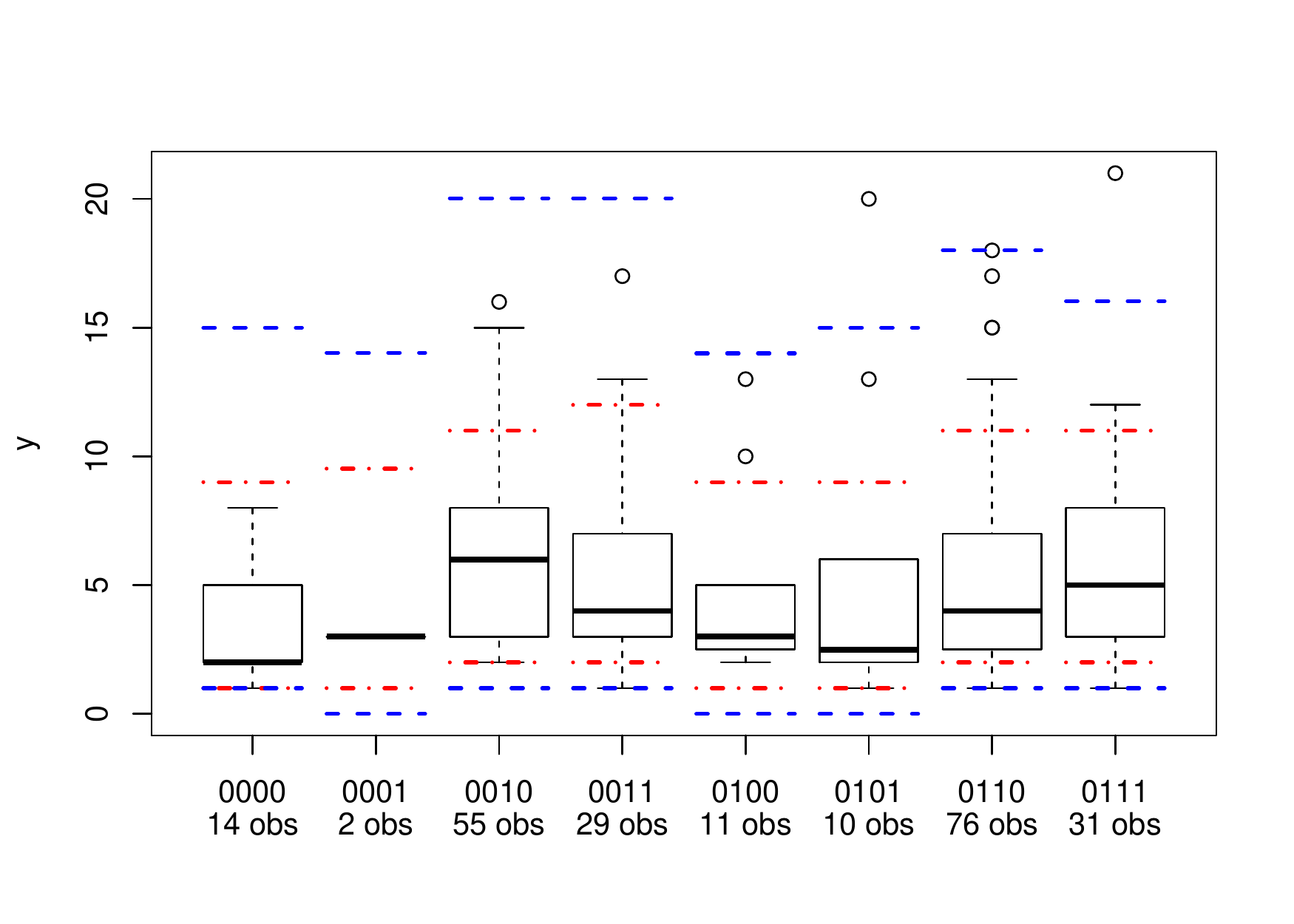}
  \label{fig:azpro-pi-boxplot1}
\end{subfigure}
\begin{subfigure}[b]{0.49\textwidth}
  \centering
  \includegraphics[width=1.0\textwidth]{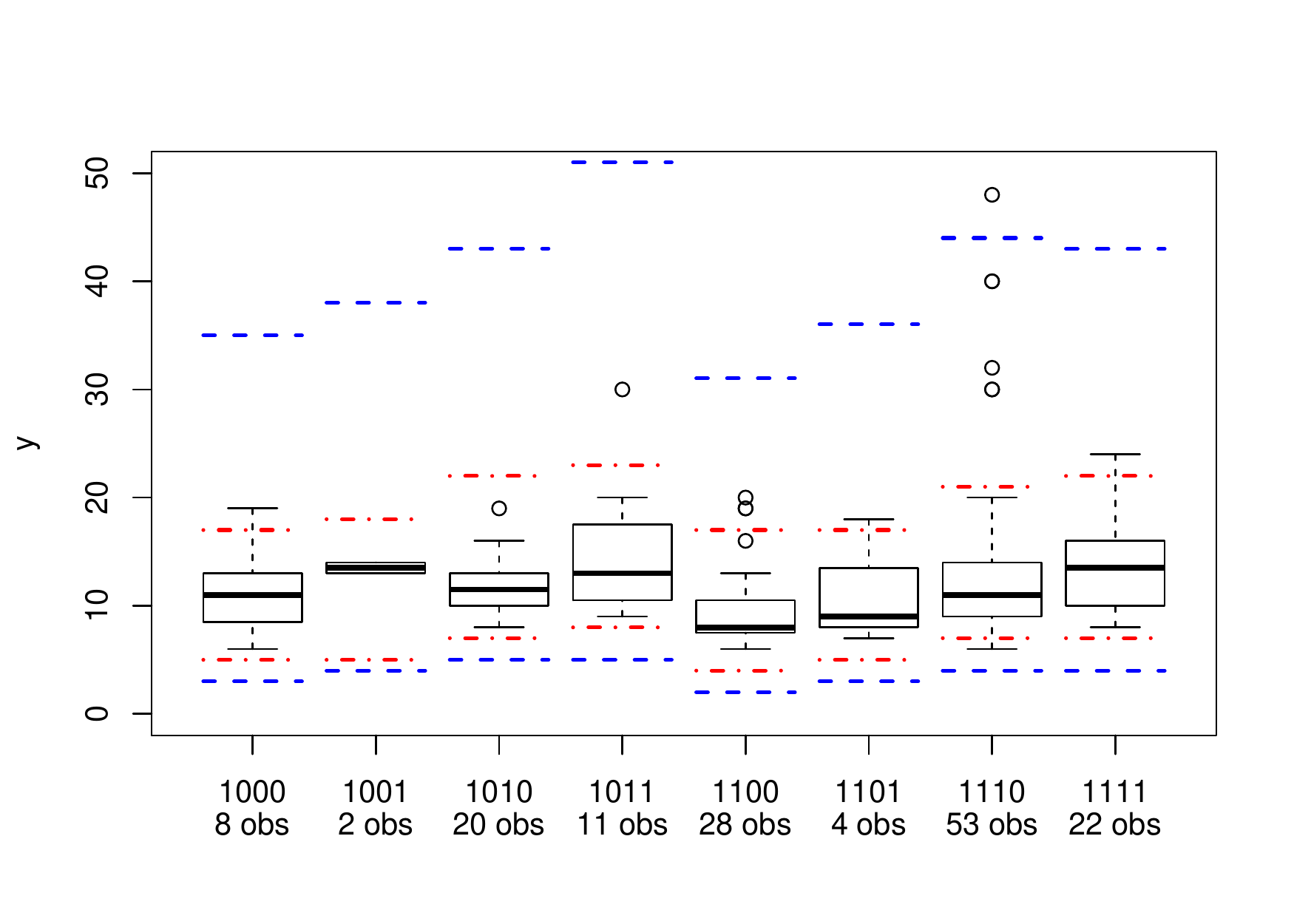}
  \label{fig:azpro-pi-boxplot2}
\end{subfigure}
\caption{Boxplots of observed $y_i$ for each of the 16 possible covariate values in the Arizona Medpar data. Covariate values are displayed as a string representing (procedure, sex, admit, age75). For example, ``1010'' represents $\text{procedure} = \text{admit} = 1$ and $\text{sex} = \text{age75} = 0$. Red dash-dot lines represent 95\% prediction limits from Poisson and blue dashed lines are from MixLink.}
\label{fig:azpro-pi-boxplot}
\end{figure}

\begin{figure}
\begin{subfigure}[t]{0.32\textwidth}
  \centering
  \includegraphics[width=1.0\textwidth]{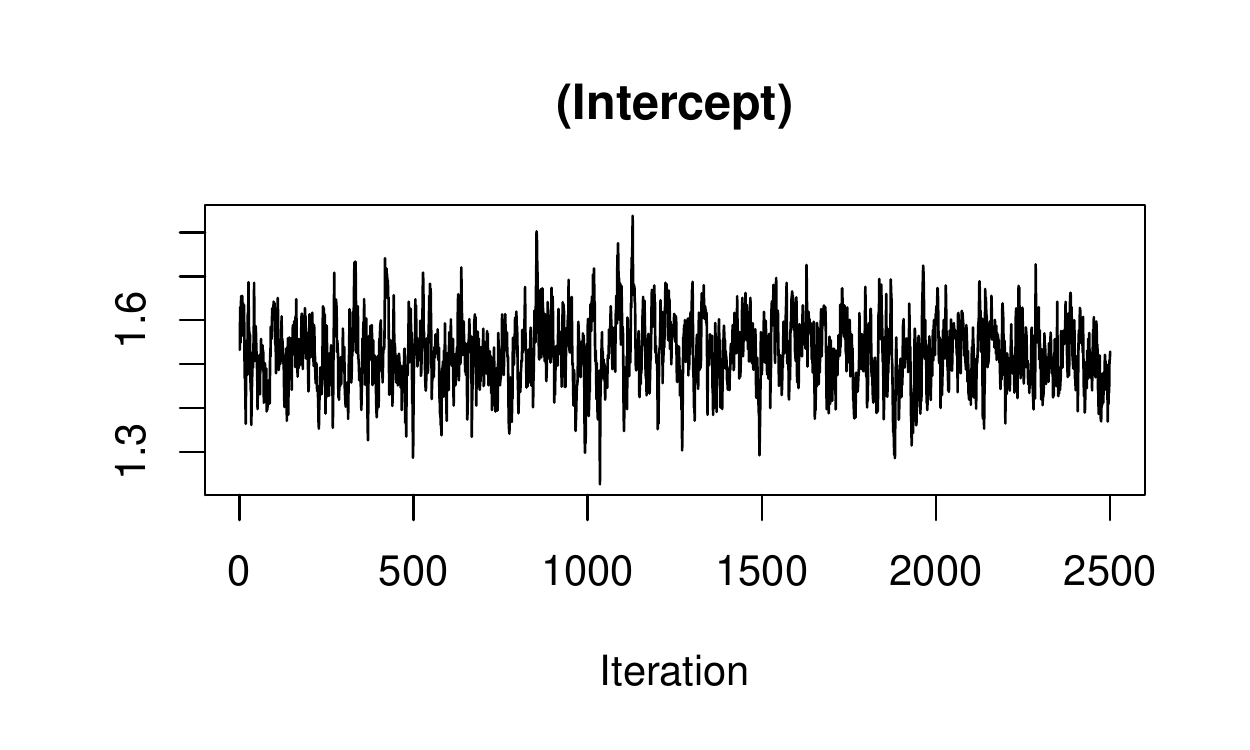}
  \label{fig:azpro-trace-Beta0}
\end{subfigure}
\begin{subfigure}[t]{0.32\textwidth}
  \centering
  \includegraphics[width=1.0\textwidth]{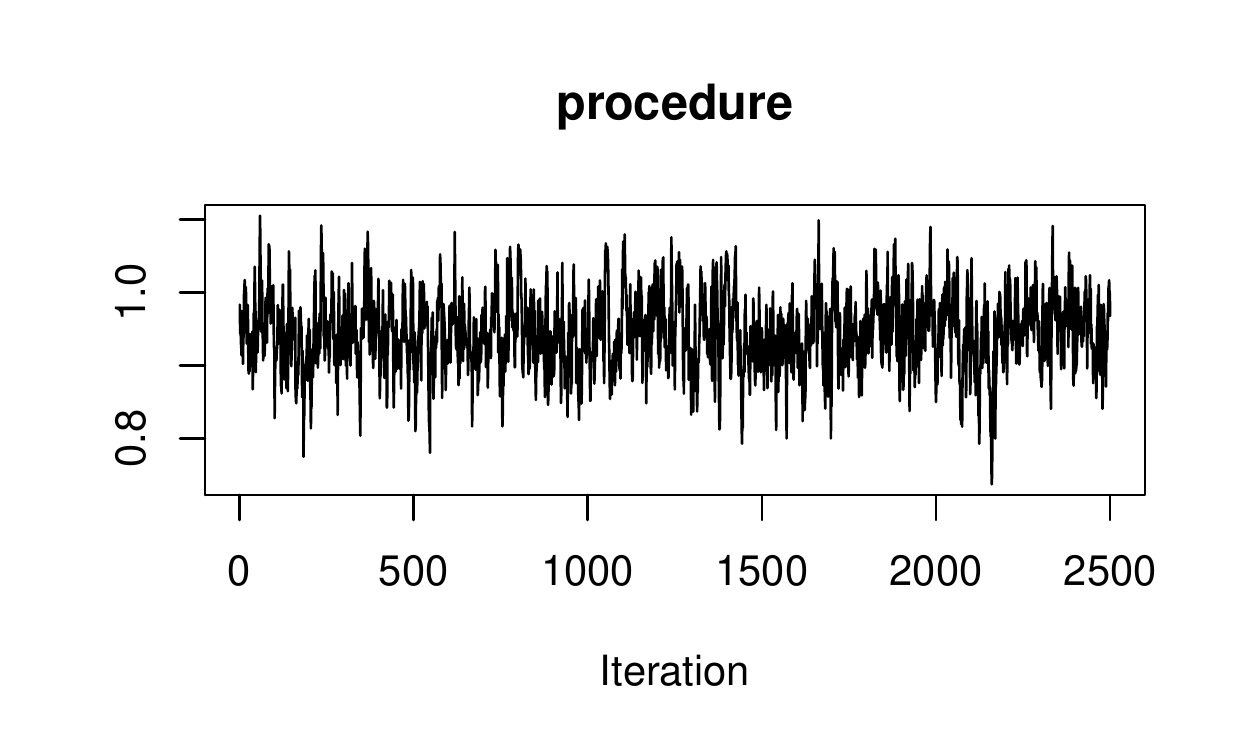}
  \label{fig:azpro-trace-Beta1}
\end{subfigure}
\begin{subfigure}[t]{0.32\textwidth}
  \centering
  \includegraphics[width=1.0\textwidth]{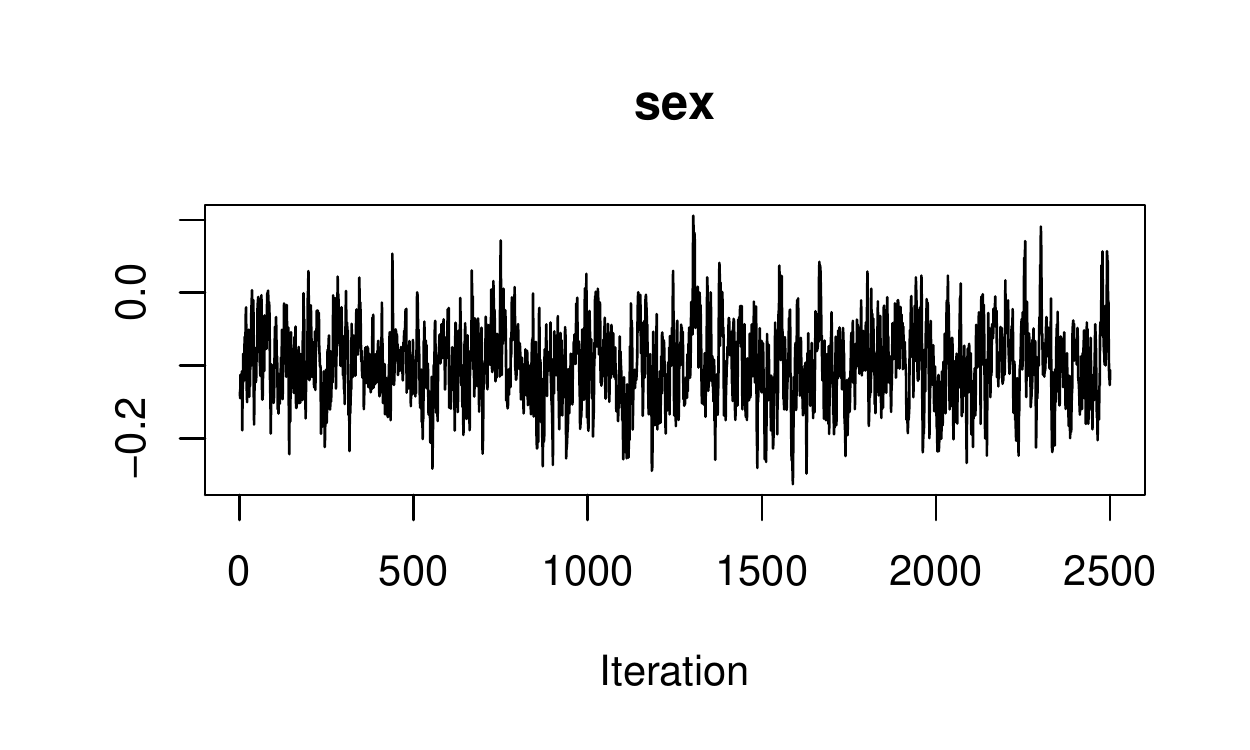}
  \label{fig:azpro-trace-Beta2}
\end{subfigure}
\begin{subfigure}[t]{0.32\textwidth}
  \centering
  \includegraphics[width=1.0\textwidth]{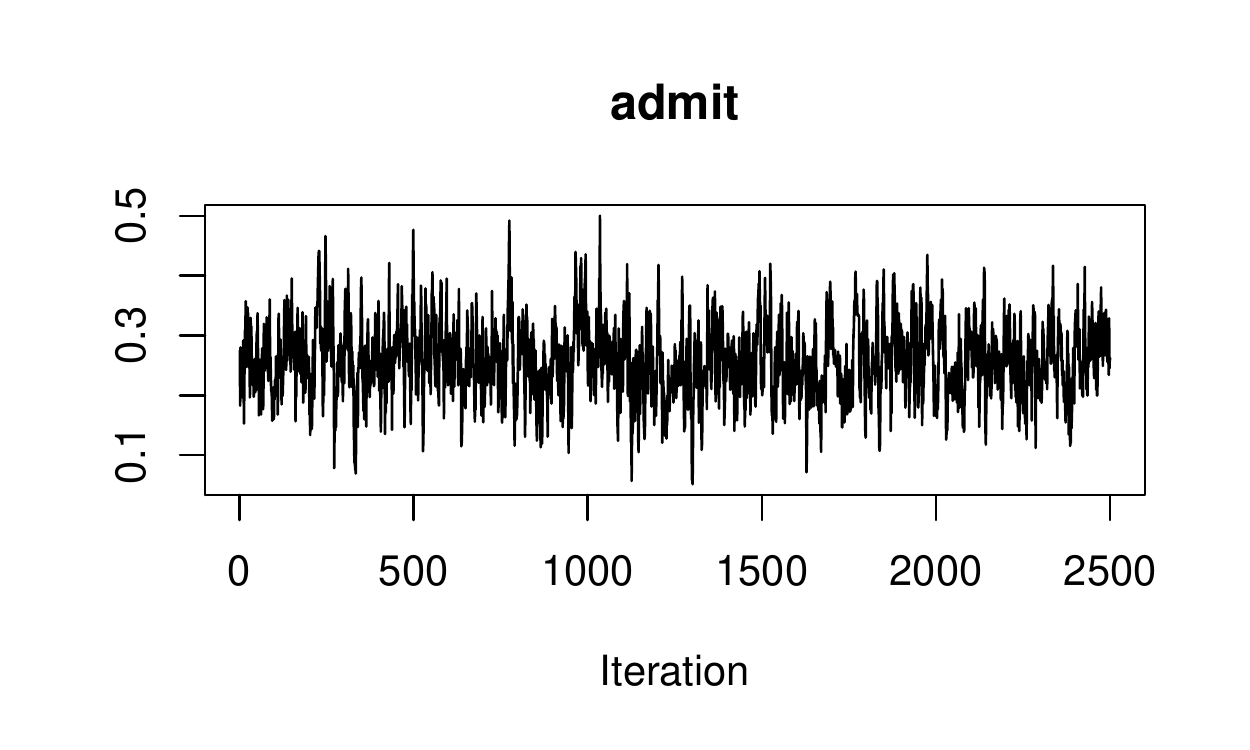}
  \label{fig:azpro-trace-Beta3}
\end{subfigure}
\begin{subfigure}[t]{0.32\textwidth}
  \centering
  \includegraphics[width=1.0\textwidth]{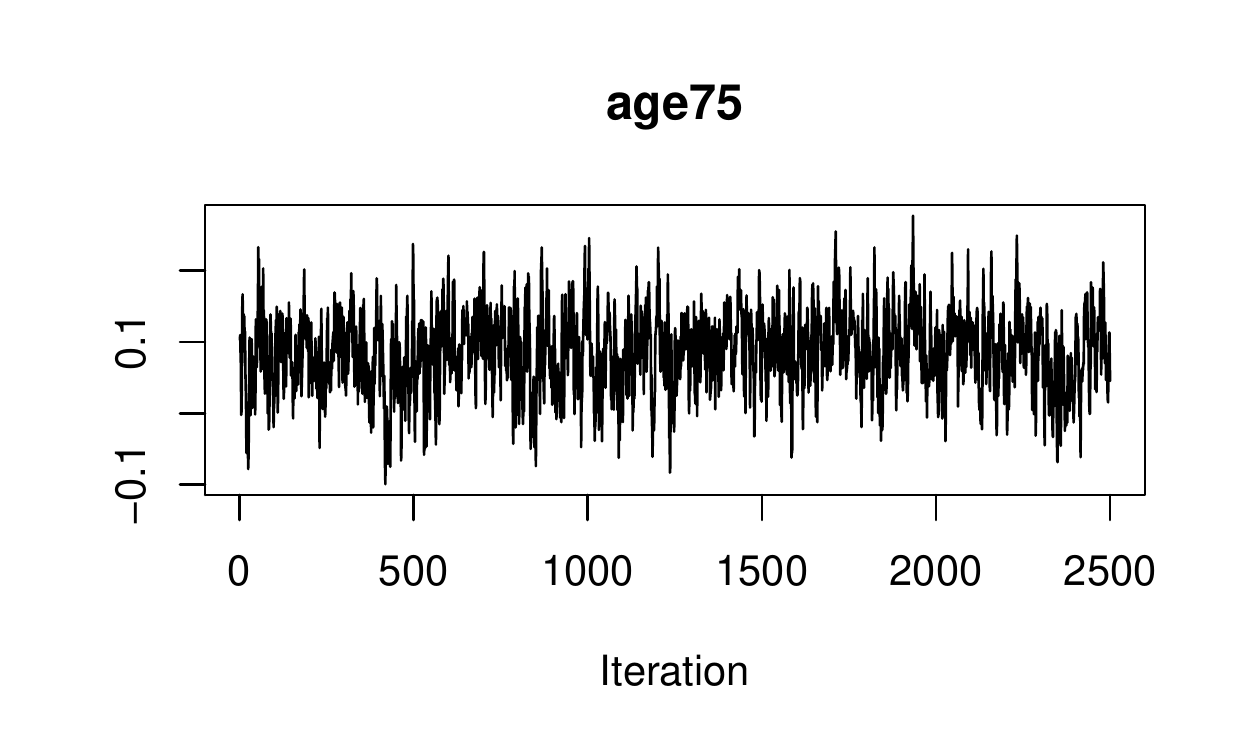}
  \label{fig:azpro-trace-Beta4}
\end{subfigure}
\begin{subfigure}[t]{0.32\textwidth}
  \centering
  \includegraphics[width=1.0\textwidth]{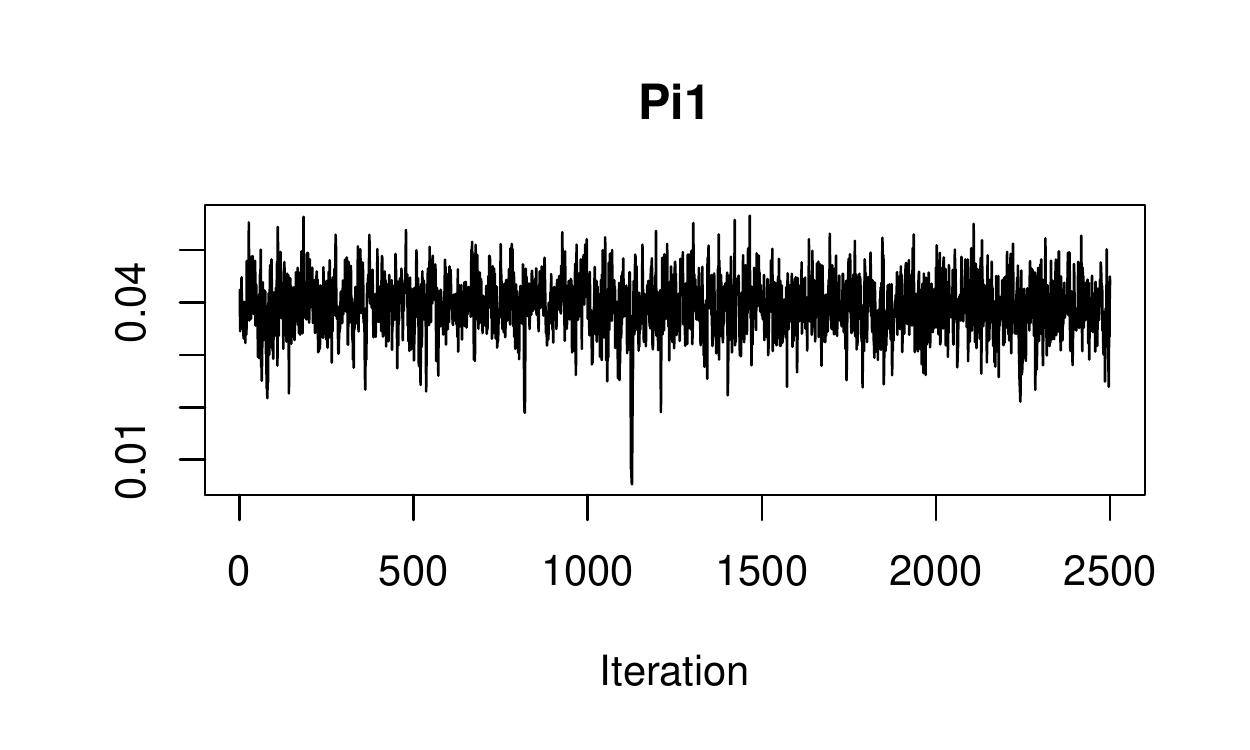}
  \label{fig:azpro-trace-Pi1}
\end{subfigure}
\begin{subfigure}[t]{0.32\textwidth}
  \centering
  \includegraphics[width=1.0\textwidth]{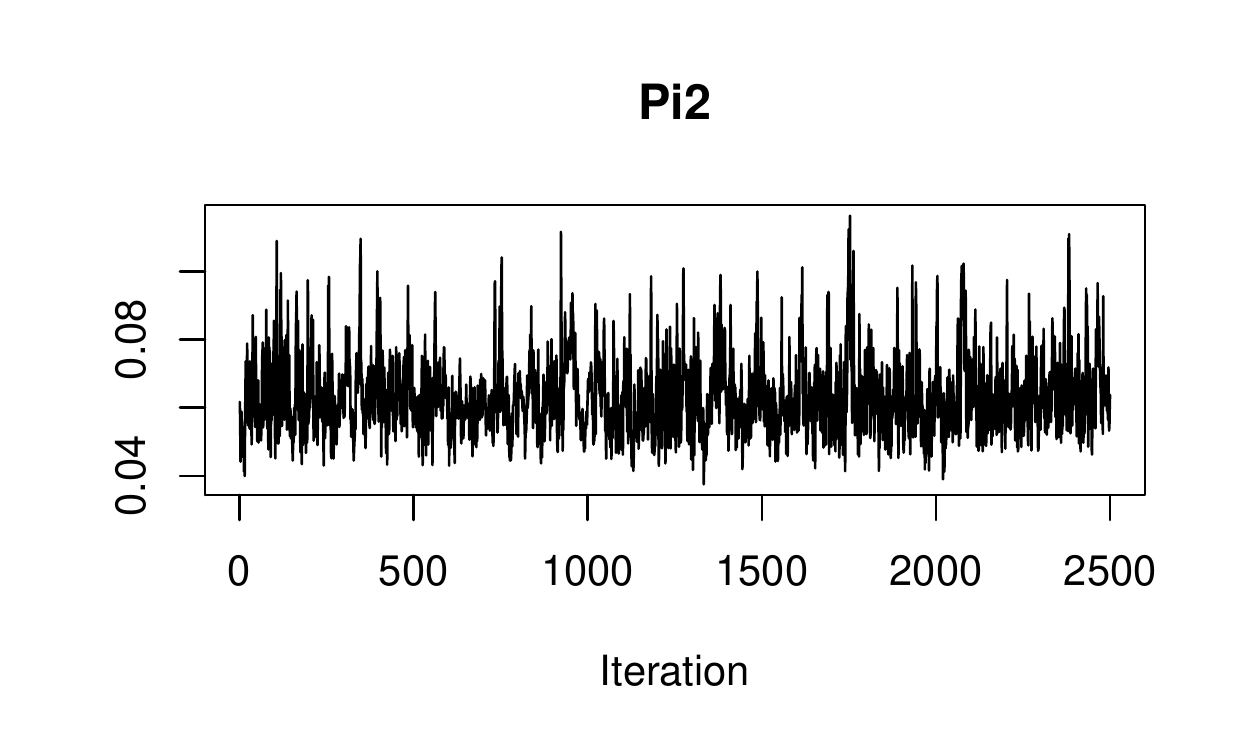}
  \label{fig:azpro-trace-Pi2}
\end{subfigure}
\begin{subfigure}[t]{0.32\textwidth}
  \centering
  \includegraphics[width=1.0\textwidth]{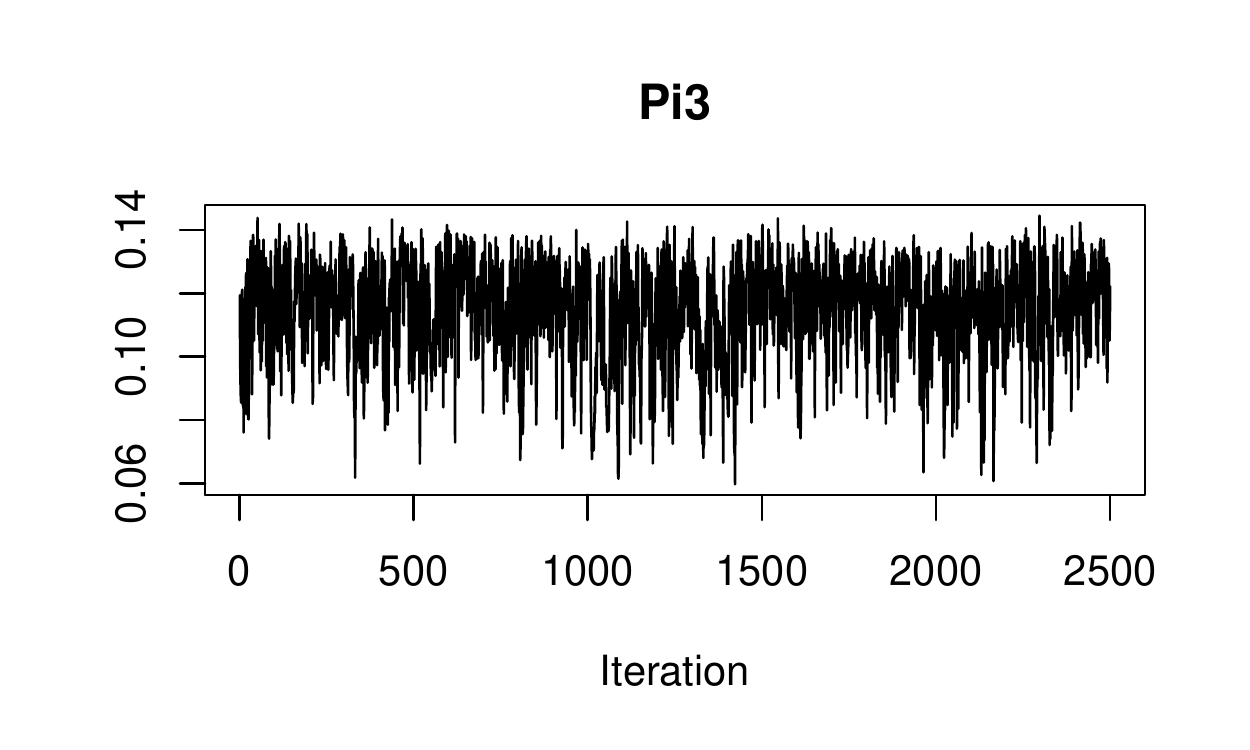}
  \label{fig:azpro-trace-Pi3}
\end{subfigure}
\begin{subfigure}[t]{0.32\textwidth}
  \centering
  \includegraphics[width=1.0\textwidth]{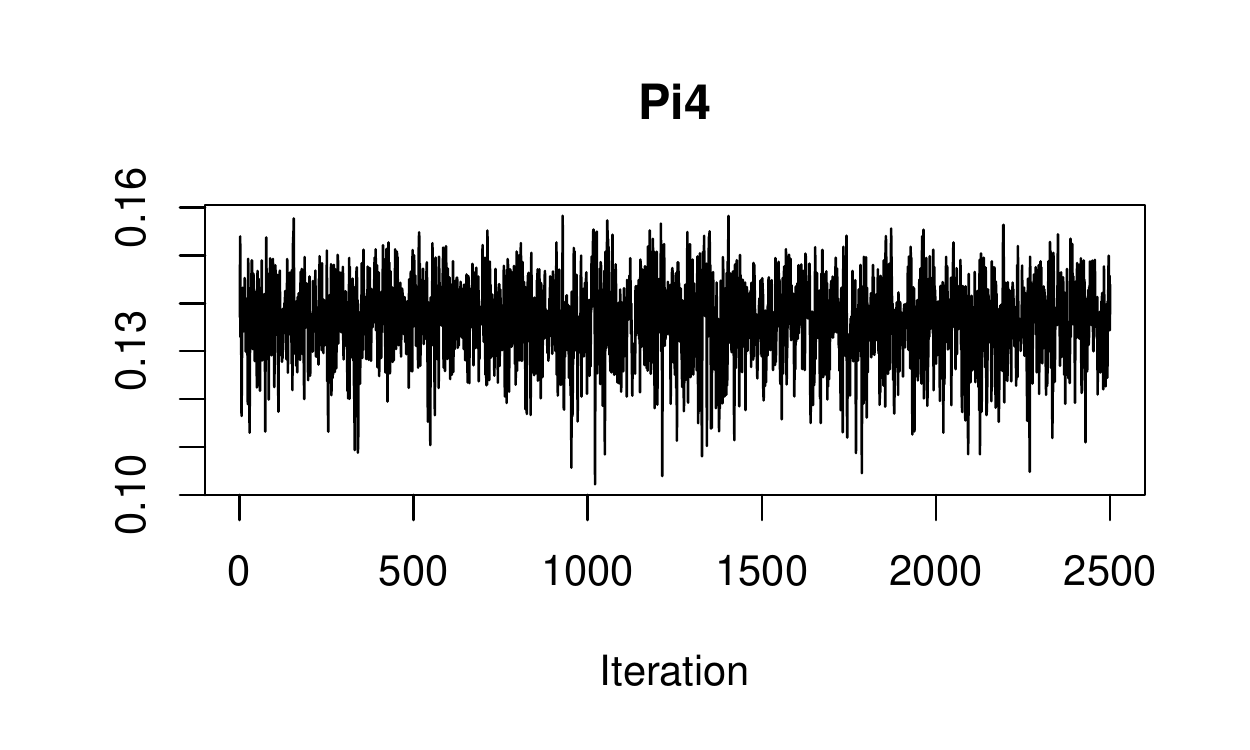}
  \label{fig:azpro-trace-Pi4}
\end{subfigure}
\begin{subfigure}[t]{0.32\textwidth}
  \centering
  \includegraphics[width=1.0\textwidth]{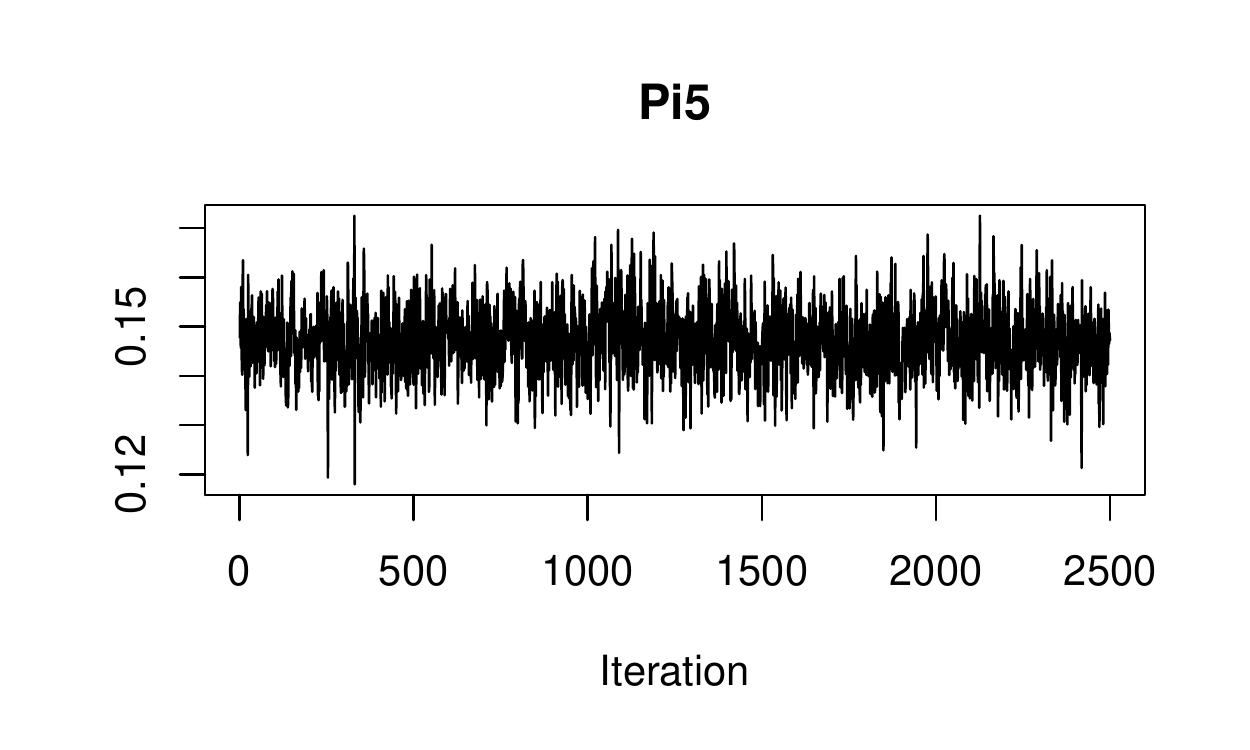}
  \label{fig:azpro-trace-Pi5}
\end{subfigure}
\begin{subfigure}[t]{0.32\textwidth}
  \centering
  \includegraphics[width=1.0\textwidth]{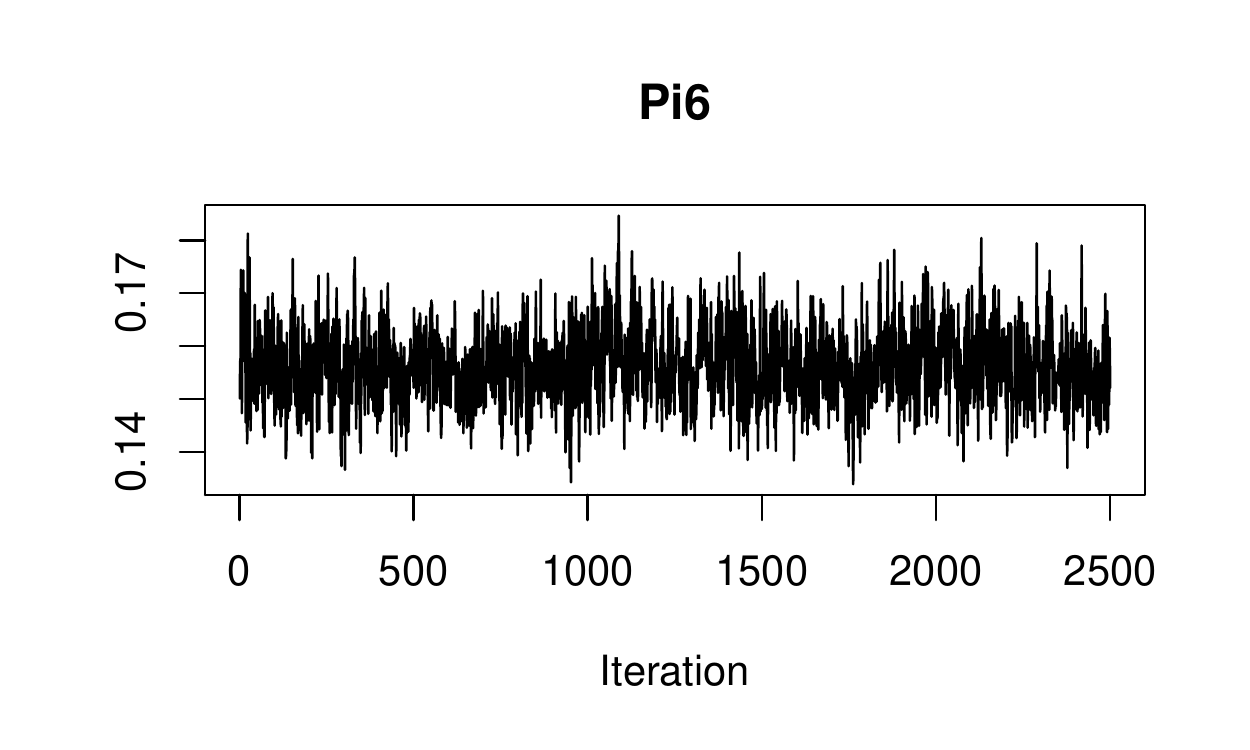}
  \label{fig:azpro-trace-Pi6}
\end{subfigure}
\begin{subfigure}[t]{0.32\textwidth}
  \centering
  \includegraphics[width=1.0\textwidth]{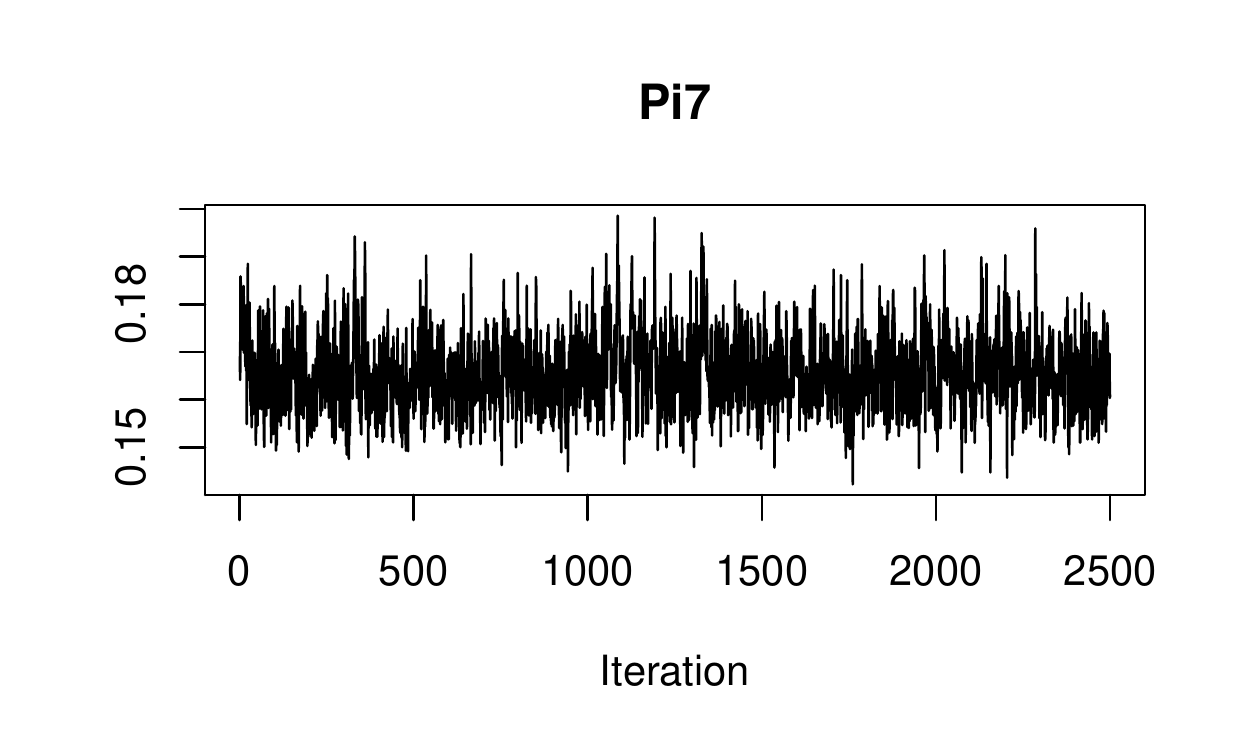}
  \label{fig:azpro-trace-Pi7}
\end{subfigure}
\begin{subfigure}[t]{0.32\textwidth}
  \centering
  \includegraphics[width=1.0\textwidth]{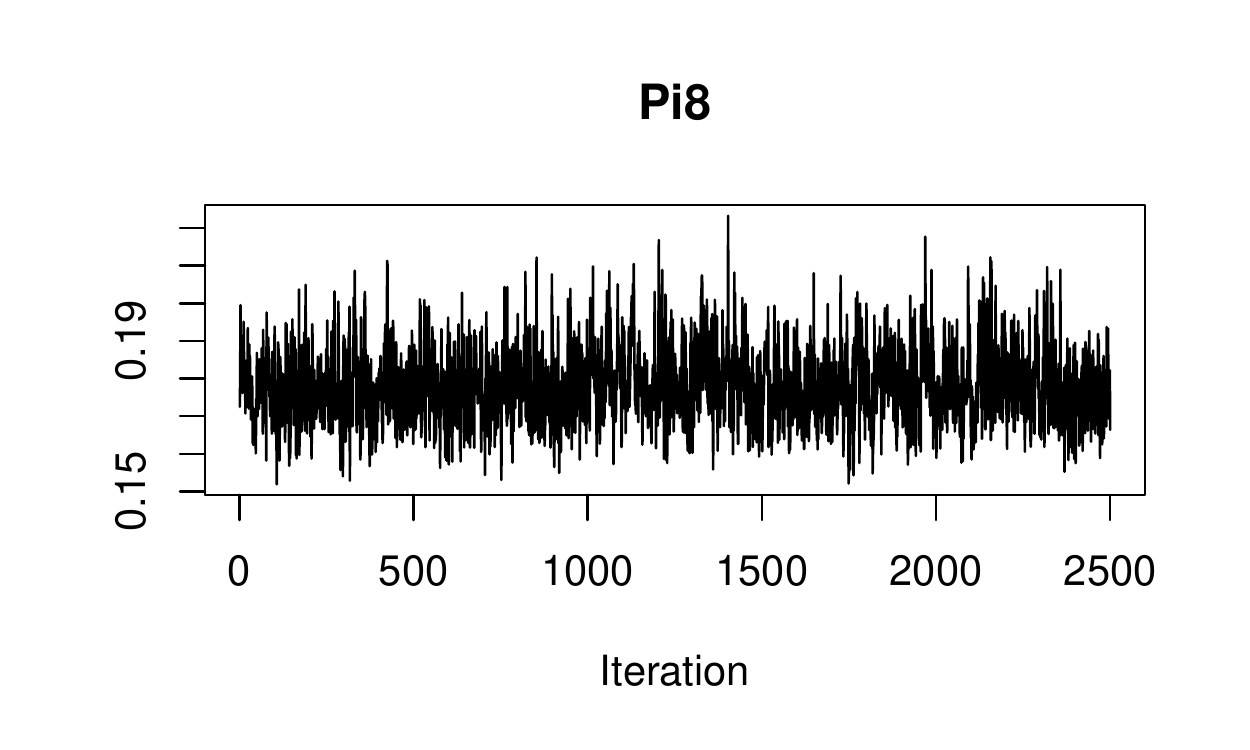}
  \label{fig:azpro-trace-Pi8}
\end{subfigure}
\begin{subfigure}[t]{0.32\textwidth}
  \centering
  \includegraphics[width=1.0\textwidth]{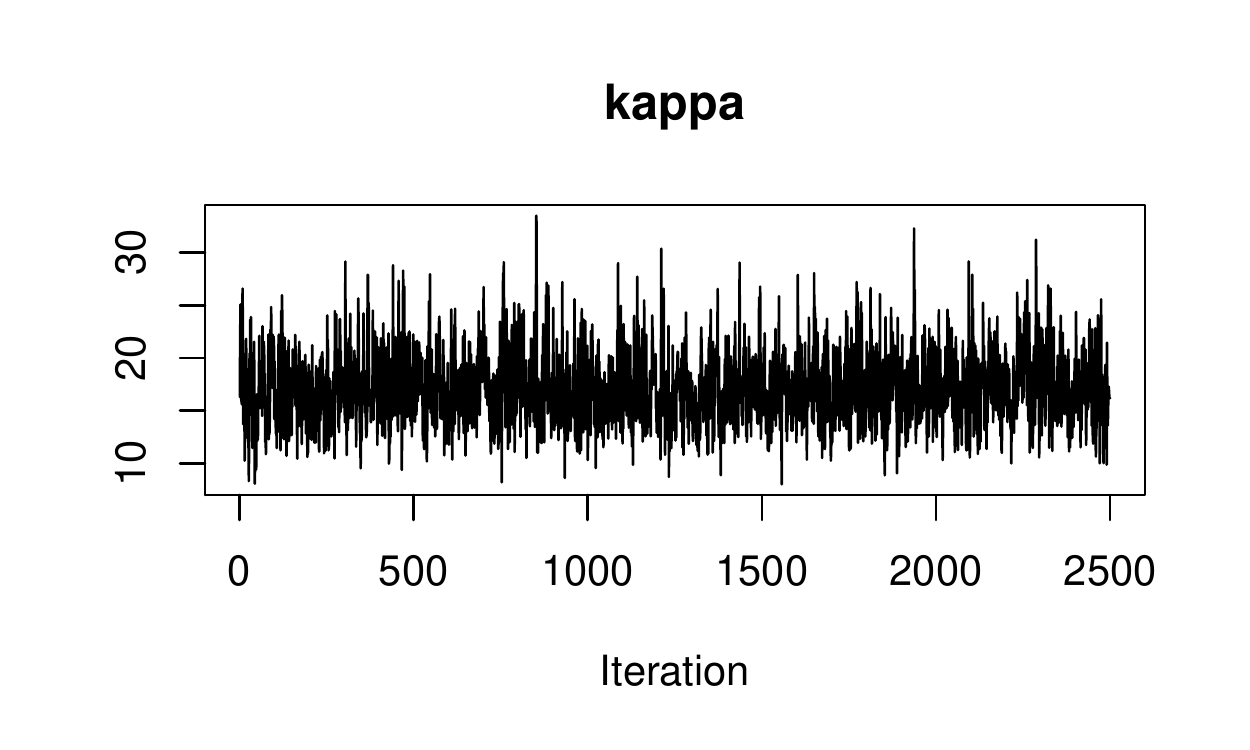}
  \label{fig:azpro-trace-kappa}
\end{subfigure}
\caption{Trace plots for MixLinkJ8 model fit to Arizona Medpar dataset.}
\label{fig:azpro-trace}
\end{figure}

\section{Conclusions}
\label{sec:conclusions}
Regression on the mean is commonly carried out with exponential family distributions in the Generalized Linear Model framework, but extending this idea to finite mixture distributions is not completely straightforward. This paper formulated the Mixture Link distribution, which establishes a link from a finite mixture mean to the regression function by assuming a random effects structure on the constrained parameter space. Specific variants of Mixture Link were obtained for Binomial, Poisson, and Normal outcomes. Integrals in the general Binomial case appeared not to have a tractable form, but the Normal case could be integrated to yield another (constrained) Normal finite mixture, and integrals in the Poisson case were evaluated using the confluent hypergeometric function. Some interesting connections were noted, for example, between Mixture Link Binomial and the Random-Clumped Binomial and Beta-Binomial distributions. Example regression analyses using Mixture Link Binomial and Poisson models demonstrated utility in handling overdispersion. Simpler models could adequately estimate the regression, yet failed to capture variability seen in the data. This became especially apparent in portions of analysis that depend heavily on the model, such as diagnosing model fit with quantile residuals or computing prediction intervals from the posterior predictive distribution. The fact that Mixture Link is completely likelihood-based ensures that such procedures are available; this could be seen as an advantage over quasi-likelihood methods when a flexible mean-variance relationship is needed. \code{R} code for the Mixture Link model is available in the \code{mixlink} package, available at \url{http://cran.r-project.org}.%
\footnote{The package currently provides Mixture Link Binomial and Poisson distributions and MCMC samplers. Functions to compute maximum likelihood estimates using numerical optimization are also implemented.}

The Mixture Link approach leads to a novel class of distributions with an interesting set of challenges for practical use in data analysis. Initial results in \citet{RaimPhDThesis2014}, \citet{MixLinkJSM2015}, and the present paper appear promising, especially using Bayesian inference, but more work is needed to determine the suitability of Mixture Link for wider application. In particular, it may be worthwhile to investigate analytical properties of Mixture Link models, such as differentiability, especially in the Binomial case. Such properties may be needed to establish appropriate methods for maximum likelihood estimation, large sample properties of maximum likelihood estimates, and approximation of the posterior distribution by a Normal distribution.

\section*{Acknowledgements}
We thank Professors Thomas Mathew, Yi Huang, and Yaakov Malinovsky at the University of Maryland, Baltimore County (UMBC) for serving on the committee of the dissertation in which this work was initiated. We thank the UMBC High Performance Computing Facility for use of its computational resources, and for financial support of the first author through a multiple year graduate assistantship.

\appendix

\section{Appendix: MCMC for Binomial Mixture Link}
\label{sec:appendix-mcmc-binomial}
An MCMC algorithm based on model \eqref{eqn:binomial-mixlink} can be formulated with $\psi_{ij}$ as augmented data. This approach avoids expensive numerical integration needed to compute the likelihood. The joint distribution of all random quantities is
\begin{align*}
&f(\vec{y}, \vec{\psi}, \bbeta, \vec{\pi}, \kappa) = \left\{ \prod_{i=1}^n Q(y_i, \vec{\psi}_i, \bbeta, \vec{\pi}, \kappa) \right\} f(\bbeta) f(\vec{\pi}) f(\kappa), \\
&\qquad \text{where} \quad Q(y_i, \vec{\psi}_i, \bbeta, \vec{\pi}, \kappa) =
\sum_{j=1}^J \pi_j \text{Bin}(y_i \mid m_i, H_{ij}(\psi_{ij})) \mathcal{B}(\psi_{ij} \mid a_{ij}, b_{ij}),
\end{align*}
and $H_{ij}(x) = (u_{ij} - \ell_{ij}) x + \ell_{ij}$. Gibbs steps to sample $\bbeta$, $\vec{\pi}$, $\kappa$, and $\vec{\Psi} = \{ \vec{\psi} _i: i = 1, \ldots, n \}$ will not yield closed forms. Instead, we will use simple Random Walk Metropolis Hastings \citep[Section 7.5]{RobertCasella2010} to propose draws for each random quantity.

To obtain draws of the constrained parameters $\vec{\pi}$, $\kappa$, and $\vec{\Psi}$, we draw unconstrained random variables from the sampler and transform them to the constrained space. Generally, denote $\vec{\xi}$ as one of the constrained parameters whose full conditional density is $f(\vec{\xi} \mid \text{Rest})$, and let $h$ be a bijection from the space of $\vec{\xi}$ to a Euclidean space $\mathbb{R}^k$. The density of $\vec{\phi} = h(\vec{\xi})$ is then $f(h^{-1}(\vec{\phi}) \mid \text{Rest}) | \det \mathfrak{J}(\vec{\phi})|$, where $\mathfrak{J}(\vec{\phi}) = \partial \vec{\xi} / \partial \vec{\phi}$. Starting from a given $\vec{\phi} = h(\vec{\xi})$, a proposed $\vec{\phi}^*$ will be accepted with probability
\begin{align*}
\min\left\{ 1, 
\frac{
  f(h^{-1}(\vec{\phi}^*) \mid \text{Rest}) \cdot |\det \mathfrak{J}(\vec{\phi}^*)|
}{
  f(h^{-1}(\vec{\phi}) \mid \text{Rest}) \cdot |\det \mathfrak{J}(\vec{\phi})|
} \right\}.
\end{align*}

Note that the function $Q(y_i, \vec{\psi}_i, \bbeta, \vec{\pi}, \kappa)$ needs to be evaluated in each step. By computing $Q$ in \code{C/C++}, it is possible to improve the performance greatly over a pure \code{R} \citep{R2015} implementation of our sampler. The \code{Rcpp} package by \citet{EddelbuettelFrancois2011}, for example, greatly facilitates a hybrid implementation of \code{R} and \code{C++}.

\paragraph{Gibbs step for $\bbeta$.} Consider the unnormalized density
\begin{align*}
q(\bbeta \mid \text{Rest}) = \left\{ \prod_{i=1}^n Q(y_i, \vec{\psi}_i, \bbeta, \vec{\pi}, \kappa) \right\} f(\bbeta).
\end{align*}
Suppose $\bbeta^{(r)}$ is the current iterate of $\bbeta$ in the simulation and draw $\bbeta^*$ from the proposal distribution $N(\bbeta^{(r)}, \vec{V}_\beta^\text{prop})$. Draw $U \sim \mathcal{U}(0,1)$, and let
\begin{align*}
\bbeta^{(r+1)} = \begin{dcases}
\bbeta^* & \text{if $U < \frac{q(\bbeta^* \mid \text{Rest}) }{ q(\bbeta^{(r)} \mid \text{Rest}) }$} \\
\bbeta^{(r)} & \text{otherwise}.
\end{dcases}
\end{align*}

\paragraph{Gibbs step for $\vec{\pi}$.} Consider the unnormalized density
\begin{align*}
q(\vec{\pi} \mid \text{Rest}) = \left\{ \prod_{i=1}^n Q(y_i, \vec{\psi}_i, \bbeta, \vec{\pi}, \kappa) \right\} f(\vec{\pi}).
\end{align*}
Suppose $\vec{\pi}^{(r)}$ is the current iterate of $\vec{\pi}$ in the simulation. Denote $\mathbb{S}^J$ as the probability simplex in dimension $J$ with typical element $\vec{p} = (p_1, \ldots, p_J)$. Note that the multinomial logit function $h(\vec{p}) = (\log(p_1 / p_J), \ldots, \log(p_{j-1} / p_J))$ is a bijection from $\mathbb{S}^J$ to $\mathbb{R}^{J-1}$. Therefore, we can draw $\vec{\phi}^*$ from the proposal distribution $N(h(\vec{\pi}^{(r)}), \vec{V}_\pi^\text{prop})$ on $\mathbb{R}^{J-1}$ and let $\vec{\pi}^* = h^{-1}(\vec{\phi}^*)$ be the candidate for the next iterate. Denote $\mathfrak{J}(\vec{\phi}) = \frac{\partial \vec{\pi} }{ \partial \vec{\phi} }$ as the $J \times (J-1)$ Jacobian of the transformation from $\vec{\phi}$ to $\vec{\pi}$, and let $\det \mathfrak{J}(\vec{\phi})$ be the determinant ignoring the $J$th row. Draw $U \sim \mathcal{U}(0,1)$, and let
\begin{align*}
\vec{\pi}^{(r+1)} = \begin{dcases}
\vec{\pi}^* & \text{if $U < \frac{q(\vec{\pi}^* \mid \text{Rest}) }{ q(\vec{\pi}^{(r)} \mid \text{Rest}) } \frac{|\det \mathfrak{J}(\vec{\phi}^*)|}{|\det \mathfrak{J}(\vec{\phi}^{(r)})|} $} \\
\vec{\pi}^{(r)} & \text{otherwise}.
\end{dcases}
\end{align*}

\paragraph{Gibbs step for $\kappa$.} Consider the unnormalized density
\begin{align*}
q(\kappa \mid \text{Rest}) = \left\{ \prod_{i=1}^n Q(y_i, \vec{\psi}_i, \bbeta, \vec{\pi}, \kappa) \right\} f(\kappa).
\end{align*}
Suppose $\kappa^{(r)}$ is the current iterate of $\kappa$ in the simulation. Draw $\phi^*$ from the proposal distribution $N(\log(\kappa^{(r)}), V_\kappa^\text{prop})$ and let $\kappa^* = \exp(\phi^*)$ be the candidate for the next iterate. The Jacobian of the transformation from $\phi$ to $\kappa$ is $\frac{\partial \kappa}{\partial \phi} = \exp(\phi)$. Draw $U \sim \mathcal{U}(0,1)$, and let
\begin{align*}
\kappa^{(r+1)} = \begin{dcases}
\kappa^* & \text{if $U < \frac{q(\kappa^* \mid \text{Rest}) }{ q(\kappa^{(r)} \mid \text{Rest}) } \frac{\exp(\phi^*)}{\exp(\phi^{(r)})} $} \\
\kappa^{(r)} & \text{otherwise}.
\end{dcases}
\end{align*}

\paragraph{Gibbs step for $\vec{\psi}$.} Consider the unnormalized density
\begin{align*}
q(\vec{\psi} \mid \text{Rest}) = \prod_{i=1}^n Q(y_i, \vec{\psi}_i, \bbeta, \vec{\pi}, \kappa).
\end{align*}
We can see that $\vec{\psi}_i$ are independent conditional on the remaining random variables and we may therefore consider drawing one at a time. Suppose $\vec{\psi}_i^{(r)}$ is the current iterate of $\vec{\psi}_i$ in the simulation. Let $G$ be the CDF of the logistic distribution, which is a bijection from $\mathbb{R}$ to the unit interval. Denote $\vec{\phi}^{(r)} = (G^{-1}(\psi_{i1}^{(r)}), \ldots, G^{-1}(\psi_{iJ}^{(r)})$. The Jacobian of the transformation from $\vec{\phi}$ to $\vec{\psi}_i$ is
\begin{align*}
\frac{\partial \vec{\psi}_i}{\partial \vec{\phi}} = \Diag(G'(\phi_1), \ldots, G'(\phi_J))
\quad \implies \quad
\det \left(\frac{\partial \vec{\psi}_i}{\partial \vec{\phi}} \right) = \prod_{j=1}^J G'(\phi_j),
\end{align*}
where $G'$ represents the logistic density. Draw $\vec{\phi}^*$ from the proposal distribution $N(\vec{\phi}^{(r)}, V_\phi^\text{prop})$ and let $\vec{\psi}_i^* = (G(\phi_1^*), \ldots, G(\phi_J^*))$ be the candidate for the next iterate. Draw $U \sim \mathcal{U}(0,1)$, and let
\begin{align*}
&\vec{\psi}_i^{(r+1)} = \begin{dcases}
\vec{\psi}_i^* & \text{if
$U < \frac{q(\vec{\psi}_i^* \mid \text{Rest}) }{ q(\vec{\psi}_i^{(r)} \mid \text{Rest}) }
\frac{ \prod_{j=1}^J G'(\phi_j^*) }{ \prod_{j=1}^J G'(\phi_j^{(r)}) }$}, \\
\vec{\psi}_i^{(r)} & \text{otherwise}.
\end{dcases}
\end{align*}


\end{document}